\documentclass[11pt, a4paper]{article}
\usepackage{hyperref}
\usepackage{amsmath,amssymb,amsthm}

\usepackage{graphicx,epsfig,epstopdf}
\usepackage{url,mathtools}
\usepackage{colortbl,framed}
\usepackage{lscape}
\usepackage{pdflscape}
\usepackage{rotating}
\usepackage{multirow}
\usepackage{array}
\usepackage{float}
\usepackage[authoryear]{natbib}
\bibpunct{(}{)}{;}{a}{,}{,}
\usepackage[english]{babel}
\usepackage[latin1]{inputenc}
\usepackage[T1]{fontenc}
\usepackage{lmodern}
\usepackage{amssymb}
\usepackage{xcolor,graphicx,xspace,caption}
\usepackage{listings,eso-pic,subcaption}
\usepackage{lscape}
\usepackage{multirow}
\usepackage{enumerate}
\usepackage{amsmath}
\usepackage{longtable}
\usepackage{bbm}
\usepackage{bm}
\usepackage{amsthm}
\usepackage{setspace}
\usepackage{appendix}

\newcommand{\quantnet}{\raisebox{-1pt}{\includegraphics[scale=0.05]{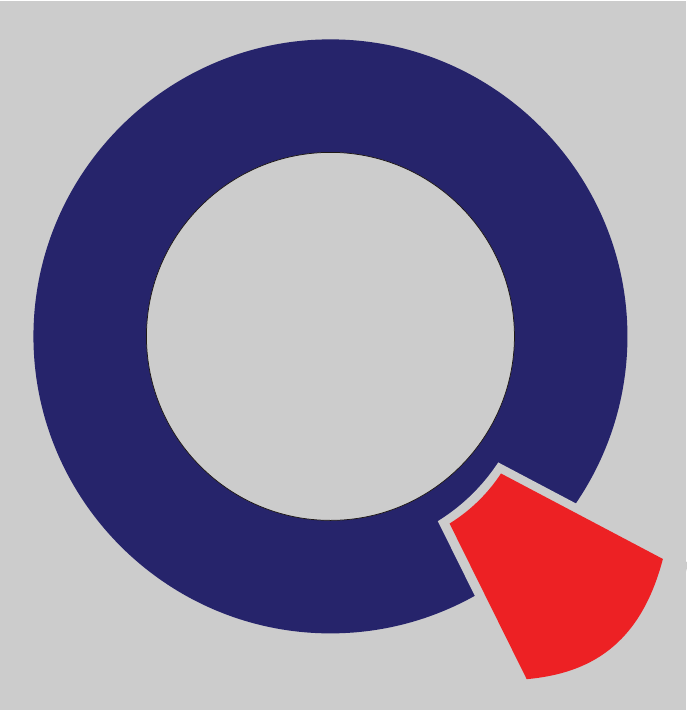}}\,}

\def\oc{\overline{c}}

\usepackage[margin=1in]{geometry}

\renewcommand{\vec}{\mbox{vec}}
\newcommand{\Z}{\mathbb{Z}}
\renewcommand{\P}{\operatorname{P}}
\newcommand{\R}{\rm I\!R}

\newcommand{\N}{\operatorname{N}}

\newcommand{\E}{\mathop{\mbox{\sf E}}}
\newcommand{\IF}{\boldsymbol{1}}

\newcommand{\vech}{\mbox{vech}}
\newcommand{\mA}{\mathcal{A}}

\newcommand{\mP}{\mathcal{P}}

\newcommand{\mF}{\mathcal{F}}

\newcommand{\mH}{\mathcal{H}}

\newcommand{\mQ}{\mathcal{Q}}
\renewcommand{\mP}{\mathcal{P}}
\renewcommand{\mF}{\mathcal{F}}

\newcommand{\supp}{\mathrm{supp}}

\newcommand{\vps}{{\varepsilon}}

\renewenvironment{proof}[1][\proofname]{{\noindent\bfseries #1.}}{\qed}

\newcommand{\bigO}{\mathcal{O}}
\newcommand{\smallO}{\mbox{\tiny $\mathcal{O}$}}

\newtheorem{theorem}{Theorem}[section]
\newtheorem{corollary}{Corollary}[section]
\newtheorem{lemma}{Lemma}[section]

\newtheorem{definition}{Definition}[section]

\theoremstyle{definition}

\newtheorem{remark}{Comment}[section]
\numberwithin{remark}{section}
\newtheorem{example}{Example}
\numberwithin{equation}{section}
\numberwithin{theorem}{section}
\numberwithin{figure}{section}
\numberwithin{table}{section}

\newcommand{\eps}{\varepsilon}

\newcommand{\beq}{\begin{equation}}
\newcommand{\eeq}{\end{equation}}
 \usepackage{xr}
\usepackage{url}
\newcommand{\F}{\mathcal{F}}

\renewcommand{\hat}{\widehat}
\newcommand{\En}{\frac{1}{n} \sum_{t=1}^n}

\renewcommand{\hat}{\widehat}
\renewcommand{\leq}{\leqslant}
\renewcommand{\geq}{\geqslant}

\DeclareMathOperator{\Var}{Var}

\newcommand{\diag}{{\rm diag}}

\renewcommand{\[}{\left[}
\renewcommand{\]}{\right]}

\def\defeq{\stackrel{\mathrm{def}}{=}}  

\begin{document}

%


\title{
LASSO-Driven Inference in Time and Space
\thanks{We thank Weibiao Wu, Oliver Linton, Bryan Graham, Manfred Deistler, Hashem Pesaran, Michael Wolf, Valentina Corradi, Zudi Lu, Liangjun Su, Peter Phillips, Frank Windmeijer, Wenyang Zhang and Likai Chen for helpful comments and suggestions. We remain responsible for any errors or omissions. Financial support from the Deutsche Forschungsgemeinschaft via IRTG 1792 ``High Dimensional Non Stationary Time Series'', Humboldt-Universit\"at zu Berlin, is gratefully acknowledged.}}
\author{Victor Chernozhukov\footnote{Department of Economics and Center for Statistics and Data Science, Massachusetts Institute of Technology.}, Wolfgang K.~H\"{a}rdle\footnote{IRTG1792, Humboldt-Universit\"{a}t zu Berlin. School of Business, Singapore Management University. Faculty of Mathematics and Physics, Charles University. Department of Information Management and Finance, National Chiao Tung University.}, Chen Huang\footnote{Department of Economics and Business Economics and CREATES, Aarhus University. Corresponding author: chen.huang@econ.au.dk}, Weining Wang\footnote{Department of Economics and Related Studies, University of York. Ladislaus von Bortkiewicz Chair of Statistics, Humboldt-Universit\"{a}t zu Berlin.}}
\maketitle
\begin{center}
\end{center}
\vspace{-1cm}

\begin{abstract}
\footnotesize{\noindent}
 We consider the estimation and inference in a system of high-dimensional regression equations allowing for temporal and cross-sectional dependency in covariates and error processes, covering rather general forms of weak temporal dependence. {A sequence of regressions with many regressors using LASSO (Least Absolute Shrinkage and Selection Operator) is applied for variable selection purpose}, and an overall penalty level is carefully chosen by a block multiplier bootstrap procedure to account for multiplicity of the equations and dependencies in the data. Correspondingly, oracle properties with a jointly selected tuning parameter are derived. We further provide high-quality de-biased simultaneous inference on the many target parameters of the system. We provide bootstrap consistency results of the test procedure, which are based on a general Bahadur representation for the $Z$-estimators with dependent data. Simulations demonstrate good performance of the proposed inference procedure.
Finally, we apply the method to quantify spillover effects of textual sentiment indices in a financial market and to test the connectedness among sectors.
\par
\vspace{0.5cm}
\noindent {\em JEL classification}: C12, C22, C51, C53\\
\noindent {\em Keywords}: LASSO, time series, simultaneous inference, system of equations, $Z$-estimation, Bahadur representation, martingale decomposition
\end{abstract}

\section{Introduction}\label{intro}
Many applications in  statistics, economics, finance, biology and psychology are concerned with a system of ultra high-dimensional objects that communicate within complex dependency channels. Given a complex system involving many factors, one builds a network model by taking a large set of regressions, i.e. regressing every factor in the system on a large subset of other factors. Examples include analysis of financial systemic risk by quantile predictive graphical models with LASSO \citep{hautsch2014financial,hardle2016tenet,BCC2016}, limit order book network modeling via the penalized vector autoregressive approach \citep{chen_2018}, analysis of psychology data with temporal and cross- sectional dependencies \citep{epskamp2018gaussian}. Another example is quantifying the spillover effects or externalities for a social network, especially when the social interactions (or the interconnectedness) is not obvious \citep{manresa2013estimating}. Besides, there are numerous applications concerning association network analysis in other fields of applied statistics; see Chapter 7 in \cite{kolaczyk2014statistical}. 
In general, a step-by-step LASSO procedure is very helpful for the correlation network formation. In pursuing a highly structural approach, one certainly favors a simple set of regressions that allows multiple insights on the statistical structure of the data. Therefore, a sequence of regressions with LASSO is a natural path to take. Especially in cases of reduced forms of simultaneous equation models and structural vector autoregressive models, one can attain valuable pre-information on the core structure by running a set of simple regressions with LASSO shrinkage.

A first important question arising in this framework is how to decide on a unified level of penalty. In this article we advocate an approach to selecting the overall level of the tuning parameter in a system of equations after performing a set of single step regressions with shrinkage. A feasible (block) bootstrap procedure is developed and the consistency of parameter estimation is studied. In addition, we provide a uniform near-oracle bound for the joint estimators. The proposed technique is applicable to ultra-high dimensional systems of regression equations with high-dimensional regressors.

A second crucial issue is to establish simultaneous inference on parameters, which is an important question regarding network topology inference.
For example, in a large-scale linear factor pricing model,
it is of great interest to check the significance of the intercepts of cross sectional regressions (connected with zero pricing errors), e.g. \cite{pesaran2017testing}. Our approach is an alternative testing solution compared to the Wald test statistics proposed therein. To achieve the goal of simultaneous inference, we develop a uniform robust post-selection or post-regularization inference procedure for time series data.
This method is generated from a uniform Bahadur representation of de-biased instrumental variable estimators. In particular, we need to establish maximal inequalities for empirical processes for a general Huber's $Z$-estimation. Note that the commonly used technique for independent data, such as the symmetrization technique, is not directly applicable in the dependent data case; see Chapter 11.6 of \cite{kosorok2008introduction} for a related overview.


Our contribution lies in three aspects. First, we select the penalty level by controlling the aggregated errors in a system of high-dimensional sparse regressions, and we establish the bounds on the estimated coefficients. Furthermore, we show the implication of the restricted eigenvalue (RE) condition at a population level. Secondly, an easily implemented algorithm for effective estimation and inference is proposed. In fact, the offered estimation scheme allows us to make local and global inference on any set of parameters of interest. Thirdly, we run numerical experiments to illustrate good performance of our joint penalty relative to the single equation estimation, and we show the finite sample improvement of our multiplier block bootstrap procedure on the parameter inference. Finally, an application of textual sentiment spillover effects on the stock returns in a financial market is presented.

In the literature, the fundamental results on achieving near oracle rate for penalized $\ell_1$-norm estimators are developed by \citet{bickel2009simultaneous}. There are many related articles on deriving near-oracle bounds using the $\ell_1$-norm penalization function for the i.i.d. case, such as \citet{BCH2011,belloni2009least}. There are also many extensions to the LASSO estimation with dependent data. For example, {\citet{basu2015regularized} study the consistency of the estimator in sparse high-dimensional Gaussian time series models; \citet{kock2015oracle} consider the high-dimensional near-oracle inequalities in large vector autoregressive (VAR) models; \citet{lin2017regularized} look at the regularized estimation and testing for high-dimensional multi-block VAR models.} However, the majority of the literature imposes a {Gaussian} or sub-Gaussian assumption on the error distribution; this is rather restrictive and excludes heavy tail distributions. For dependent data, \cite{wu2016performance} discuss the possibility of relaxing the sub-Gaussian assumption by generalizing Nagaev-type inequalities allowing for only moment assumptions. For the case of LASSO the analysis assumes the fixed design, which rules out the most important applications mentioned earlier in the introduction.

Theoretically, the LASSO tuning parameter selection requires characterizing the asymptotic distribution of the maximum of a high dimensional random vector. \citet{CCK13AoS} develop a Gaussian approximation for the maximum of a sum of high-dimensional random vectors, which is in fact the basic tool for modern high-dimensional estimation. Here it is applied to the LASSO inference. Moreover, \citet{CCK13testing} deliver results for the case of $\beta$-mixing processes.
Although it is quite common to assume a mixing condition {which is at base a concept yielding asymptotic independence}, it is not in general easy to verify the condition for a particular process, and some simple linear processes can be excluded from the strong mixing class, \cite{andrews1984non}. With an easily accessible dependency concept, 
\citet{ZW15gaussian} derive Gaussian approximation results for a wide class of stationary processes. Note that the dependence measure is linked to martingale decompositions and is therefore readily connected with a pool of results on tail probabilities, moment inequalities and central limit theorems of martingale theory. Our results are built on the above-mentioned theoretical works and we extend them substantially to fit into the estimation in a system of regression equations. In particular, our LASSO estimation is with random design for dependent data; therefore, we need to deal with the population implications of the Restricted Eigenvalue (RE) condition. {Moreover, we show the interaction between the tail assumption and the dimensionality of the covariates in our theoretical results.}

In the meantime, the issue of simultaneous inference is challenging and has motivated a series of research articles. For the case of i.i.d. data, \cite{BCH2011,BCH2014}, \citet{zhang2014debiased}, \citet{JM2014}, \citet{van2014desparsified}, \citet{neykov2015unified}, \citet{DML}, \citet{zhu2017linear}, among others, develop confidence intervals of low-dimensional variables in high-dimensional models with various forms of de-biased/orthogonalization methods. Still in the case of i.i.d. data, \cite{BCK15Bio} establish a uniform post-selection inference for the target parameters defined via de-biased Huber's $Z$-estimators when the dimension of the parameters of interest is potentially larger than the sample size, where they employ the multiplier bootstrap to the estimated residuals. Wild and residual bootstrap-assisted approaches are also studied in \citet{dezeure2016high,zhang2017simultaneous} for the case of mean regression. {And more recently, \citet{krampe2018} extend the approaches to test large groups of coefficients in sparse VAR models.}
We pick up the line of the inference analysis of \citet{BCK15Bio} and employ it in a temporal and cross-sectional dependence framework, thus making it applicable to a rich class of high-dimensional time series. {This allows us to embed the high-dimensional VAR model as a special case.} Our core proof strategy is different, as it is well known that the technique for handling the suprema of empirical processes indexed by functional classes with dependent data is not the same as in i.i.d. cases. For instance, the key Bahadur representation in \citet{BCK15Bio} applies maximal inequalities derived in \citet{chernozhukov2014gaussian} for i.i.d. random variables, while we derive the key concentration inequalities based on a martingale approximation method.

{Our proposed estimation framework is complement to the literature on model selection for Gaussian Graphical model (GGM) (see e.g. \cite{yuan2007model}), which has a wide spectrum of applications in statistics. A GGM can be connected with LASSO regression for estimating sparse correlation networks, and therefore is equivalent to our context with a partial correlation network, \cite{meinshausen2006high}. In particular, we may find an equation-by-equation relationship to the GGM, and we acknowledge that a similar framework with spatial temporal dependence can be developed. In addition, there is a big literature on social network analysis, which embeds the network information into a dynamic model in advance; see for example \cite{zhu2017network,zhu2016network,chen2019tail,huang2016statistical}. Relatively, our approach is less structural as we treat the network structure to be unknown and uncover it using LASSO.  }

The following notations are adopted throughout this paper. For a vector $v=(v_1,\ldots,v_p)^\top$, let $|v|_\infty\defeq\max_{1\leq j\leq p}|v_j|$ and $|v|_s\defeq(\sum_{j=1}^p|v_j|^s)^{1/s}$, $s\geq1$. For a random variable $X$, let $\|X\|_q\defeq(\E|X|^q)^{1/q}$, $q>0$. For any function on a measurable space $g:\mathcal{W}\rightarrow\R$, $\E_n(g)\defeq n^{-1}\sum_{t=1}^n\{g(\omega_t)\}$ and $G_n(g)\defeq n^{-1/2}\sum_{t=1}^n[g(\omega_t)-\E\{g(\omega_t)\}]$. Given two sequences of positive numbers $a_n$ and $b_n$, write $a_n\lesssim b_n$ if there exists constant $C>0$ (does not depend on $n$) such that $a_n/b_n\leq C$. For a sequence of random variables $x_n$, we use the notation $x_n\lesssim_{\P} b_n$ to denote $x_n=\bigO_{\P}(b_n)$. {For any {finitely discrete} measure $\mathcal Q$ on a measurable space, 
	let $\mathcal L^q(\mathcal Q)$ denote the space of all measurable functions $f:Z\rightarrow{\R}$ such that $\|f\|_{\mathcal Q,q}\defeq(\mathcal Q|f|^q)^{1/q}<\infty$, where $\mathcal Qf\defeq\int fd\mathcal Q$. For a class of measurable functions $\mathcal F$, the $\epsilon$-covering number with respect to the $\mathcal L^q(\mathcal Q)$-semimetric is denoted as $\mathcal N(\epsilon,\mathcal F, \|\cdot\|_{\mathcal Q,q})$, and let $\operatorname{ent}(\epsilon,\mathcal F)=\log\sup_{\mathcal Q}\mathcal N(\epsilon\|\bar F\|_{\mathcal Q,q},\mathcal F,\|\cdot\|_{\mathcal Q,q})$ with $\bar F=\sup_{f\in\mathcal F}|f|$ (the envelope) denote the uniform entropy number. It should be noted that we suppress the notation of the outer expectation $\E^*$ to $\E$ and outer probability $\P^*$ to $\P$ when measurability issues are encountered. Details may be found in the Chapter 1 of \cite{van1996weak}.}

The rest of the article is organized as follows. Section \ref{model} shows the system model with a few examples. Section \ref{lasso} introduces the sparsity method for effective prediction and provides an algorithm for the joint penalty level of LASSO via bootstrap. In Section \ref{inference} we propose approaches to implementing individual and simultaneous inference on the coefficients. Main theorems are listed in Section \ref{theorem}. In Section \ref{sim} and \ref{app} we deliver the simulation studies and an empirical application on textual sentiment spillover effects. The technical proofs and other details are given in the supplementary materials. The codes to implement the algorithms are publicly accessible via the website \href{https://github.com/QuantLet/LASSO_Time_Space}{\quantnet www.quantlet.de}.

\section{The System Model}\label{model}
In this section, we present a general framework which covers many applications in statistics.
Consider the system of regression equations (SRE):
$$
Y_{j,t} = X_{j,t}^\top\beta_{j}^0 + \vps_{j,t},   \quad \E \vps_{j,t} X_{j,t} = 0, \quad j=1,...,J, \quad t=1,\ldots, n,
$$
where $X_{j,t}=(X_{jk,t})_{k=1}^{K_j}$. 
{Without loss of generality, we assume the dimension of the covariates is identical among all equations thereafter, namely $K_j = \dim(X_{j,t})\equiv K$, for $j=1,\ldots,J$.} We allow the dimension $K$ of $X_{j,t}$ and the number of equations, $J$ to be large, potentially larger than $n$, which creates an interplay with the tail assumptions on the error processes $\vps_{j,t}$.
Both spatial and temporal dependency are allowed and we will obtain results on prediction and inference.

The SRE framework is a system of regression equations, which includes the following important special cases.

\begin{example}[Many Regression Models]
	Suppose that we are interested in estimating the predictive models for the response variables $U_{m,t}$:
	$$
	U_{m,t} = X_t^\top \gamma_{m}^0 + \varepsilon_{m,t}, \quad X_t\in{\R}^{K}, \quad \E \varepsilon_{m,t} X_{t} = 0, \quad m=1, \ldots, M,
	$$
	with auxiliary regressions to model predictive relations between covariates:
	$$
	X_{k,t} = X_{-k,t}^\top \delta_{k}^0 + \nu_{k,t},  \quad \E \nu_{k,t}  X_{-k,t} = 0, \quad k=1,\ldots,K, 
	$$
	where $X_{-k,t}=(X_{\ell,t})_{\ell\neq k}\in{\R}^{K-1}$, and $\delta_k^0$ is defined by the OLS estimator in population, namely $\arg\underset{\delta_k}{\min}\En\E(X_{k,t}-X_{-k,t}^\top \delta_{k})^2$. This is a special SRE model with
	$$
	(Y_{j,t}, X_{j,t}, \vps_{j,t}, \beta_{j}^0) = (U_{j, t}, X_{t}, \varepsilon_{j,t}, \gamma_{j}^0),\quad j = 1, \ldots, M,$$
	$$
	(Y_{j,t}, X_{j,t}, \vps_{j,t}, \beta_{j}^0) = (X_{(j-M),t}, X_{-(j-M),t}, \nu_{(j-M),t}, \delta_{(j-M)}^0), \quad j =M+1, \ldots,  J = M+ K.
	$$
	
	It can be seen that we only put contemporaneous exogeneity conditions for $X_t$. It is worth mentioning that this SRE case is closely related to the semiparametric estimation framework studied in Section 2.4 in \cite{BCK15Bio}. Here, the understanding of the predictive relations between covariates is important for constructing joint confidence intervals for the entire parameter vector $\{(\gamma_{mk}^0)_{k=1}^K\}_{m=1}^M$ in the main regression equations. Indeed, the construction relies on the semi-parametrically efficient point estimators obtained from the empirical analog of the following orthogonalized moment equation:
	\begin{equation}
	\E [(U^0_{mk,t} - X_{k,t} \gamma_{mk}^0) \nu_{k,t}]=0,  \quad k=1,\ldots,K, \quad m=1, \ldots, M,
	\end{equation}
	where $U^0_{mk,t} = U_{m,t} - X_{-k,t}^\top\gamma_{m(-k)}^0 $ is the response variable minus the part explained by the covariates other than $k$. Note that the empirical analog would have all unknown nuisance parameters replaced by the estimators.
	%
	
	
\end{example}

\begin{example}[Simultaneous Equation Systems (SES)]
	Suppose there are many regression equations in the following form:
	$$
	U_{m,t} = U_{-m,t}^\top\delta_m^0 + X_t^\top \gamma_{m}^0 + \varepsilon_{m,t}, \quad m=1, \ldots, M.
	$$
	Move all the endogenous variables to the left-hand side and rewrite the model in the vector form
	$$
	\mathbf{D}U_{t} = \bm{\Gamma} X_t+ \varepsilon_{t},
	$$
	which is also called the structural form of the model. Suppose that $D$ is invertible. Then the corresponding reduced form is given by
	\begin{equation}
	U_t =\mathbf{B}X_{t} + \nu_t, \quad \E \nu_{m,t} X_{t} = 0,\quad m=1, \ldots, M,
	\end{equation}
	with $\mathbf{B}=\mathbf{D}^{-1}\bm{\Gamma}$ and $\nu_t=\mathbf{D}^{-1}\vps_t$. In this case the $Y_{j,t}$'s and $X_{j,t}$'s in SRE have no overlapping variables. A high-dimensional SES can be considered as a special case of SRE with 
	$$
	(Y_{j,t}, X_{j,t}, \vps_{j,t}, \beta_{j}^0) = (U_{j, t}, X_t, \nu_{j,t}, \mathbf{B}_{j\cdot}^\top),\quad j = 1, \ldots, M.
	$$
\end{example}

\begin{example}[Large Vector Autoregression Models]\label{examp4}
	In the case where the covariates involve lagged variables of the response, SRE can be written as a large vector autoregression model. 
	For example, the VAR($p$) model,
	\begin{equation}
	U_t =\sum_{\ell=1}^p\mathbf{B}^\ell U_{t-\ell} + \eps_t, \quad \E \vps_{m,t} U_{t-\ell} = 0,\quad m=1, \ldots, M,
	\end{equation}
	where $U_t = (U_{1,t}, U_{2,t}, \ldots, U_{M,t})^\top$, and $\vps_t$ is an $M$-dimensional white noise or innovation process; see e.g. Chapter 2.1 in \citet{lutkepohl2005new}. It is a special SRE case again with
	$$
	(Y_{j,t}, X_{j,t}, \vps_{j,t}, \beta_{j}^0) = (U_{j, t}, (U_{t-1}^\top,\ldots,U_{t-p}^\top)^\top, \vps_{j,t}, (\mathbf{B}^1_{j\cdot},\ldots,\mathbf{B}^p_{j\cdot})^\top),\quad j = 1, \ldots, M.
	$$	
	
	{
		Such dynamics are of interest in biology to understand dynamic gene expression network association using micro array data; see for example \cite{opgen2007correlation,ramirez2017dynamic,dimitrakopoulou2011dynamic}. It is understood that a crucial feature for many gene networks is their inherent sparsity. The issue of the number of variables involved is potentially larger than the sample size can be addressed by LASSO. Our methodology can help to analyze a gene interaction correlation network in a high dimensional regression scheme. In particular,
		suppose that each vertex represents a gene $j$ collected at time point $t$ with $U_{j,t}$ as its gene expression and an edge connects two genes if they are correlated.}
	
\end{example}

%
%

{We refer to Section \ref{pae} in the supplementary materials for more practical examples.}

\section{Effective Prediction Using Sparsity Method}\label{lasso}
In this section, we present our model setup and the LASSO estimation algorithm, including the joint penalty selection procedure.
\subsection{Sparsity in SRE}

The general SRE structure makes it possible to predict $Y_{j,t}$ using $X_{j,t}$ effectively. Note that the dimension of $X_{j,t}$ is large, potentially larger than $n$.
Without loss of generality we assume exact sparsity of $\beta_{j}^0$ throughout the paper:
\beq\label{sparse}
s_j = |\beta_{j}^0|_0 \leq s = \smallO(n), \quad j =1, \ldots J,
\eeq
where the $\ell_0$-norm, $|\cdot|_0$, is the number of nonzero components of a vector.
\begin{remark}
	It is now well understood that sparsity can be easily extended to approximate sparsity, in which the sorted absolute values of coefficients decrease fast to zero. {
		To be more specific, when $\beta^0_{jk}$ is not sparse, we shall define an intermediary optimal value for our true coefficients, i.e. $\beta^*_{jk}$. Let $LC_{p} \defeq \underset{|\beta_j|_0 \leq p}{\min} [\E_n\{X_{j,t}^{\top}(\beta_{j}- \beta^0_{j})\}^2]^{1/2}$, additionally with proper conditions on the design matrix, the optimal sparsity level is given by $s_j^* = \underset{0 \leq p\leq (K \wedge n)}{\min} LC_{p}^2+ (\underset{1\leq k\leq K}{\max} \Psi^2_{jk}) p/n$, where $ \Psi^2_{jk}$ is the long run variance of $ \frac{1}{\sqrt{n}} \sum_{t=1}^n \vps_{j,t} X_{jk,t}$. Then the oracle $\beta^*_{jk} $ is defined to be $\arg\underset{|\beta_{j}|_0 \leq s_j^*}{\min} \E_n\{X_{j,t}^\top(\beta_{j}- \beta^0_{j})\}^2 $. Thus an additional term involving $LC_{s_j^*}$ will appear in the bound in case of the true signal  $\beta^0_{jk}$ is not sparse. With approximate sparsity we mean that the true signal is not sparse but nevertheless can be approximated by an exact sparsity set-up well, namely $|\beta^0_{jk}| \leq A k^{-\gamma}$ (ranked in descending order), where $\gamma> 0.5$, and by taking $s_j^* \propto n^{1/(2\gamma)}$ the goal would be achieved. }
\end{remark}

For this situation one employs an $\ell_1$-penalized estimator of $\beta_{j}^0$ of the form:
\beq\label{beta}
\hat \beta_j = \arg\min_{\beta \in  {\R}^{K}} \En (Y_{j,t} - X_{j,t}^\top\beta)^2
+ \frac{\lambda}{n} \sum_{k=1}^{K} | \beta_{k}| \Psi_{jk},
\eeq
where $\lambda$ is the joint "optimal" penalty level and $\Psi_{jk}$'s are penalty loadings, which are defined below in (\ref{eq1}).

A first aim is to obtain performance bounds with respect to the prediction norm:
$$
|\hat \beta_j - \beta_{j}^0|_{j, pr} \defeq \bigg[\En \big\{ X_{j,t}^\top(\hat \beta_j - \beta_{j}^0)\big\}^2 \bigg]^{1/2},
$$
{where the outside $j$ indicates to use the covariates in the $j$th equation $X_{j,t}$ in computing the prediction norm,} and the Euclidean norm:
$$
|\hat \beta_j - \beta_{j}^0|_{2} \defeq  \big\{\sum_{k=1}^K(\hat \beta_{jk} - \beta_{jk}^0)^2 \big\}^{1/2}.
$$
To achieve good performance bounds, we first consider "ideal" choices (IC) of the penalty level and the penalty loadings. Let
$$
S_{jk} = \frac{1}{\sqrt{n}} \sum_{t=1}^n \vps_{j,t} X_{jk,t},
$$
where for a moment we assume to be able to observe  $\vps_{j,t} = Y_{j,t} - X_{j,t}^\top\beta_{j}^0$. In practice one obtains an approximation by stepwise LASSO. Set
\begin{eqnarray}
&& {\Psi_{jk} \defeq \sqrt{ \operatorname{avar}(S_{jk})}} \label{eq1}, \\
&& \lambda^0(1-\alpha) \defeq (1-\alpha)-\textrm{quantile of } 2c \sqrt{n}\max_{1 \leq j\leq J, 1\leq k \leq K} |S_{jk}/\Psi_{jk}|,
\end{eqnarray}
where $c >1$, e.g., $c=1.1$, and $1-\alpha$ is a confidence level, e.g. $\alpha =0.1$, {where the long run variance is denoted by $\operatorname{avar}$}. 

Theoretically, we can characterize the rate of $\lambda^0(1-\alpha)$ by the tail probability of $S_{jk}$ (see Theorem \ref{gammabound}), also via Gaussian Approximation as in corollary \ref{betabound4}.
To calculate $\lambda^0(1-\alpha)$ from data, we can also use a Gaussian approximation based on:
$$
Q(1-\alpha) \defeq (1-\alpha)-\textrm{quantile of } {2c \sqrt{n}}\max_{1 \leq j\leq J, 1\leq k \leq K}|Z_{jk}/\Psi_{jk}|,
$$
where $\{Z_{jk}\}$ are multivariate Gaussian centered random variables with the same {long run covariance} structure as {$\{S_{jk}\}$}. Alternatively, we can employ a multiplier bootstrap procedure to estimate IC empirically to achieve a better finite sample performance; see for example \cite{CCK13AoS}.
In case of dependent observations over time, it is understood that data cannot be resampled directly as in the the i.i.d. case, as the dependency structure of the underlying processes will be lost. A usual solution to this problem is to consider a block bootstrap procedure, where the data are grouped into blocks, resampled and concatenated. In particular, we will adopt an estimate of IC by a multiplier block bootstrap procedure.
{The theoretical properties of LASSO and the tuning parameter choices are presented in Section \ref{lasso.ic}-\ref{lasso.joint}.}

%

\subsection{Multiplier Bootstrap for the Joint Penalty Level}\label{mblambda}
In this subsection, we introduce an algorithm to approximate the joint penalty level via a block multiplier bootstrap procedure, which is particularly non-overlapping block bootstrap (NBB). Consider the system of equations with dependent data:
\begin{equation}\label{system}
Y_{j,t} = X_{j,t}^\top\beta_{j}^0 + \varepsilon_{j,t},  \quad \E \varepsilon_{j,t} X_{j,t} = 0, \quad j=1,...,J,
\quad t=1,\ldots, n,
\end{equation}
\begin{enumerate}
	\item[S1] Run the initial $\ell_1$-penalized regression equation by equation, i.e. for the $j$th equation,
	\begin{equation}\label{beta.tilde}
	\widetilde \beta_j = \arg\min_{\beta \in  {\R}^{K}} \En (Y_{j,t} - X_{j,t}^\top\beta)^2
	+ \frac{\lambda_j}{n} \sum_{k=1}^{K_j} | \beta_{jk}| \Psi_{jk},
	\end{equation}
	where $\lambda_j$ are the penalty levels and $\Psi_{jk}$ are the penalty loadings. For instance, we can take the $X$-independence choice using Gaussian approximation (in the heteroscedasticity case): $2c'\sqrt{n}\Phi^{-1}\{1-\alpha'/(2K)\}$ for $\lambda_j$, where $\Phi(\cdot)$ denotes the cdf of $\N(0,1)$, 
	$\alpha'=0.1$, $c'=0.5$, and choose $\sqrt{\operatorname{lvar}(X_{jk,t}\breve{\varepsilon}_{j,t})}$
	for the penalty loadings, where $\breve{\varepsilon}_{j,t}$ are preliminary estimated errors {and $\operatorname{lvar}(X_{jk,t}\breve{\varepsilon}_{j,t})$ is an estimate of the long-run variance $\sum_{\ell=-\infty}^{\infty}\E(X_{jk,t}\breve{\varepsilon}_{j,t}X_{jk,(t-\ell)}\breve{\varepsilon}_{j,(t-\ell)})$, e.g. the Newey-West estimator is given by $$\sum_{\ell=-p_n}^{p_n}k(\ell/p_n)\operatorname{cov}(X_{jk,t}\breve{\varepsilon}_{j,t},X_{jk,(t-\ell)}\breve{\varepsilon}_{j,(t-\ell)}),$$ with $k(z)=(1-|z|)\IF(|z|\leq1)$.} We note that the $X$-independent penalty (using Gaussian approximation) is more conservative, as the correlations among regressors can be adapted in the $X$-dependent case (using a multiplier bootstrap) with a less aggressive penalty level.
	\item[S2] Obtain the residuals for each equation by $\widetilde{\vps}_{j,t}=Y_{j,t} - X_{j,t}^\top\widetilde{\beta}_{j}$, and compute 
	$\Psi_{jk}=\sqrt{\operatorname{lvar}(X_{jk,t}\widetilde{\varepsilon}_{j,t})}$. 
	\item[S3] 
	Divide $\{\tilde{\varepsilon}_{j,t}\}$ into $l_n$ blocks containing the same number of observations $b_n$, $n=b_nl_n$, where $b_n, l_n\in\Z$.
	Then choose $\lambda=2c\sqrt{n}q^{[B]}_{(1-\alpha)}$, where $q^{[B]}_{(1-\alpha)}$ is the $(1-\alpha)$ quantile of $\underset{1\leq j \leq J, 1\leq k \leq K}{\max} |Z^{[B]}_{jk}{/\Psi_{jk}}|$, and $Z^{[B]}_{jk}$ are defined as
	\begin{equation}\label{mbb}
	\quad Z^{[B]}_{jk}=\frac{1}{\sqrt{n}} \sum_{i=1}^{l_n} e_{j,i}\sum_{l=(i-1)b_n+1}^{ib_n} \tilde{\varepsilon}_{j,l} X_{jk,l},
	\end{equation}
	$e_{j,i}$ are i.i.d. $\operatorname{N}(0,1)$ random variables independent of the data. 
\end{enumerate}

The bootstrap consistency regarding $Z^{[B]}_{jk}$ is proved in Theorem \ref{validboot}.
\begin{remark}[Block bootstrap procedures]
	
	\begin{enumerate}
		\item[(i)] 
		Concerning the determination of $b_n$, 
		we shall report the prediction norm with several block sizes $b_n$ {and select the one with the best prediction performance in the simulation study}.
		{In addition, if it is the case that $n$ cannot be divided by $b_n$ with no remainder, one can simply take $l_n=\lfloor n/b_n\rfloor$ and drop the remaining observations. }
		\item[(ii)] Other forms of multiplier bootstrap with any random multipliers centered around 0 can also be considered.
		\item[(iii)] Alternative block bootstrap procedures can be adopted, such as the circular bootstrap and the stationary bootstrap among others; see for example \cite{lahiri1999theoretical} for an overview.
		
	\end{enumerate}
\end{remark}

\section{Valid Inference on the Coefficients}\label{inference}
With a reasonable fitting of LASSO on hand, we can proceed to investigate the issue of simultaneous inference. This section focuses on SRE of Example 2. We allow the covariates in each equation to be different.

The basic idea to facilitate inference is to formulate the estimation in a semi-parametric framework. With partialing out the effect of the nonparametric coefficient(s), we can achieve the desired estimation accuracy of the parametric component of interest. This trick is referred to as "Neyman orthogonalization". Notably, the procedure is equivalent to the well known de-sparsification procedure in the mean square loss case, which is developed for the inference on the estimated zero coefficients by LASSO. It thus serves the same purpose of generating a (robust) de-sparsified estimation for LASSO inference.

We list three algorithms to estimate $\beta^0_{jk}$. 
Algorithm 1 is easy to implement and algorithm 2 is tailored to the cases of heavy-tailed distribution of the error term, as Least Absolute Deviation (LAD) regression is well known to be robust against outliers. Algorithm 3 considers a double selection procedure aimed at remedying the bias due to omitted variables by one step selection, 
while also accounting for the cases of heteroscedastic errors.
%

\textit{Algorithm} 1: LS-based algorithm
\begin{itemize}\label{algo1}
	\item[S1] Consider $Y_{j,t} = X_{jk,t}\beta^0_{jk}+ X_{j(-k),t}^{\top}\beta^0_{j(-k)}+\varepsilon_{j,t}$, run (post) LS LASSO procedure (for each $j$), 
	and keep the quantity $X_{j(-k),t}^{\top}\hat{\beta}^{[1]}_{j(-k)}$ for each $k$.
	\item[S2] Run (post) LS LASSO (for each $j,k$) 
	by regressing $X_{jk,t} = X_{j(-k),t}^\top\gamma_{j(-k)}^0+ v_{jk,t}$, 
	and keep the residuals as $\hat{v}_{jk,t} = X_{jk,t} -  X_{j(-k),t}^{\top}\hat{\gamma}_{j(-k)}$.
	\item[S3] Run LS IV regression of $Y_{j,t} - X_{j(-k),t}^{\top}\hat{\beta}^{[1]}_{j(-k)}$ on $X_{jk,t}$ using $\hat{v}_{jk,t}$ as an instrument variable, attaining the final estimator $\hat{\beta}^{[2]}_{jk}$.
\end{itemize}

\textit{Algorithm} 2: LAD-based algorithm
\begin{itemize}\label{algo2}
	\item[S1] and S2 are the same as Algorithm \hyperref[algo1]{1}.
	\item[S3$'$] Run LAD IV regression of $Y_{j,t} - X_{j(-k),t}^{\top}\hat{\beta}^{[1]}_{j(-k)}$ on $X_{jk,t}$ using $\hat{v}_{jk,t}$ as an instrument variable, attaining the final estimator $\hat{\beta}^{[2]}_{jk}$. We refer to \citet{BCK15Bio,CH08IV} for more details about how to achieve the estimator in this step.
\end{itemize}

{The theoretical properties of the estimators $\hat\beta_{j(-k)}^{[1]}$ and $\hat\gamma_{j(-k)}$ in S1 and S2 are provided in Corollary \ref{betabound2} or \ref{betabound4} (see Corollary \ref{lambdabound.joint} or \ref{betabound4.joint} in the supplementary correspondingly if the joint penalty over equations is employed), and Theorem \ref{post} for post LASSO, respectively. The uniform Bahadur representation and the Central Limit Theorem of the estimator $\hat\beta^{[2]}_{jk}$ in S3 or S3$'$ are established in Theorem \ref{bahadur} and Corollary \ref{uninorm}.}

\begin{remark} Our algorithms follow patterns discussed in  \cite{BCK15Bio,BCK15_sup} in the i.i.d. settings. The IV estimator obtained in S3 of Algorithm \hyperref[algo1]{1} reduced to the de-biased LASSO estimator \citep{zhang2014debiased,van2014desparsified} and is also first-order equivalent to the double LASSO
	method in \citet{BCH2011,BCH2014}. In particular, the estimator under LS IV regression (2-step least square regression) is given by
	\begin{align}\label{de-sparsified}
	\hat{\beta}_{jk}^{[2]}&=(\hat{v}_{jk}^\top X_{jk})^{-1}\hat{v}_{jk}^\top(Y_j-X_{j(-k)}^{\top}\hat{\beta}^{[1]}_{j(-k)})\notag\\
	&=(\hat{v}_{jk}^\top X_{jk})^{-1}\hat{v}_{jk}^\top Y_j - \underset{m\neq k}{\sum}\frac{\hat{v}_{jk}^\top X_{jm}}{\hat{v}_{jk}^\top X_{jk}}\hat{\beta}^{[1]}_{jm}.
	\end{align}
	The second line in \eqref{de-sparsified} is exactly the same as the de-biased or de-sparsified LASSO estimator given in Eq. (5) in \citet{zhang2014debiased} or Eq. (5) in \citet{van2014desparsified}.
	As remarked in \cite{BCK15Bio,BCK15_sup}, one can alternatively implement an algorithm via double selection as in \citet{BCH2011,BCH2014}. In particular, heteroscedastic LASSO is employed in S2$''$ and the IV regression is replaced by a either LASSO or LAD regression on the target variable and all covariates selected in the first two steps. \qed
\end{remark}

\textit{Algorithm} 3: Double selection-based algorithm
\begin{itemize}\label{algo3}
	\item[S1$''$] Run LS LASSO (for each $j$) of $Y_{j,t}$ on $X_{j,t}$:
	\beq
	\hat \beta_j^{[1]} = \arg\min_{\beta} \En (Y_{j,t} - X_{j,t}^\top\beta)^2
	+ \frac{\lambda}{n} |\widehat\Psi_j\beta|_1.\notag
	\eeq
	\item[S2$''$] Run Heteroscedastic LASSO (for each $j,k$) 
	of $X_{jk,t}$ on $X_{j(-k),t}$:
	\beq
	\hat \gamma_{j(-k)} = \arg\min_{\gamma} \En (X_{jk,t} - X_{j(-k),t}^\top\gamma)^2
	+ \frac{\lambda'}{n} |\widehat\Gamma_j\gamma|_1,\notag
	\eeq
	where penalty loadings $\widehat\Gamma_{j}$ can be initialized as $\sqrt{\operatorname{lvar}\{X_{j\ell,t}(X_{jk,t}-\En X_{jk,t})\}}$ and then refined by $\sqrt{\operatorname{lvar}(X_{j\ell,t}\hat v_{jk,t})}$,
	for $\ell\neq k$, and $\hat{v}_{jk,t} = X_{jk,t} -  X_{j(-k),t}^{\top}\hat{\gamma}_{j(-k)}$ can be obtained by using the initial ones.
	\item[S3$''$] Run LS regression of $Y_{j,t}$ on $X_{jk,t}$ and the covariates selected in S1$''$ and S2$''$:
	\beq
	\hat \beta_j^{[2]} = \arg\min_{\beta}\{\En (Y_{j,t} - X_{j,t}^\top\beta) ^2:\,\supp(\beta_{-k})\subseteq\supp(\hat\beta^{[1]}_{j(-k)})\cup\supp(\hat\gamma_{j(-k)})\}.\notag
	\eeq
	\item[S3$'''$] Run LAD regression of $Y_{j,t}$ on $X_{jk,t}$ and the covariates selected in S1$''$ and S2$''$:
	\beq
	\hat \beta_j^{[2]} = \arg\min_{\beta}\{\En |Y_{j,t} - X_{j,t}^\top\beta|:\,\supp(\beta_{-k})\subseteq\supp(\hat\beta^{[1]}_{j(-k)})\cup\supp(\hat\gamma_{j(-k)})\}.\notag
	\eeq
\end{itemize}		
As shown in \citet{BCH2011} and \citet{BCK15_sup}, the double selection approach in S3$''$ or S3$'''$ creates an orthogonality condition with respect to the space spanned by the covariates selected by both steps, and thus generates an orthogonal relation to any space spanned by a linear projection of the covariates, e.g. $\hat v_{jk,t}$. Therefore, the inference on the parameters may still be applied as in the framework of Algorithm \hyperref[algo1]{1} and \hyperref[algo2]{2}. {Therefore, one may still find the theoretical properties of estimators in S1$''$, S2$''$, S3$''$ (S3$'''$) in Section \ref{theorem} according to the links mentioned above.}

\subsection{Confidence Interval for a Single Coefficient}\label{sec.single}
We discuss an inference framework developed for a single coefficient obtained from the aforementioned algorithms.

Let $\psi_{jk}(Z_{j,t},\beta_{jk},h_{jk})$ denote the score function, where $Z_{j,t}=(Y_{j,t},X^\top_{j,t})^\top$, $h_{jk}(X_{j(-k),t})=(X_{j(-k),t}^\top\beta_{j(-k)},X_{j(-k),t}^\top\gamma_{j(-k)})^\top$. Consider the LAD-based case with $\psi_{jk}(Z_{j,t},\beta_{jk},h_{jk})=\{1/2-\IF(Y_{j,t}\leq X_{jk,t}\beta_{jk}+X_{j(-k),t}^\top\beta_{j(-k)})\}v_{jk,t}$, define $\omega_{jk} \defeq\E \{(\frac{1}{\sqrt{n}}\sum_{t=1}^n\psi^0_{jk,t})^2\}=\sum_{\ell=-(n-1)}^{n-1}(1-\frac{|\ell|}{n})\operatorname{cov}(\psi^0_{jk,t},\psi^0_{jk,(t-\ell)})$ with $\psi_{jk,t}^0\defeq\psi_{jk}(Z_{j,t},\beta^0_{jk},h^0_{jk})$, 
and $\phi_{jk}\defeq\frac{\partial \E\{\psi_{jk}(Z_{j,t},\beta,h^0_{jk})\}}{\partial \beta}|_{\beta=\beta_{jk}^0}$.

Suppose we are interested in testing $H_0: \beta^0_{jk}=0$. For this purpose we employ the uniform Bahadur representation (Theorem \ref{bahadur}) to construct the confidence interval via a multiplier bootstrap procedure. In particular, the distribution of the asymptotically pivotal statistics:
\beq \label{single}
T_{jk} =  \frac{\sqrt{n}(\hat{\beta}^{[2]}_{jk}- \beta_{jk}^0)}{\hat{\sigma}_{jk}},
\eeq
is approximated via its block multiplier bootstrap counterpart:
\beq \label{blockboot}
T^\ast_{jk}=\frac{1}{\sqrt{n}} \sum_{i=1}^{l_n} e_{j,i}\sum_{l=(i-1)b_n+1}^{ib_n} \hat\zeta_{jk,l},
\eeq
where {$\hat\zeta_{jk,t}$ are pre-estimators of $\zeta_{jk,t}=-\phi^{-1}_{jk}\sigma_{jk}^{-1}\psi^0_{jk,t}$ such that \\
	$\underset{(j,k),(j',k')}{\max}|\sum_{i=1}^{l_n} \hat\eta_{j'k',i}\hat\eta_{jk,i} - \sum_{i=1}^{l_n} \eta_{j'k',i}\eta_{jk,i}|=\smallO_{\P}(\{\log(JK)\}^{-2})$, with $\eta_{jk,i} \defeq \frac{1}{\sqrt{n}} \sum_{l=(i-1)b_n+1}^{ib_n} \zeta_{jk,l} $ and $\hat{\eta}_{jk,i} \defeq\frac{1}{\sqrt{n}}  \sum_{l=(i-1)b_n+1}^{ib_n} \hat \zeta_{jk,l}$}, $e_{j,i}$ are independently drawn from $\operatorname{N}(0,1)$, $l_n$ and $b_n$ are the numbers of blocks and block size, respectively.
{More discussion on how one can construct the consistent pre-estimators $\hat\zeta_{jk,t}$ is stated in the supplementary material; see Comment \ref{zeta.remark}.}

Let $\hat{\sigma}_{jk}$ be any consistent estimator of $\sigma_{jk}$. 
Then the confidence interval is given by
\begin{equation}\label{CI.boot.ind}
\operatorname{CI}_{jk}^\ast(\alpha) : [\hat{\beta}^{[2]}_{jk}-\hat{\sigma}_{jk} n^{-1/2}q^{\ast}_{jk}(1-\alpha), \hat{\beta}^{[2]}_{jk}+\hat{\sigma}_{jk} n^{-1/2}q^{\ast}_{jk}(1-\alpha)],
\end{equation}
where $q^{\ast}_{jk}(1-\alpha)$ is the $(1-\alpha)$ quantile of the bootstrapped distribution of $|T_{jk}^{\ast}|$.

{
	\begin{remark}[Asymptotic Normality of $\hat\beta^{[2]}_{jk}$]
		As shown in Corollary \ref{asy.norm} we have the limit distribution of $\hat\beta^{[2]}_{jk}$:
		\begin{equation}\label{asynorm}
		\sigma^{-1}_{jk} n^{1/2}(\hat{\beta}^{[2]}_{jk}- \beta^0_{jk}) \stackrel{\mathcal{L}}{\rightarrow}  \N (0,1),
		\end{equation}
		where $\sigma_{jk} = (\phi_{jk}^{-2}\omega_{jk})^{1/2}$. Therefore, the two-sided $100(1-\alpha)$ confidence interval by asymptotic normality for $\beta^0_{jk}$ is given by
		\begin{equation}\label{CI.asy}
		\operatorname{CI}_{jk}(\alpha) : [\hat{\beta}^{[2]}_{jk}-\hat{\sigma}_{jk} n^{-1/2}\Phi^{-1}(1-\alpha/2), \hat{\beta}^{[2]}_{jk}+\hat{\sigma}_{jk}n^{-1/2}\Phi^{-1}(1-\alpha/2)].
		\end{equation}
\end{remark}}

\begin{remark}[Residual Multiplier Bootstrap]
	Alternative bootstrap procedures may be considered as well, e.g. the residual multiplier bootstrap procedure:
	\beq
	\hat \eps_{j,t} = Y_{j,t} - X_{j,t}^\top\hat{\beta}^{[1]}_{j},\notag
	\eeq
	then divide $\{\hat\eps_{j,t}\}$ into $l_n$ blocks of size $b_n$, where $b_nl_n=n$, and for each block $i=1,\ldots,l_n$,
	\beq
	\eps^{\ast}_{j,t} = (\hat \eps_{j,t} - \En\hat\eps_{j,t}) e_{j,i},\,\, \text{for }t\in\{(i-1)b_n+1,\ldots,ib_n\}. \notag
	\eeq
	Define $Y^{\ast}_{j,t} = X_{j,t}^{\top}\hat{\beta}^{[1]}_{j}+\eps^{\ast}_{j,t} $ and compute the bootstrap counterpart as
	\beq
	T_{jk}^{\ast} = \frac{\sqrt{n}(\hat{\beta}^{\ast}_{jk}- \hat{\beta}^{[1]}_{jk})}{\hat{\sigma}^{\ast}_{jk}},\notag
	\eeq
	where $\hat{\beta}^{\ast}_{jk}$ and $\hat{\sigma}^{\ast}_{jk}$ are estimated using the bootstrap sample $\{Y^{\ast}_{j,t}, X_{j,t}\}$.
	%
	
\end{remark}

\subsection{Joint Confidence Region for Simultaneous Inference}\label{scr}
We now continue to extend the single coefficient inference to simultaneous inference on a set of coefficients. As shown in the practical examples in Section \ref{pae}, it is essential to conduct simultaneous inference on a group of parameters $G$. In this case, the null hypothesis is:
$\mathbf{H}_0:  \beta_{jk}^0 = 0$, $\forall (j,k) \in G $, and the alternative $\mathbf{H}_A: \beta_{jk}^0 \neq 0$, for some $(j,k) \in G$, where the group $G$ is a set of coefficients with cardinality $|G|$. 
Suppose for the $j$-th equation there are $p_j$ target coefficients and the cardinality $|G| = \sum^J_{j=1} p_{j}$.
This can be understood as a multiple estimation problem compared to Section \ref{sec.single}.
Without loss of generality, we can rearrange the order of the variables and rewrite the regression equation for each $j$ as (consider the LAD-based model here)
\begin{equation}
Y_{j,t} = \sum^{p_j}_{l= 1} X_{jl,t}\beta^0_{jl}+\sum^{K}_{l= p_j+1}X_{jl,t}\beta^0_{jl}+\varepsilon_{j,t}, \quad F_{\varepsilon_{j}}(0) = 1/2
\end{equation}

One follows the algorithms to obtain $\hat{\beta}_{jl} (1\leq l\leq p_{j})$ for each $j$. Then the idea of simultaneous inference is very straightforward. We aggregate the statistics $T_{jk}$
in \eqref{single} by taking the maximum and minimum over the set $G$. Finally, the component-wise confidence interval is constructed with the quantiles of the bootstrap statistics over all bootstrap samples.



Denote $q^{\ast}_{G}(1-\alpha)$ as the $(1-\alpha)$ quantile of $\underset{(j,k)\in G}{\max} |T^{\ast}_{jk}|$. A joint confidence region is then:
\begin{equation}\label{CI.boot.sim}
\Big\{\beta\in{\rm I\!R}^{|G|}:\underset{(j,k)\in G}{\max}T_{jk}\leq q^{\ast}_{G}(1-\alpha) \text{ and } \underset{(j,k)\in G}{\min}T_{jk}\geq -q^{\ast}_{G}(1-\alpha)\Big\},
\end{equation}
and for each component $(j,k)\in G$, the confidence interval $\widetilde{\operatorname{CI}}^\ast_{jk}(\alpha)$ is given by 
$ [\hat{\beta}^{[2]}_{jk}-\hat{\sigma}_{jk} n^{-1/2}q^{\ast}_{G}(1-\alpha), \hat{\beta}^{[2]}_{jk}+\hat{\sigma}_{jk}n^{-1/2}q^{\ast}_{G}(1-\alpha)]$. We show in Corollary \ref{proofboot} the consistency of this bootstrap confidence band for simultaneous inference.
{Note that when there is only one parameter in $G$ for inference, the joint confidence region \eqref{CI.boot.sim} will reduce to the single parameter confidence interval \eqref{CI.boot.ind} as a special case.}

\section{Main Theorems}\label{theorem}
In this section, we present the theoretical foundations for the procedures given earlier. In particular, we discuss the properties of the theoretical choices of penalty level and the validity of the other two empirical choices, as well as the theoretical support for the simultaneous inference.

Throughout the whole section, we define $S_{jk} \defeq n^{-1/2} \sum_{t=1}^n \varepsilon_{j,t} X_{jk,t}$, $S_{j\cdot}= (S_{jk})_{k=1}^K$, and $\Psi_{jk}\defeq\sqrt{\operatorname{avar}(S_{jk})}$, which is the square root of the long-run variance of $X_{jk,t}\vps_{j,t}$, namely \\$\{\sum^{\infty}_{\ell=-\infty} \E (X_{jk,t}X_{jk,(t-\ell)}\vps_{j,t} \vps_{j,(t-\ell)} )\}^{1/2}$. 
Recall that for a single equation LASSO, we select the penalty in the following ways:
\begin{itemize}
	\item[a)]\label{lam_a} theoretically, for each regression, $\lambda_j$ is $\lambda_j^0(1-\alpha)$ (IC), i.e. the $(1-\alpha)$ quantile of \\
	$ 2 c\sqrt{n} \underset{1\leq k \leq K}{\max}|S_{jk}{/\Psi_{jk}}|$ (note that this penalty takes into account the correlation among regressors and is design adaptive);
	\item[b)]\label{lam_b} an empirical choice given a Gaussian approximation result is {$Q_j(1-\alpha) $, which is defined to be the $(1-\alpha)$ quantile of $ 2 c \underset{1\leq k \leq K}{\max}\sqrt{n}|Z_{jk}{/\Psi_{jk}}|$, where $Z_{jk}$'s are multivariate Gaussian centered random variables with the same long run covariance structure as $S_{jk}$. Alternatively, a canonical choice disregarding the correlation among regressors can be considered as
		$\widetilde{Q}_{j}(1-\alpha) \defeq 2 c \sqrt{n}\Phi^{-1}\{1-\alpha/(2K)\}$. We shall note that $Q_j(1-\alpha)$ is not feasible but can be estimated by simulations of Gaussian random variable $Z_{jk}$ with estimated long run variance covariance matrix. Typically $\widetilde{Q}_{j}(1-\alpha)$ is more conservative than $Q_j(1-\alpha)$.}
	\item[c)]\label{lam_c} another empirical choice of the penalty level is $\Lambda_j(1-\alpha)$ as the $(1-\alpha)$ quantile of \\
	$2 c\sqrt{n} \underset{1\leq k \leq K}{\max}| Z^{[B]}_{jk}{/\hat\Psi_{jk}}|$ ($Z^{[B]}_{jk}$'s are defined in \eqref{mbb}), and obtainable via the multiplier block bootstrap technique.
\end{itemize}

\subsection{Near Oracle Inequalities under IC}\label{lasso.ic}

We first provide the near oracle inequalities for the single equation LASSO estimation $\tilde{\beta}_j$ obtained from \eqref{beta.tilde} under the ideal choices (IC). For this purpose, a few assumptions and definitions are required.
\begin{itemize}
	\item[(A1)]\label{A1}
	For $j=1,\ldots,J,k=1,\ldots,K$, let $X_{jk,t}$ and $\vps_{j,t}$ be stationary processes admitting the following representation forms $X_{jk,t} = g_{jk}(\mathcal{F}_{t})=g_{jk}(\ldots, \xi_{t-1}, \xi_{t})$ and $\vps_{j,t} = h_{j}(\mathcal{F}_{t}) = h_{j}(\ldots, \eta_{t-1}, \eta_{t})$, where $\xi_{t}, \eta_{t} $ are i.i.d. random elements (innovations or shocks, allowing for overlap; see Comment \ref{exo}) across $t$, $\mathcal F_t=(\ldots,\xi_{t-1},\eta_{t-1},\xi_t,\eta_t)$, $g_{jk}(\cdot)$ and $h_{j}(\cdot)$ are measurable functions (filters). $\E (X_{jk,t}\vps_{j,t}) = 0,$ for any $j,k\in 1, \cdots, J, 1, \cdots, K$.
\end{itemize}

\begin{definition}\label{norm}
	Let 
	$\xi_{0}$ be replaced by an i.i.d. copy of $\xi_{0}^\ast$, and 
	$X_{jk,t}^\ast = g_{jk}(\ldots, \xi^\ast_0,\ldots,\xi_{t-1}, \xi_{t})$.	For $q \geq 1$, define the functional dependence measure 
	$\delta_{q,j,k,t} \defeq \|X_{jk,t}- X_{jk,t}^\ast\|_q$, which measures the dependency of $\xi_{0}$ on $X_{jk,t}$. Also define $\Delta_{m,q,j,k} \defeq \sum^{\infty}_{t=m} \delta_{q,j,k,t} $, which measures the cumulative effect of $\xi_{0}$ on $X_{jk,t\geq m}$. Moreover, we introduce the dependence adjusted norm of $X_{jk,t}$
	as $\|X_{jk,\cdot}\|_{q,\varsigma}\defeq \sup_{m\geq 0}(m+1)^{\varsigma} \Delta_{m,q,j,k} (\varsigma>0)$. Similarly, let $\eta_{0}$ be replaced by an i.i.d. copy of $\eta_{0}^\ast$, and $\vps_{j,t}^\ast = h_{j}(\ldots, \eta^\ast_0,\ldots,\eta_{t-1}, \eta_{t})$, we define $\|\vps_{j,\cdot}\|_{q,\varsigma}\defeq\sup_{m\geq 0}(m+1)^{\varsigma}\sum^{\infty}_{t=m}\|\vps_{j,t}- \vps_{j,t}^\ast\|_q$ and $\|X_{jk,\cdot}\vps_{j,\cdot}\|_{q,\varsigma}\defeq\sup_{m\geq 0}(m+1)^{\varsigma}\sum^{\infty}_{t=m}\|X_{jk,t}\vps_{j,t}- X_{jk,t}^\ast\vps_{j,t}^\ast\|_q$.
\end{definition}
It should be noted that \hyperref[A1]{(A1)} admits a wide class of processes. The largest value of $\varsigma$ which ensures a finite dependence adjusted norm characterizes the dependency structure of the process.
The moment-based measure is directly connected with the impulse functions.
{A few examples for univariate time series $Z_t$ are listed in Appendix \ref{example} in the supplementary materials.}

\begin{itemize}
	\item[(A2)]\label{A2}
	Restricted eigenvalue (RE): given $\bar c \geq 1$, for $\delta\in {\R}^{K}$, with probability $1-\smallO(1)$,
	\beq
	\kappa_j(\bar  c) \defeq 
	\min_{|\delta_{T_j^c}|_1\leq \bar c |\delta_{T_j}|_1,\,\delta \neq 0}\frac{\sqrt{s_j}|\delta|_{j,pr}}{|\delta_{T_j}|_1}>0,\notag
	\eeq
	where $T_j \defeq \{k:\beta_{jk}^0\neq 0\}$ and $s_j=|T_j|=\smallO(n)$, $\delta_{T_j k}=\delta_k$ if $k\in T_j$, $\delta_{T_j k}=0$ if $k\notin T_j$.
	\item[(A3)]\label{A3}
	$\|\eps_{j,\cdot}\|_{q,\varsigma}<\infty$ and $\|X_{jk,\cdot}\|_{q,\varsigma}<\infty$ ($q \geq 8$).
\end{itemize}

\begin{remark}\label{exo}
	We allow for overlap in the elements in $\xi_{t}$ and $\eta_{t}$, as long as the contemporaneous exogeneity condition $\E (X_{jk,t}\vps_{j,t}) = 0$ is satisfied. {For example, consider the VAR(1) model: $Y_t = AY_{t-1}+ \vps_t$, with $Y_t,\vps_t\in{\R}^J$, and suppose that $Y_t$ admits the representation $Y_t = \sum^{\infty}_{l=0} A^l \vps_{t-l}$ with $\vps_{t-l}$ as measurable functions of $\xi_{-\infty},\ldots, \xi_{t-l}$.  Thus $X_{jk,t} = g_{jk}(\ldots, \xi_{t-1}) = \sum^\infty_{l=0} [A^l]_k\vps_{t-1-l}$, where $[A^l]_k $ is the $k$th row of the matrix $A^l$, $k=1,\ldots,J$. In this case no serial correlation in the innovations $\vps_{t}$'s would be sufficient for  $\E(X_{jk,t}\vps_{j,t}) = 0$.}
\end{remark}

\begin{remark}
	We show in Theorem \ref{ratere} (see the supplementary materials) that the RE \hyperref[A2]{(A2)} and RSE \hyperref[A5]{(A5)} conditions can be implied by assumptions on the corresponding population variance-covariance matrix. This illustrates the feasibility of the RE/RSE assumption.
\end{remark}


\begin{lemma}[Prediction Performance Bound of Single Equation LASSO]\label{betabound}
	Suppose \hyperref[A1]{(A1)} and \hyperref[A2]{(A2)} (with $\bar c=\frac{c+1}{c-1}, c>1$), under the exact sparsity assumption \eqref{sparse} and given the event $\lambda_j \geq 2c\sqrt{n}\underset{1\leq k \leq K}{\max}|S_{jk}{/\Psi_{jk}}|$ and another event which RE holds, then with probability $1-\smallO(1)$, $\tilde\beta_j$ obtained from \eqref{beta.tilde} satisfy
	\beq
	|\tilde \beta_j - \beta_{j}^0|_{j, pr} \leq (1+1/c)\frac{\lambda_j\sqrt{s_j}}{n\kappa_j(\oc)}{\max_{1\leq k \leq K}\Psi_{jk}}.
	\eeq
	In addition, if \hyperref[A2]{(A2)} (with $2\bar c$) holds, then with probability $1-\smallO(1)$,
	\begin{equation}
	|\tilde \beta_j - \beta_{j}^0|_1 \leq \frac{(1+2\bar{c})\sqrt{s_j}}{\kappa_j(2\bar{c})}| \tilde \beta_j - \beta_{j}^0|_{j, pr}.
	\end{equation}
\end{lemma}

{Lemma \ref{betabound} follows Theorem 1 of \cite{belloni2009least}. As the proof is built on inequalities and for the case of dependent data \hyperref[A1]{(A1)} they remain unchanged, we omit the detailed proof here.} To further characterize the rate of IC, we provide a tail probability for $2 c\sqrt{n} \underset{1\leq k \leq K}{\max}|S_{jk}{/\Psi_{jk}}|$ under the moment assumption \hyperref[A3]{(A3)}. 
In particular, the rate depends on the dependence adjusted norm $\|X_{jk,\cdot}\vps_{j,\cdot}\|_{q,\varsigma}$.

\begin{theorem}\label{gammabound}
	Under \hyperref[A1]{(A1)} and \hyperref[A3]{(A3)}, we have
	\begin{align}\label{prob}
	\P(2c\sqrt{n}\max_{1\leq k \leq K}|S_{jk}{/\Psi_{jk}}|\geq r) \leq  &C_1\varpi_nnr^{-q}\sum_{k=1}^K\frac{\|X_{jk,\cdot}\vps_{j,\cdot}\|^q_{q,\varsigma}}{\Psi_{jk}^q}+C_2\sum_{k=1}^K \exp\Big(\frac{-C_3r^2\Psi_{jk}^2}{n\|X_{jk,\cdot}\vps_{j,\cdot}\|^2_{2,\varsigma}}\Big),
	\end{align}
	where for $\varsigma > 1/2-1/q$ (weak dependence case), $\varpi_n = 1$; for $\varsigma< 1/2-1/q$  (strong dependence case), $\varpi_n = n^{q/2-1- \varsigma q}$. $C_1,C_2,C_3$ are constants depending on $q$ and $\varsigma$. 
\end{theorem}

\begin{remark}\label{var}
	It can be seen in Theorem \ref{gammabound} that the rate of the dependence adjusted norm $\|X_{jk,\cdot}\vps_{j,\cdot}\|_{q,\varsigma}$ plays an important role in the tail probability for $2c\sqrt{n} \underset{1\leq k \leq K}{\max}|S_{jk}{/\Psi_{jk}}|$. Here we discuss the rate under some special cases.
	\begin{itemize}
		{\item[1.] \textbf{VAR(1)} (Example \ref{examp4}, continued): Consider the VAR(1) model given by $Y_t=AY_{t-1}+\vps_t$, where $Y_t,\vps_t \in {\R}^J$, and $\vps_t\sim\mbox{i.i.d.}\N(0, \Sigma)$. In this case $X_{jk,t}=Y_{j,t-1}$ and $K=J$. Suppose there exists a stationary representation of the model as $Y_t = \sum^{\infty}_{l=0} A^{l}\vps_{t-l}$. Then we have $\|X_{jk,t}\vps_{j,t} - X_{jk,t}^\ast \vps_{j,t}^\ast\|_q = \|Y_{j,t-1}\vps_{j,t} - Y_{j,t-1}^\ast \vps_{j,t}\|_q = \|[A^{t-1}]_j(\vps_0-\vps_0^\ast)\vps_{j,t}\|_q\leq2|[A^{t-1}]_j|_1\mu_q^2$, where $\mu_q\defeq\max_j\|\vps_{j,t}\|_q$ and $[A^{t-1}]_j$ is the $j$th row of the matrix $A^{t-1}$. Assume $\max_j|[A^{t}]_j|_1 \leq |c|^{t}$ with $|c|<1$ (a geometric decay rate). It follows that $\|X_{jk,\cdot}\vps_{j,\cdot}\|_{q,\varsigma}=\frac{2\mu_q^2}{1-|c|}\sup_{m\geq 0}(m+1)^\varsigma\sum_{t=m}^{\infty}|c|^{t-1}\leq (C/|c|)\vee \{C(m^*+1)|c|^{m^*-1}\}$, where $m^*=(-\varsigma/\log|c|-1)\vee 0$ and $C>0$ depends on $\mu_q$.
			Moreover, to justify the geometric decay rate, we consider the example of Network Autoregressive (NAR) model as in \citet{zhu2017network} with $A=\rho W$, where $W$ is a row-normalized adjacency matrix which is pre-specified to indicate the social network connectedness and $\rho$ is the network parameter suggesting the strength of the network effects. In that case, assuming a geometric decay rate {$\max_j|[A^{t}]_j|_1 \leq |c|^{t}$ with $|c|<1$} again gives similar results.}
		\item[2.] \textbf{Spatial autoregressive structure in $\vps_t$}: Consider the model $Y_{j,t}=X_{j,t}^\top\beta_j +\vps_{j,t}$, with $\vps_t=\rho W\vps_t +\eta_t$, where $W$ is a spatial weight matrix, $\eta_t$ are i.i.d. and have finite $q$th moments $\mu_q^\eta\defeq\max_j\|\eta_{j,t}\|_q$. For simplicity, here we assume $X_{j,t}$ and $\vps_{j,t}$ are independent. Suppose there exists a stationary representation of the error process given by $\vps_t = \sum^{\infty}_{l=0} \rho^{l}W^{l}\eta_{t-l}$. Then we have $\|X_{jk,t}\vps_{j,t} - X_{jk,t}^\ast \vps_{j,t}^\ast\|_q \leq \|(X_{jk,t} - X_{jk,t}^\ast)\vps_{j,t}\|_q + \|X_{jk,t}(\vps_{j,t} - \vps_{j,t}^\ast)\|_q \leq \|X_{jk,t} - X_{jk,t}^\ast\|_q \|\vps_{j,t}\|_q + \|X_{jk,t}\|_q\|[\rho^tW^t]_j(\eta_0-\eta^\ast_0)\|_q \leq |[(\mathbf{I}-\rho W)^{-1}]_j|_1\mu^\eta_q\|X_{jk,t} - X_{jk,t}^\ast\|_q + 2|[\rho^tW^t]_j|_1\mu^\eta_q\|X_{jk,t}\|_q$. Assume $\max_j|[\rho^tW^t]_j|_1 \leq |c|^{t}$ with $|c|<1$. It follows that $\|X_{jk,\cdot}\vps_{j,\cdot}\|_{q,\varsigma}\leq C_1 \|X_{jk,\cdot}\|_{q,\varsigma} + C_2\sup_{m\geq 0}(m+1)^\varsigma\sum_{t=m}^{\infty}|c|^{t}\leq C_1\|X_{jk,\cdot}\|_{q,\varsigma} + C_3(m^*+1)|c|^{m^*-1}$, where $m^*=(-\varsigma/\log|c|-1)\vee 0$ and $C_1,C_2,C_3>0$ depend on $\mu_q^\eta$ and $\|X_{jk,t}\|_q$.
		\item[3.] \textbf{General linear processes}: To study more general spatial and temporal dependency, consider the model $Y_{j,t}=X_{j,t}^\top\beta_j +\vps_{j,t}$, with $\vps_t = \sum^{\infty}_{l=0}A^l \eta_{t-l}$. Again $\eta_t$ are i.i.d. and have finite $q$th moments $\mu_q^\eta\defeq\max_j\|\eta_{j,t}\|_q$. If all the $A^l$ are diagonal matrices, there is just temporal dependence, and if $A^l = 0$ for $l\geq 1$ there exists only spatial dependence. Let $a_{jk}^t \defeq [A^t]_{jk}$ be the element on the $j$th row and $k$th column of $A^t$.
		Assume $\sum_{t=0}^\infty\sum_{k}|a_{jk}^t|<\infty$, $X_{j,t}$ and $\vps_{j,t}$ to be independent. We have 
		$\|X_{jk,\cdot}\vps_{j,\cdot}\|_{q,\varsigma}\leq C_1\|X_{jk,\cdot}\|_{q,\varsigma} + C_2\sup_{m\geq0}(m+1)^\varsigma\sum_{t=m}^\infty\sum_{k}|a_{jk}^t|$, where $C_1,C_2>0$ depend on $\mu_q^\eta$ and $\|X_{jk,t}\|_q$. Moreover, we have $\|\max_{jk}(X_{jk,\cdot}\vps_{j,\cdot})\|_{q,\varsigma}\leq \|\max_{jk}X_{jk,\cdot}\|_{q,\varsigma} \|\max_j\vps_{j,\cdot}\|_{q,\varsigma}$, and particularly $\||\vps_{t}|_\infty\|_{q}\leq \|\max_j \sum_k a_{jk}^t (\eta_{k,0}-\eta_{k,0}^\ast)\|_{q}\lesssim q\|\max_{k}\max_j a_{jk}^t (\eta_{k,0}-\eta_{k,0}^\ast)\|_q+ \sqrt{q \log J}\{\sum_k \max_j (a_{jk}^t)^2 (\mu_2^\eta)^2 \}^{1/2} \lesssim q\sum_k\max_j|a_{jk}^t|\mu_{q}^\eta\,\vee\,\sqrt{q\log J}\{\sum_k \max_j (a_{jk}^t)^2\}^{1/2}\mu_{2}^\eta $, where the Rosenthal-Burkholder inequality is applied. Suppose that $\sum_{t=m}^\infty(\sum_k\max_j|a_{jk}^t|)\lesssim J(m \vee 1)^{-c}$, for some constant $c>0$. If $\varsigma< c$, we have $\|\max_j\vps_{j,\cdot}\|_{q,\varsigma} \leq C_3\sup_{m\geq 1} (m+1)^{\varsigma} (m \vee 1)^{-c} J \sqrt{\log J } \leq C_3\sup_{m\geq 1}(m+1)^{\varsigma-c} J\sqrt{\log J}$, where $C_3>0$ depends on $\mu_q^\eta$.
	\end{itemize}
	
	To summarize, if the $q$th moments are bounded by constant, the dependence adjusted norm $\|X_{jk,\cdot}\vps_{j,\cdot}\|_{q,\varsigma}$ is also bounded in the first two examples where a geometric decay rate on the coefficients is assumed; while in the case of general linear processes, it would depend on the rate of $\sum_{t=0}^\infty\sum_{k}|a_{jk}^t|$. In particular, suppose $\sum_{t=m}^\infty\sum_{k}|a_{jk}^t|\lesssim (m\vee1)^{-c}$ for $c>0$. If $c>\varsigma$, $\|X_{jk,\cdot}\vps_{j,\cdot}\|_{q,\varsigma}$ is bounded (assume $\|X_{jk,\cdot}\|_{q,\varsigma}$ is bounded).
	
\end{remark}

Under the choice (IC) $\lambda_j^0(1-\alpha)$ is given by the $(1-\alpha)$ quantile of $2 c\sqrt{n} \underset{1\leq k \leq K}{\max}|S_{jk}{/\Psi_{jk}}|$, combining the results of Lemma \ref{betabound} and Theorem \ref{gammabound} we can get the bounds for $\lambda_j^0(1-\alpha)$ and further obtain the oracle inequalities as in Corollary \ref{betabound2}.

\begin{corollary}[Bounds for $\lambda_j^0(1-\alpha)$ and Oracle Inequalities under IC]\label{betabound2}
	Under \hyperref[A1]{(A1)}-\hyperref[A3]{(A3)}, 
	given $\lambda^0_j(1-\alpha)$ satisfying
	\beq
	\lambda^0_j(1-\alpha) \lesssim \max_{1\leq k \leq K}\bigg\{\|X_{jk,\cdot}\vps_{j,\cdot}\|_{2,\varsigma}\sqrt{n\log(K/\alpha) }\vee\|X_{jk,\cdot}\vps_{j,\cdot}\|_{q,\varsigma}(n\varpi_nK/\alpha)^{1/q}\bigg\},
	\eeq
	and the exact sparsity assumption \eqref{sparse}, then $\tilde\beta_j$ obtained from \eqref{beta.tilde} under IC satisfies
	\beq
	|\tilde \beta_j - \beta_{j}^0|_{j, pr} \lesssim \frac{\sqrt{s_j} }{\kappa_j(\bar{c})}\max_{1\leq k \leq K}\Psi_{jk}\bigg\{\|X_{jk,\cdot}\vps_{j,\cdot}\|_{2,\varsigma}\sqrt{\log (K/\alpha)/n }\vee\|X_{jk,\cdot}\vps_{j,\cdot}\|_{q,\varsigma}n^{1/q-1}(\varpi_nK/\alpha)^{1/q} \bigg\},
	\eeq
	with probability $1- \alpha - \smallO(1)$, where for $\varsigma > 1/2-1/q$ (weak dependence case), $\varpi_n = 1$; for $\varsigma< 1/2-1/q$  (strong dependence case), $\varpi_n = n^{q/2-1- \varsigma q}$.
\end{corollary}

\begin{remark}
	The Nagaev type of inequality in \eqref{prob} has two terms, namely an exponential term and a polynomial term. 
	It should be noted that if the polynomial term dominates, the above bound does not allow for ultra high dimension of $K$. Basically, we only allow for a polynomial rate {$K = \bigO(n^{\tilde c})$}, and the rate of $K$ interplays with the dependence adjusted norm $\|X_{jk,\cdot}\vps_{j,\cdot}\|_{q,\varsigma}$. In particular, to make sure that the estimators are consistent (i.e. the error bounds tend to zero for sufficiently large $n$),
	{for example, we need $\tilde c< q-1-\upsilon q/2 - dq$, if there exists $q$ 
		to guarantee $\|X_{jk,\cdot}\vps_{j,\cdot}\|_{q,\varsigma}=\bigO(n^d)$ and $0<\upsilon<1$ such that $s_j=\bigO(n^\upsilon)$.} 
\end{remark}
We now discuss the case of sub-Gaussian tail or sub-exponential tail, which is mostly assumed in the literature.
\begin{remark}\label{exp.bound}
	Suppose that a stronger exponential moment condition is satisfied,
	\beq\label{exp.moment}
	\|X_{jk,\cdot}\vps_{j,\cdot}\|_{\psi_\nu, \varsigma}=\sup_{q\geq2}q^{-\nu}\|X_{jk,\cdot}\vps_{j,\cdot}\|_{q,\varsigma}<\infty,\eeq
	%
	where $\|X_{jk,\cdot}\vps_{j,\cdot}\|_{\psi_\nu,\varsigma}$ is interpreted as the dependence adjusted sub-exponential ($\nu=2$) or sub-Gaussian ($\nu=1$) norm. {Consider the special case of VAR(1). As shown above, we have $\|X_{jk,t}\vps_{j,t} - X_{jk,t}^\ast \vps_{j,t}^\ast\|_q \leq2|[A^{t-1}]_j|_1\mu_q^2$. In particular, it is known that $\mu_q\lesssim q$ for sub-exponential variables and $\mu_q\lesssim \sqrt{q}$ for sub-Gaussian variables. Let $\nu=2$ and $\nu=1$ for the two cases respectively, $\|X_{jk,\cdot}\vps_{j,\cdot}\|_{\psi_\nu, \varsigma}\lesssim (m^*+1)|c|^{m^*-1}$.} Then applying the exponential tail bounds as in Lemma \ref{lemma.exp} in the supplementary material, we arrive at the following error bounds with probability $1- \alpha - \smallO(1)$,
	\beq\label{expbound}
	|\tilde \beta_j - \beta_{j}^0|_{j, pr} \lesssim \frac{\sqrt{s_j} }{\kappa_j(\bar{c})}\max_{1\leq k \leq K}\Psi_{jk}\|X_{jk,\cdot}\vps_{j,\cdot}\|_{\psi_\nu,0}\frac{\{\log (K/\alpha)\}^{ 1/\gamma}}{\sqrt{n}},\quad \gamma  = 2/(2\nu +1),
	\eeq
	as $\lambda_j^0(1-\alpha) \lesssim \sqrt{n}(\log K)^{1/\gamma}\underset{1\leq k \leq K}{\max}\|X_{jk,\cdot}\vps_{j,\cdot}\|_{\psi_\nu,0}$. The bound \eqref{expbound} works with ultra-high dimensional rate $\exp(n^{r\gamma})$ ($r< 1$) of $K$ as only the exponential term shows in the inequality. In particular, suppose $s_j = \bigO(n^{\upsilon})$, and $\|X_{jk,\cdot}\vps_{j,\cdot}\|_{\psi_\nu,0} = \bigO(n^{d})$, then $r+ d+ \upsilon/2 <1/2$ is required to ensure the consistency.
	
	In the special case with i.i.d. data, the dependence adjusted norm would be $\|X_{jk,\cdot}\vps_{j,\cdot}\|_{q,\varsigma} \leq 2 \|X_{jk,t}\vps_{j,t}\|_{q}$, and $\|X_{jk,\cdot}\vps_{j,\cdot}\|_{\psi_\nu,0}$ will be bounded by a constant which is relevant to the corresponding tail assumptions of the moments. Compared to the standard rate for LASSO estimators such as in Theorem 1 of \citet{belloni2009least} with independent errors, our results will be the same for the case of Gaussian innovation (i.e. $\nu=0$). Moreover, for time series data, disregarding the dependency adjusted norm term, our convergence rate of prediction norm $\sqrt{s_j\log K /n}$ (given $\nu=0$) is also of the same order as the rate for stable Gaussian processes studied in \citet{basu2015regularized}.
	
\end{remark}

\subsection{Gaussian Approximation for Dependent Data}
Now we look at the validity of the choice of $Q_j(1-\alpha)$, which relies on a Gaussian approximation theorem. First we define the Kolmogorov distance between any two $K$-dim random vectors.
\begin{definition}
	Let $\bm{X} = (X_1,\cdots, X_K)^\top\in{\R}^K$, $\bm{Y} = (Y_1,\cdots, Y_K)^\top\in{\R}^K$. The Kolmogorov distance between $\bm X$ and $\bm Y$ is defined as
	\beq
	\rho(\bm{X}, \bm{Y}) = \sup_{r\geq0}\big|\P(|\bm X|_{\infty} \geq r) - \P(|\bm Y|_{\infty} \geq r)\big|.\nonumber
	\eeq
\end{definition}


{For each single equation $j$, aggregate the dependence adjusted norm over $k=1,\ldots,K$:
	\begin{equation}\label{dan}
	\||X_{j,\cdot}|_\infty\|_{q,\varsigma}\defeq\sup_{m\geq 0}(m+1)^{\varsigma} \sum^{\infty}_{t=m} \delta_{q,j,t}, \,\, \delta_{q,j,t} \defeq \| |X_{j,t}- X_{j,t}^\ast|_\infty \|_q,
	\end{equation}
	where $q\geq1$ and $\varsigma>0$. Moreover, define the following quantities
	\begin{align}
	&\Phi_{j,q,\varsigma}\defeq 2\max_{1\leq k\leq K} \|X_{jk,\cdot}\|_{q,\varsigma}\|\vps_{j,\cdot}\|_{q,\varsigma}, \,\, \Gamma_{j,q, \varsigma} \defeq  2\|\vps_{j,\cdot}\|_{q,\varsigma}\bigg(\sum_{k=1}^K \|X_{jk,\cdot}\|^{q/2}_{q,\varsigma}\bigg)^{2/q}\notag\\
	&\Theta_{j, q,\varsigma} \defeq \Gamma_{j,q,\varsigma}\wedge \big\{2\||X_{j,\cdot}|_\infty\|_{q,\varsigma}\|\vps_{j,\cdot}\|_{q,\varsigma}(\log K)^{3/2}\big\}.
	\end{align}}
{It is worth noting that the norm $ \||X_{j,\cdot}|_\infty\|_{q,\varsigma}$ is a kind of aggregated dependence adjusted norm for a vector of processes in comparison to the dependence adjusted norm for a univariate process as in Definition \ref{norm}. }



Some additional assumptions are required. Define $L_{1,j} = \{\Phi_{j,4,\varsigma}\Phi_{j,4,0} (\log K)^2\}^{1/\varsigma}$, $W_{1,j}$ $=$ $(\Phi^6_{j,6,0}+ \Phi^4_{j,8,0})\{\log(Kn)\}^7$, $W_{2,j}$ $=$ $\Phi^2_{j,4,\varsigma}\{\log(Kn)\}^4$, $W_{3,j}$ $=$ $[n^{-\varsigma} \{\log (Kn)\}^{3/2} \Theta_{j,2q,\varsigma}]^{1/(1/2-\varsigma-1/q)}$, $N_{1,j} =(n/\log K)^{q/2} \Theta_{j,2q, \varsigma}^{q}$, $N_{2,j}=n(\log K)^{-2}\Phi_{j,4,\varsigma}^{-2}$, $N_{3,j} = \{n^{1/2}(\log K)^{-1/2}\Theta^{-1}_{j,2q, \varsigma}\}^{1/(1/2-\varsigma)}$.
\begin{itemize}
	\item[(A4)]\label{A4}
	i) (weak dependency case) Given $\Theta_{j, 2q,\varsigma} < \infty$ with $q \geq 4$ and $\varsigma > 1/2 - 1/q$, then \\
	$\Theta_{j,2q, \varsigma} n^{1/q-1/2}\{\log (Kn)\}^{3/2} \to 0$ and $L_{1,j}\max(W_{1,j}, W_{2,j}) = \smallO(1) \min (N_{1,j},N_{2,j})$.\\
	ii) (strong dependency case) Given $0<\varsigma< 1/2 -1/q$, then $\Theta_{j,2q,\varsigma}(\log K)^{1/2} = \smallO(n^{\varsigma})$ and $L_{1,j}\max(W_{1,j},W_{2,j},W_{3,j}) = \smallO(1)\min(N_{2,j},N_{3,j})$.
	
\end{itemize}

The assumptions impose mild restrictions on the dependency structure of covariates and error terms. They include a wide class of potential correlation and heterogeneity (including conditional heteroscedasticity), with possible allowance of the lagged dependent variables. {Two examples of large VAR and ARCH for high-dimensional time series can be found in Appendix \ref{example} in the supplementary materials.}

{
	\begin{remark}[Admissible Dimension Rates by the Conditions for Gaussian Approximation]\label{garate}
		As discussed in \citet{ZW15gaussian}, consider the case with $\Theta_{j,2q, \varsigma}=\bigO(K^{1/q})$ and $\Phi_{j,2q,\varsigma}=\bigO(1)$, where $\varsigma>1/2-1/q$. Then $\Theta_{j,2q, \varsigma} n^{1/q-1/2}\{\log (Kn)\}^{3/2} \to 0$ becomes $K\{\log(nK)\}^{3q/2}=\smallO(n^{q/2-1})$, which implies that $L_{1,j}\max(W_{1,j}, W_{2,j}) = \smallO(1) \min (N_{1,j},N_{2,j})$. This means with \hyperref[A4]{(A4)}, the dimension $K$ has to satisfy the condition $K(\log K)^{3q/2}=\smallO(n^{q/2-1})$.
	\end{remark}
}

\begin{theorem}[Gaussian Approximation Results for Dependent Data] \label{gausappro}
	Under \hyperref[A1]{(A1)} and \hyperref[A3]{(A3)}-\hyperref[A4]{(A4)}, for each $j=1,\ldots,J$ assume that there exists a constant $c_j>0$ such that {$\underset{1\leq k\leq K}\min\operatorname{avar}(S_{jk})\geq c_j$}, then we have
	\beq\label{conv}
	\rho\big(D_{j}^{-1}S_{j\cdot}, D_{j}^{-1}Z_j\big)\rightarrow 0, \quad  \text{as } n \to \infty,
	\eeq
	where $Z_j\sim \operatorname{N} (0, \Sigma_j)$, $\Sigma_j$ is the $K \times K$ long-run variance-covariance matrix of $X_{j,t}\varepsilon_{j,t}$, and $D_{j}$ is a diagonal matrix with the square root of the diagonal elements of $\Sigma_j$, namely {$$\bigg\{\sum^{\infty}_{\ell=-\infty} \E (X_{jk,t}X_{jk,(t-\ell)}\vps_{j,t} \vps_{j,(t-\ell)} )\bigg\}^{1/2}=\sqrt{\operatorname{avar}(S_{jk})}, \text{ for }  k=1,\ldots,K.$$}
\end{theorem}
{
	\begin{remark}
		The conclusion in Theorem \ref{gausappro} can be held with stronger tail assumptions, following Theorem 5.2 in \cite{ZW15gaussian}.
\end{remark}}

Theorem \ref{gausappro} justifies the choice of $\lambda_j$ and $\tilde{Q}_j(1-\alpha)$, which leads to the following corollary:
\begin{corollary}\label{gausappro.c}
	Under the conditions of Theorem \ref{gausappro}, for each $j$ we have
	\beq
	\sup_{\alpha \in(0,1)}\big|\P\{\max_{1\leq k\leq K} 2c\sqrt{n} |S_{jk}/\Psi_{jk}| \geq  Q_{j}(1-\alpha) \} - \alpha\big| \to 0, \quad   \text{as } n \to \infty.
	\eeq
\end{corollary}
{
	It is worth noting that in practice the variance involved in the Gaussian approximation in \ref{gausappro.c} is not known; we shall discuss how we estimate the variance and also the validity of the Gaussian approximation result with an estimated variance.
	Given the realization $X_{j,1}\vps_{j,1},\ldots,X_{j,n}\vps_{j,n}$, we propose to estimate the $K\times K$ long-run variance-covariance matrix $\Sigma_j$ for $j=1,\ldots,J$ as follows, given $\E X_{j,t}\vps_{j,t}=0$, and consider:
	\begin{align}\label{variance}
	\hat \Sigma_j &= \frac{1}{b_nl_n}\sum_{i=1}^{l_n}\big(\sum_{l=(i-1)b_n+1}^{ib_n}X_{j,l}\vps_{j,l}\big)\big(\sum_{l=(i-1)b_n+1}^{ib_n}X_{j,l}\vps_{j,l}\big)^\top.
	\end{align}
	Moreover, the following corollary ensures that the Gaussian approximation results still hold if we use the estimate in \eqref{variance}.
	
	\begin{corollary}\label{gausappro.est}
		Let the conditions of Theorem \ref{gausappro} hold, and assume $\Phi_{j,2q,\varsigma}<\infty$ with $q>4$, $b_n = \bigO(n^{\eta})$ for some $0 <\eta< 1$. Let $F_{\varsigma} = n$, for $\varsigma >1-2/q$; $F_{\varsigma} = l_nb_n^{q/2-\varsigma q/2}$, for $1/2-2/q<\varsigma<1-2/q$; $F_{\varsigma} = l_n^{q/4-\varsigma q/2}b_n^{q/2-\varsigma q/2} $, for $\varsigma<1/2-2/q$. Further assume \\
		$n^{-1}\log^2 K\max\big\{n^{1/2}b_n^{1/2}\Phi^2_{j,2q,\varsigma}, n^{1/2} b_n^{1/2} \sqrt{\log K}\Phi_{j,8,\varsigma}^2,F^{2/q}_{\varsigma} \Gamma^2_{j,2q, \varsigma}K^{2/q}, \Phi_{j,2,0}\Phi_{j,2,\varsigma}v'(b_n)n/\sqrt{\log K}\big\}=\smallO(1)$, with $v'(b_n) = (b_n+1)^{-\varsigma}+ 2v_{n,2}/b_n$, $v_{n,2} = \log b_n $ (resp. $b_n^{-\varsigma+1}$ or 1) for $\varsigma = 1$ (resp. $\varsigma < 1$ or $\varsigma>1$). Then for each $j$ we have	
		\beq
		\rho\big(\hat D_{j}^{-1}S_{j\cdot}, D_{j}^{-1}Z_j\big)\rightarrow 0, \quad  \text{as } n \to \infty,
		\eeq
		where $\hat D_j=\{\diag(\hat\Sigma_j)\}^{1/2}$.
	\end{corollary}
}

{
	It should be noted that given the Gaussian approximation results in Theorem \ref{gausappro}, we can have a refined bound for $\lambda_j^0(1- \alpha)$ and also the oracle inequalities under IC.
	\begin{corollary}[Bounds for $\lambda_j^0(1-\alpha)$ and Oracle Inequalities under IC with Gaussian Approximation Results]\label{betabound4}
		Under the conditions of Theorem \ref{gausappro} together with \hyperref[A2]{(A2)},
		let $2(\log K)^{-1/2} + \rho(D_{j}^{-1}S_{j\cdot}, D_{j}^{-1}Z_j) = \smallO(\alpha)$ and $Z_{\alpha}= 2\tilde c\sqrt{n \log K}$, for $\tilde c\geq \sqrt{2} c$, where $c$ is the one in the definition of $\lambda^0_j(1-\alpha)$, then we have $\lambda^0_j(1-\alpha)$ satisfying
		\beq
		\lambda^0_j(1-\alpha) \leq Z_{\alpha},
		\eeq
		and given the exact sparsity assumption \eqref{sparse}, then $\tilde\beta_j$ obtained from \eqref{beta.tilde} under IC satisfies
		\beq
		|\tilde \beta_j - \beta_{j}^0|_{j, pr} \lesssim \frac{\sqrt{s_j} }{\kappa_j(\bar{c})}\max_{1\leq k \leq K}\Psi_{jk}\sqrt{\log K/n},
		\eeq
		with probability $1 - \alpha - \smallO(1)$.
	\end{corollary}
	We note that the allowed dimension $K$ is still of polynomial rate restricted by \hyperref[A4]{(A4)}.
	
}

\subsection{Multiplier Block Bootstrap Procedure}

In this subsection, we discuss how $\Lambda_j(1-\alpha)$ is attainable via block bootstrap. The data over $t =1,\ldots,n $ are divided into $l_n$ blocks with the same number of observations $b_n$, $n=b_nl_n$ (without loss of generality), where $b_n, l_n\in\Z$.
%

Recall that $\Lambda_j(1-\alpha)=2c\sqrt{n}q^{[B]}_{j,(1-\alpha)}$, $q^{[B]}_{j,(1-\alpha)}$ is the $(1-\alpha)$ quantile of $\underset{ 1\leq k \leq K}{\max} |Z^{[B]}_{jk}{/\Psi_{jk}}|$, where $Z^{[B]}_{jk}$ are defined as
\begin{equation}\label{mbb2}
\quad Z^{[B]}_{jk}=\frac{1}{\sqrt{n}} \sum_{i=1}^{l_n} e_{j,i}\sum_{l=(i-1)b_n+1}^{ib_n} \vps_{j,l} X_{jk,l},
\end{equation}
and $e_{j,i}$ are i.i.d. $\operatorname{N}(0,1)$ random variables independent of $X$ and $\vps$.

In fact, the above construction relies on knowing the true residuals $\vps_{j,t}$. In practice, one needs to pre-estimate them using a conservative choice of penalty levels and loadings. 
{We discuss the consistency rate of the bootstrap statistics with generated errors in the supplementary material; see Comment \ref{residuals.remark} and Theorem \ref{residuals.theorem}.}

\begin{theorem}[Validity of Multiplier Block Bootstrap Method] \label{validboot}
	Under the conditions of Theorem \ref{gausappro}, and assume $\Phi_{j,2q,\varsigma}<\infty$ with $q>4$, $b_n = \bigO(n^{\eta})$ for some $0 <\eta< 1$ {(the detailed rate is calculated in \eqref{ratebn} in the supplementary materials)}, then we have
	\beq\label{sup1}
	\sup_{\alpha \in(0,1)}\big|\P\big(\max_{1\leq k\leq K} |S_{jk}{/\Psi_{jk}}| \geq q^{[B]}_{j,(1-\alpha)} 
	\big) - \alpha\big| \to 0,\, \text{as } n \to \infty.
	\eeq 
\end{theorem}

{
	\subsection{Joint Penalty over Equations}\label{lasso.joint}
	
	Recall that the theoretical choice $\lambda^0(1-\alpha)$ is defined as the $(1-\alpha)$ quantile of $\underset{1\leq k\leq K, 1\leq j\leq J}{\max}2c\sqrt{n}|S_{jk}{/\Psi_{jk}}|$. The empirical choices of the joint penalty level can be:
	\begin{itemize}	
		\item[a)] $Q(1-\alpha)$:  the $(1-\alpha)$ quantile of $ 2 c \underset{1\leq k\leq K, 1\leq j\leq J}{\max} \sqrt{n}|Z_{jk}{/\Psi_{jk}}|$. 
		In practice, one can take an alternative choice such that $\widetilde{Q}(1-\alpha) \defeq 2 c \sqrt{n}\Phi^{-1}\{1-\alpha/(2KJ)\}$.
		\item[b)] $\Lambda(1-\alpha)\defeq 2c\sqrt{n}q_{(1-\alpha)}^{[B]}$, where $q_{(1-\alpha)}^{[B]}$ is the $(1-\alpha)$ quantile of $\underset{1\leq k\leq K, 1\leq j\leq J}{\max} |Z_{jk}^{[B]}{/\Psi_{jk}}|$.
	\end{itemize}
	
	Section \ref{app.joint} in the supplementary material provides the main theorems for joint equation estimation. In particular, the dimension along $k=1,\ldots, K$ and $j=1,\ldots,J$ will be considered together by vectorization, resulting in the dimension of $KJ$. Following the results for the single equation (where $j$ is fixed), we generalize the theorems above to multiple equations case by changing the dimension from $K$ to $KJ$; see Section \ref{app.joint} in the Appendix for more details.
	
}

\subsection{Post-Model Selection Estimation}
LASSO estimation is known to be biased especially for large coefficients. Therefore, a post-selection step helps to reduce the bias by running an OLS as a second step on the selected covariates in the first step.
In particular, we consider the 2-step OLS post-LASSO estimator:
\begin{enumerate}
	\item[i)] $\ell_1$-penalized regression (LASSO selection)
	\begin{equation}
	\breve \beta_j = \arg\min_{\beta \in  {\R}^{K}} \En (Y_{j,t} - X_{j,t}^\top\beta)^2
	+ \frac{\lambda}{n} \sum_{k=1}^{K} | \beta_{jk}| \Psi_{jk},
	\end{equation}
	where $\lambda$ is the joint penalty level.
	\item[ii)] We run the post-selection regression (OLS estimation)
	\begin{equation}
	\hat \beta_j^{[P]} = \arg\min_{\beta \in  {\R}^{K}} \{\En (Y_{j,t} - X_{j,t}^\top\beta)^2: \beta_{k}=0, k\notin\hat T_j\},
	\end{equation}
	where $\hat T_j\defeq \supp(\breve{\beta}_j)=\{k\in\{1,\ldots,K\}: \breve\beta_{jk}\neq0\}$.
\end{enumerate}

To provide the prediction performance bounds for the OLS post-LASSO estimators, we need the following restricted sparse eigenvalue (RSE) condition:
\begin{itemize}
	\item[(A5)]\label{A5}Restricted sparse eigenvalue (RSE): given $p<n$, for $\delta\in {\R}^K$, with probability $1-\smallO(1)$,
	\beq
	\tilde\kappa_j(p)^2 \defeq \min_{|\delta_{T_j^c}|_0\leq p, \delta \neq 0}\frac{|\delta|_{j,pr}^2}{|\delta|_2^2}>0,\quad \phi_j(p)\defeq\max_{|\delta_{T_j^c}|_0\leq p, \delta \neq 0}\frac{|\delta|_{j,pr}^2}{|\delta|_2^2}>0.\notag
	\eeq
\end{itemize}

Here $p$ denotes the restriction on the length of the active set of $T_j^c$. When $T_j=\emptyset$, \hyperref[A5]{(A5)} is reduced to the standard sparse eigenvalue condition. Moreover, let $\mu_j(p)\defeq\sqrt{\phi_j(p)}/\tilde\kappa_j(p)$, and denote by $\hat p_j\defeq|\hat T_j\setminus T_j|$ the number of components outside $T_j\defeq\supp(\beta^0_j)=\{k\in\{1,\ldots,K\}: \beta_{jk}^0\neq0\}$ selected by LASSO in the first step.

{The performance bounds for the OLS post-LASSO estimator are shown in Theorem \ref{post} in the supplementary materials.}

\subsection{Simultaneous Inference}\label{SI}
This subsection develops theory corresponding to Section \ref{inference}.
A key Bahadur representation which linearize the estimator for a proper application of the central limit theorem for inference is provided.

Recall that for each $j=1,\ldots,J$, the following model is considered
\begin{eqnarray}
Y_{j,t} &= &\sum^{p_j}_{k= 1} X_{jk,t}\beta^0_{jk}+\sum^{K}_{k= p_j+1}X_{jk,t}\beta^0_{jk}+\varepsilon_{j,t}, \quad \E(\vps_{j,t}X_{j,t})=0, \quad F_{\varepsilon_{j}}(0) = 1/2,\notag\\
X_{jk,t}& = &   X_{j(-k),t}^{\top}\gamma^0_{j(-k)}+ v_{jk,t},\quad \E( v_{jk,t}X_{j(-k),t}) = 0, \quad k=1,\ldots,p_j,
\end{eqnarray}
where we define $\gamma_{j(-k)}^0\defeq\arg\underset{\gamma_{j(-k)}}{\min}\E(X_{jk,t}- X_{j(-k),t}^{\top}\gamma_{j(-k)})^2$, {and let $F_{\vps_j}$ denote the distribution function of $\vps_{j,t}$.}
In this subsection, we show the validity of the joint confidence region for simultaneous inference on $H_0: \beta_{jk}^0=0, \forall(j,k)\in G$, with $|G| = \sum^J_{j=1}p_j$. In particular, for $j=1,\ldots,J$, $\beta_{jk}^0\,(k=1,\ldots,p_j)$ are the target parameters.
Theoretically, we formulate the estimation as a general $Z$-estimation problem, with the leading examples as the LAD/LS cases. Nevertheless, it can also include a more general class of loss functions.

For each $(j,k)\in G$, we define the score function as $\psi_{jk}\{Z_{j,t},\beta_{jk},h_{jk}(X_{j(-k),t})\}$, where $Z_{j,t}\defeq(Y_{j,t},X_{j,t}^\top)^\top$ and the vector-valued function $h_{jk}(\cdot)$ is a measurable map from ${\R}^{K-1}$ to ${\R}^M$ ($M$ is fixed). In particular, in our linear regression case we have $h_{jk}(X_{j(-k),t})=(X_{j(-k),t}^\top\beta_{j(-k)},X_{j(-k),t}^\top\gamma_{j(-k)})^\top$, and for the LAD regression  $\psi_{jk}\{Z_{j,t},\beta_{jk},h_{jk}(X_{j(-k),t})\}=\{1/2-\IF(Y_{j,t}\leq X_{jk,t}\beta_{jk}+X_{j(-k),t}^\top\beta_{j(-k)})\}(X_{jk,t}-X_{j(-k),t}^{\top}\gamma_{j(-k)})$.

Assume that there exists $s=s_n\geq1$ such that $|\beta_{j(-k)}^0|_0\leq s$, $|\gamma_{j(-k)}^0|_0\leq s$, 
for each $(j,k)\in G$. Moreover, we assume that the nuisance function $h^0_{jk}=(h^0_{jk,m})_{m=1}^M$ admits a sparse estimator $\hat h_{jk}=(\hat h_{jk,m})_{m=1}^M$ of the form
$$
\hat h_{jk,m}(X_{j(-k),t})=X_{j(-k),t}^\top\hat\theta_{jk,m},\quad |\hat\theta_{jk,m}|_0\leq s,\quad m=1,\ldots,M,
$$
where the sparsity level $s$ is small compared to $n$ ($s\ll n$).

The true parameter $\beta_{jk}^0$ is identified as a unique solution to the moment condition
\begin{equation}\label{mc}
\E[\psi_{jk}\{Z_{j,t},\beta^0_{jk},h^0_{jk}(X_{j(-k),t})\}]=0.
\end{equation}
However, the object $\arg\,\underset{\beta_{jk} \in \hat{\mathcal{B}}_{jk}}{\operatorname{zero}}\E_n|[\psi_{jk}\{Z_{j,t},\beta_{jk},h^0_{jk}(X_{j(-k),t})\}]|$ does not necessarily exist due to the discontinuity of the function $\psi_{jk}$. The estimator $\hat\beta_{jk}$ is obtained as a $Z$-estimator by solving the sample analogue of \eqref{mc}
\begin{equation}
\E{_{n}}[\psi_{jk}\{Z_{j,t},\hat{\beta}_{jk},\hat{h}_{jk}(X_{j(-k),t})\}] \leq \operatorname{inf}_{\beta_{jk} \in \hat{\mathcal{B}}_{jk} }|\E{_n}[\psi_{jk}\{Z_{j,t},\beta_{jk},\hat{h}_{jk}(X_{j(-k),t})\}]|+ \smallO(n^{-1/2 }g^{-1}_n),\notag
\end{equation}
where $g_n \defeq \{\log(e|G|)\}^{1/2}$ and $\hat{\mathcal{B}}_{jk}$ is defined in \hyperref[C2]{(C2)}.

We now lay out the following conditions needed in this section, which are assumed to hold uniformly 
over $(j,k)\in G$.
\begin{itemize}
	\item[(C1)]\label{C1}Orthogonality condition:
\begin{equation}\label{ortho}
\E\Big[\partial_{h}\E\{\psi_{jk}(Z_{j,t},\beta^0_{jk},h)|X_{j(-k),t}\}\big|_{h=h^0_{jk}(X_{j(-k),t})}h(X_{j(-k),t})\Big]=0,
\end{equation}
for any $h\in\mathcal{H}_{jk}\cup \{h^0_{jk}\}$, where $\mathcal{H}_{jk}$ is defined in \hyperref[C5]{(C5)}.
\item[(C2)]\label{C2}
The true parameter $\beta_{jk}^0$ satisfies \eqref{mc}. Let $\mathcal{B}_{jk}$ be a fixed and closed interval and $\hat{\mathcal{B}}_{jk}$ be a possibly stochastic interval such that with probability $1- \smallO(1)$, {$[\beta^0_{jk}\pm c_1r_n] \subset \hat{\mathcal{B}}_{jk} \subset \mathcal{B}_{jk}$, where {$r_n\defeq n^{-1/2}\{\log(a_n/\epsilon)\}^{1/2}\underset{(j,k)\in G}{\max}\|\psi^0_{jk,\cdot}\|_{2,\varsigma}+n^{-1}r_\varsigma \{\log(a_n/\epsilon)\}^{3/2}\big\|\underset{(j,k)\in G}{\max}|\psi^0_{jk,\cdot}|\big\|_{q,\varsigma}$, $r_n\lesssim\rho_n$ ($\rho_n$ is defined in \hyperref[C5]{(C5)})}, $a_n \defeq \max(JK,n,e)$,  and $\psi_{jk,t}^0\defeq\psi_{jk}\{Z_{j,t},\beta^0_{jk},h^0_{jk}(X_{j(-k),t})\}$. $r_{\varsigma} = n^{1/q}$ for $\varsigma > 1/2-1/q$ and $r_{\varsigma} = n^{1/2-\varsigma}$ for $\varsigma < 1/2-1/q$.}
\item[(C3)]\label{C3}Properties of the score function:
the map $(\beta,h)\mapsto\E\{\psi_{jk}(Z_{j,t},\beta,h)|X_{j(-k),t}\}$ is twice continuously differentiable, and there exists constant $L_{n}\geq1$ such that for every $\vartheta\in\{\beta,h_1,\ldots,h_M\}$, $\E[\sup\limits_{\beta\in\mathcal{B}_{jk}}|\partial_{\vartheta}\E\{\psi_{jk}(Z_{j,t},\beta,h^0_{jk}(X_{j(-k),t})|X_{j(-k),t}\}|^2]\leq {L_n}$. \\
Moreover, {there exist measurable functions $\ell_1(\cdot),\ell_2(\cdot)$,} constants $L_{1n},L_{2n}\geq1$, $\upsilon>0$, and a cube $\mathcal{T}_{jk}(X_{j(-k),t})=\times_{m=1}^M\mathcal{T}_{jk,m}(X_{j(-k),t})$ in ${\R}^M$ with center $h^0_{jk}(X_{j(-k),t})$ such that for every $\vartheta, \vartheta'\in\{\beta,h_1,\ldots,h_M\}$ we have $\sup\limits_{(\beta,h)\in\mathcal{B}_{jk}\times\mathcal{T}_{jk}(X_{j(-k),t})}|\partial_{\vartheta}\partial_{\vartheta'}\E\{\psi_{jk}(Z_{j,t},\beta,h)|X_{j(-k),t}\}|\leq\ell_1(X_{j(-k),t})$, $\E\{|\ell_1(X_{j(-k),t})|^4\}\leq L_{1n}$, and for every $\beta,\beta'\in\mathcal{B}_{jk}$, $h,h'\in\mathcal{T}_{jk}(X_{j(-k),t})$ we have
{$\E[\{\psi_{jk}(Z_{j,t},\beta,h)-\psi_{jk}(Z_{j,t},\beta',h')\}^2|X_{j(-k),t}]\leq \ell_2(X_{j(-k),t})(|\beta-\beta'|^\upsilon+|h-h'|_2^\upsilon)$, and
$\E\{|\ell_2(X_{j(-k),t})|^4\}\leq L_{2n}$.}
\item[(C4)]\label{C4}Identifiability: $2|\E[\psi_{jk}\{Z_{j,t},\beta,h^0_{jk}(X_{j(-k),t})\}]|\geq|\phi_{jk}(\beta-\beta_{jk}^0)|\wedge c_1$ holds for all $\beta\in\mathcal{B}_{jk}$, where $\phi_{jk}\defeq\partial_{\beta}\E[\psi_{jk}\{Z_{j,t},\beta^0_{jk},h^0_{jk}(X_{j(-k),t})\}]$ and $|\phi_{jk}|\geq c_1$.
\item[(C5)]\label{C5}Properties of the nuisance function: with probability $1-\smallO(1)$, $\hat h_{jk}\in\mathcal{H}_{jk}$, where $\mathcal{H}_{jk}=\times_{m=1}^M\mathcal{H}_{jk,m}$, with each $\mathcal{H}_{jk,m}$ being the class of functions $\tilde h_{jk,m}: X_{j(-k),t}\to\R$ of the form $\tilde h_{jk,m}(X_{j(-k),t})=X_{j(-k),t}^\top\theta_{jk,m}$, $|\theta_{jk,m}|_0\leq s$, $\tilde h_{jk,m}\in\mathcal{T}_{jk,m}$. {There exists sequence of constants $\rho_n\downarrow0$ such that $\E[\{\tilde h_{jk,m}(X_{j(-k),t})-h^0_{jk,m}(X_{j(-k),t})\}^2]\lesssim \rho^2_n$.}
%
\item[(C6)]\label{C6}
The class of functions $\mathcal{F}_{jk}=\{z\mapsto\psi_{jk}\{z,\beta,\tilde h(x_{j(-k)})\}: \beta\in\mathcal{B}_{jk},\tilde h\in\mathcal{H}_{jk}\cup \{h^0_{jk}\}\}$ ($z$ is a random vector taking values in a Borel subset of a Euclidean space which contains the vectors $x_{j(-k)}$ as subvectors) is pointwise measurable and {satisfies the entropy condition $\operatorname{ent}(\epsilon, \mathcal{F}_{jk})\leq Cs\log(a_n/\epsilon)$ for all $0<\epsilon\leq1$.} It also has measurable envelope $F_{jk}\geq\underset{f\in\mathcal{F}_{jk}}{\sup}|f|$, such that $F=\underset{(j,k)\in G}{\max}F_{jk}$ satisfies $\E\{F^q(z)\}<C$ for some $q\geq4$.
\item[(C7)]\label{C7}
The second-order moments of scores are bounded away from zero: $\omega_{jk} =
\E\{(\frac{1}{\sqrt{n}}\sum_{t=1}^n\psi_{jk,t}^0)^2\}\geq c_1$.
\item[(C8)]\label{C8}
Dimension growth rates: 
$\rho_{n,\upsilon}(L_{2n}s\log a_n)^{1/2}+ n^{-1/2}r_{\varsigma}(s \log a_n)^{3/2}\|F(z_{t})\|_q + \rho_n^2 n^{1/2}=\smallO(g_n^{-1})$. In particular, for the mean regression case $\rho_{n,\upsilon} = \rho_n s$ and $\rho_{n,\upsilon} = \rho_n^{1/2}$ for the median regression case. {$n^{-1/2}\{s(\log a_n/\epsilon)\}^{1/2}\underset{f\in\mathcal F'}{\max}\|f(z_t)\|_{2}+ n^{-1}r_\varsigma\{s(\log a_n/\epsilon)\}^{3/2}\|\bar F'(z_t)\|_q=\bigO(\rho_n)$.} 
{$\mathcal{F}'=\{z\mapsto\psi_{jk}\{z,\beta,\tilde h(x_{j(-k)})\}: (j,k)\in G, \beta\in\mathcal{B}_{jk},\tilde h\in\mathcal{H}_{jk}\cup \{h^0_{jk}\}\}$ with $\bar F'=\underset{f\in\mathcal F'}{\sup}|f|$.}
\item[(C9)]\label{C9}
{
Let $B^h_\Phi = \underset{m\in\{1,2\}}{\max}\Phi^h_{m,2,\varsigma}$, $B_{\Omega}^h = \underset{m\in\{1,2\}}{\max}\Omega^h_{m,q,\varsigma}$, $B^{'h}_{\Phi} = \underset{m\in\{1,2\}}{\max} \Phi^{'h}_{m,2,\varsigma}$, and $B^{'h}_{\Omega} = \underset{m\in\{1,2\}}{\max} \Omega^{'h}_{m,q,\varsigma}$ (see \eqref{Phi.h}, \eqref{Phi.beta} and \eqref{Phi'} in the supplementary for the definitions of $\Phi^h_{m,2,\varsigma}$, $\Omega^h_{m,q,\varsigma}$, $\Phi^\beta_{2,\varsigma}$, $\Omega^\beta_{q,\varsigma}$, $\Phi^{'h}_{m,2,\varsigma}$, $\Omega^{'h}_{m,q,\varsigma}$, $\Phi^{'\beta}_{2,\varsigma}$, $\Omega^{'\beta}_{q,\varsigma}$). The following restrictions are assumed:
$$s \rho_n(\log a_n)^{1/2}B^h_{\Phi} +n^{-1/2} r_{\varsigma}\rho_n s^{2} (\log a_n)^{3/2}B^h_{\Omega} = \smallO(g_n^{-1}),$$
$$\rho_n(s\log a_n)^{1/2}\Phi^{\beta}_{2,\varsigma}+ n^{-1/2}r_{\varsigma}\rho_n(s \log a_n)^{3/2} \Omega^{\beta}_{q,\varsigma} = \smallO(g_n^{-1}),$$
$$B_{\Phi}^{'h}\rho_n s^{1/2} = \bigO(\max_{f\in\mathcal F'}\|f(z_t)\|_{2}), \, B_{\Omega}^{'h}\rho_ns^{1/2} = \bigO( \|\bar F'(z_t)\|_q),$$
$$\Phi^{'\beta}_{2,\varsigma}\rho_n = \bigO(\max_{f\in\mathcal F'}\|f(z_t)\|_{2}), \, \Omega^{'\beta}_{q,\varsigma}\rho_n = \bigO(\|\bar F'(z_t)\|_q).$$
}

\item[(C9')]\label{C9e}
{Consider the stronger exponential moment condition as in \eqref{exp.moment} and corresponding to \hyperref[C5]{(C5)}, assume that $\E[\{\tilde h_{jk,m}(X_{j(-k),t})-h^0_{jk,m}(X_{j(-k),t})\}^2]\lesssim (\rho^{e}_n)^2$. Recall the definitions of $\Phi^{h}_{m,\psi_\nu,0}$, $\Phi^{\beta}_{\psi_\nu,0}$, $\Phi^{'h}_{m,\psi_\nu,0}$, $\Phi^{'\beta}_{\psi_\nu,0}$ in \eqref{Phi.exp} and \eqref{Phi.exp'} in the supplementary. The following restrictions are assumed:
$${n^{-1/2} \{(\log a_n/\epsilon)\}^{1/\gamma}\max_{(j,k)\in G} \|\psi_{jk,\cdot}^0\|_{\psi_{\nu},0}\lesssim r_n,}$$
$$(s\log a_n)^{1/\gamma}\big[\rho_{n,\upsilon}^e \vee \rho_n^e \{(s^{1/2}\underset{m\in\{1,2\}}{\max}\Phi^{h}_{m,\psi_{\nu},0}) \vee \Phi^{\beta}_{\psi_{\nu},0}\}\big] = \smallO(g_n^{-1}),$$
$${n^{-1/2}\{s(\log a_n/\epsilon)\}^{1/\gamma}\max_{f\in\mathcal F'}\|f(z_\cdot)\|_{\psi_{\nu},0}= \bigO(\rho_n^e),}$$
$$\rho_n^e \{(s^{1/2}\underset{m\in\{1,2\}}{\max}\Phi^{'h}_{m,\psi_{\nu},0}) \vee \Phi^{'\beta}_{\psi_{\nu},0}\}=\bigO(\max_{f\in\mathcal F'}\|f(z_\cdot)\|_{\psi_{\nu},0}),$$
in particular, for the mean regression case $\rho^e_{n,\upsilon}= \rho_n^{e}s $ and $\rho^e_{n,\upsilon}=\sqrt{\rho_n^{e}}$ for the median regression case.
%
%
%
}
\item[(C10)]\label{C10}
The density of error $f_{\vps_j}(\cdot)$ is continuously differentiable and both of $f_{\vps_j}(\cdot)$ and $f'_{\vps_j}(\cdot)$ are bounded from the above.

\end{itemize}

Conditions \hyperref[C1]{(C1)}-\hyperref[C4]{(C4)} and \hyperref[C7]{(C7)} assume mild restrictions on the $Z$-estimation problems. They include the LAD-based regression (used in Algorithm \hyperref[algo2]{2}) with non-smooth score function. {Conditions \hyperref[C2]{(C2)} and \hyperref[C8]{(C8)} imply that $\underset{(j,k)\in G}{\max}\|\psi^0_{jk,\cdot}\|_{2,\varsigma}\lesssim s^{1/2}\underset{f\in\mathcal F'}{\max}\|f(z_t)\|_{2}$ and $\big\|\underset{(j,k)\in G}{\max}|\psi^0_{jk,\cdot}|\big\|_{q,\varsigma}\lesssim s^{3/2} \|\bar F'(z_t)\|_q$.}
In \hyperref[C5]{(C5)}, we suppose that the nuisance parameters have estimators with good sparsity and convergence rate properties. As discussed in previous sections, given the ideal choice of the tuning parameter, the oracle inequalities provided in Corollary \ref{betabound2} {ensures that our proposed algorithms can produce the estimator of the form $| \hat \beta_{j(-k)}^{[1]} - \beta_{j(-k)}^0|_{j, pr} \lesssim_{\P} \{\sqrt{s\log(a_n/\alpha)/n}\vee n^{1/q-1}(\varpi_na_n/\alpha)^{1/q}\}\underset{1\leq k \leq K}{\max}\|X_{jk,\cdot}\vps_{j,\cdot}\|_{q,\varsigma}$, where for $\varsigma > 1/2-1/q$ (weak dependence case), $\varpi_n = 1$; for $\varsigma< 1/2-1/q$  (strong dependence case), $\varpi_n = n^{q/2-1- \varsigma q}$. 
The moments of the envelopes are assumed to be finite in \hyperref[C6]{(C6)}.}

\begin{remark}[Discussion on the dimension growth rates\label{comment.rate}]
Consider the special case of VAR(1) model. Following the discussion in Comment \ref{var} (Example \ref{examp4}, continued), given a geometric decay rate, we have $L_{2n}, B^h_{\Phi}, B^{'h}_{\Phi},\Phi_{2,\varsigma}^{\beta} ,\Phi_{2,\varsigma}^{'\beta},\underset{f\in\mathcal F'}{\max}\|f(z_t)\|_{2}, \underset{(j,k)\in G}{\max}\big\||\psi^0_{jk,\cdot}|\big\|_{2,\varsigma}\lesssim M_n$, where $M_n$ only depends on the $2q$-th moments of $\vps_t$ and $\varsigma$. Moreover, suppose these quantities are bounded by constant and let $d_n\defeq(|G|\vee J)$, we have $B^h_{\Omega},B^{'h}_{\Omega}\lesssim d_n^{1/q}(1\vee s^{1/2}\rho_n)$, $\Omega_{q,\varsigma}^{\beta},\Omega_{q,\varsigma}^{'\beta}\lesssim d_n^{1/q}s^{1/2}\rho_n$ for mean regression case, and $B^h_{\Omega},B^{'h}_{\Omega}\lesssim d_n^{3/(4q)}(1\vee s^{1/2}\rho_n)$, $\Omega_{q,\varsigma}^{\beta},\Omega_{q,\varsigma}^{'\beta}\lesssim d_n^{1/(2q)}s^{1/2}\rho_n$ for the median regression. Moreover, $\|F(z_t)\|_q, \|F'(z_t)\|_q \lesssim d_n^{1/q}(1\vee\rho_n)$, $\big\|\underset{(j,k)\in G}{\max}|\psi^0_{jk,\cdot}|\big\|_{q,\varsigma} \lesssim d_n^{1/q}(1\vee\rho_n)$. The detailed derivation of these rates can be found in the Comment \ref{omega.rate} in the supplementary. Inserting them into \hyperref[C8]{(C8)} and \hyperref[C9]{(C9)} yields $$n^{-1/2}s^2(\log a_n)^{3/2} + n^{-1}r_\varsigma s^3(\log a_n)^{5/2}d_n^{1/q} + n^{-1/2}r_\varsigma s^{3/2}(\log a_n)^{2}d_n^{1/q}=\smallO(1),$$ and $$n^{-1/4}s^{3/4}(\log a_n)^{5/4} + n^{-1/2}r_\varsigma^{1/2} s^{5/4}(\log a_n)^{7/4}d_n^{3/(8q)} + n^{-1/2}r_\varsigma s^{3/2}(\log a_n)^{2}d_n^{3/(4q)}=\smallO(1),$$ for the smooth and non-smooth cases respectively. As a result, we only allow the dimension $(|G|\vee J)$ is of polynomial order with respect to $n$ if $q$ is not tending to infinity. In particular, under the case of $\varsigma>1/2$ and $q=\infty$, the required rate reduces to $n^{-1/2}s^2(\log a_n)^{3/2} + n^{-1}s^3(\log a_n)^{5/2} + n^{-1/2}s^{3/2}(\log a_n)^{2}=\smallO(1)$ or $n^{-1/4}s^{3/4}(\log a_n)^{5/4} + n^{-1/2}s^{5/4}(\log a_n)^{7/4} + n^{-1/2}s^{3/2}(\log a_n)^{2}=\smallO(1)$, respectively.
In the ideal case where we have weak dependency, the dimension growth rates are slightly slower than the i.i.d. case as in \cite{BCK15Bio} (i.e., $s^2 \log a_n^3=\smallO(n)$ or $s^3 \log a_n^5=\smallO(n)$ for the smooth or non-smooth case, respectively), as we apply a different way to bound the dependence adjusted norm in the concentration inequality.

More generally, suppose $\max\big\{L_{2n}, B^h_{\Phi}, B^{'h}_{\Phi},\Phi_{2,\varsigma}^{\beta} ,\Phi_{2,\varsigma}^{'\beta},\underset{f\in\mathcal F'}{\max}\|f(z_t)\|_{2}, \underset{(j,k)\in G}{\max}\big\||\psi^0_{jk,\cdot}|\big\|_{2,\varsigma}\big\}=\bigO(n^{k_1})$, and $\max\big\{B^h_{\Omega},B^{'h}_{\Omega}, \Omega_{q,\varsigma}^{\beta},\Omega_{q,\varsigma}^{'\beta}, \|F(z_t)\|_q, \|F'(z_t)\|_q,\big\|\underset{(j,k)\in G}{\max}|\psi^0_{jk,\cdot}|\big\|_{q,\varsigma} \big\}=\bigO(n^{k_2})$, with $0\leq k_1\leq k_2$, and let $s=\bigO(n^v)$, $\log a_n=\bigO(n^r)$. Then \hyperref[C8]{(C8)} and \hyperref[C9]{(C9)} imply that
$$r<\max\bigg\{\frac{1-4v-2k_1}{3},-\frac{2}{5q}+\frac{2-6v-2k_2}{5},-\frac{1}{2q}+\frac{1-3v-2k_2}{4}\bigg\},\,\text{if } \varsigma>1/2-1/q,$$
$$r<\max\bigg\{\frac{1-4v-2k_1}{3},\frac{2\varsigma+1-6v-2k_2}{5},\frac{2\varsigma-3v-2k_2}{4}\bigg\},\,\text{if } \varsigma<1/2-1/q,$$ and
$$r<\max\bigg\{\frac{1-3v-4k_1}{5},-\frac{2}{7q}+\frac{2-5v-2k_2}{7},-\frac{1}{2q}+\frac{1-3v-2k_2}{4}\bigg\},\,\text{if } \varsigma>1/2-1/q,$$
$$r<\max\bigg\{\frac{1-3v-4k_1}{3},\frac{2\varsigma+1-5v-2k_2}{7},\frac{2\varsigma-3v-2k_2}{4}\bigg\},\,\text{if } \varsigma<1/2-1/q,$$
for the smooth and non-smooth cases.

%

\end{remark}

\begin{theorem}[Uniform Bahadur Representation]\label{bahadur}
Under conditions \hyperref[A1]{(A1)}-\hyperref[A4]{(A4)} and \hyperref[C1]{(C1)}-\hyperref[C10]{(C10)}, with probability $1- \smallO(1)$, we have 
\begin{equation}
\max_{(j,k) \in G}|n^{1/2}\sigma_{jk}^{-1}(\hat{\beta}_{jk}- \beta^0_{jk})+  n^{-1/2}\sigma^{-1}_{jk}\phi_{jk}^{-1}\sum^n_{t =1} \psi_{jk,t}^0| = \smallO(g^{-1}_n),\,\text{as }n\to \infty,
\end{equation}
where $\sigma_{jk}^2 \defeq \phi_{jk}^{-2}\omega_{jk}$, $\omega_{jk} \defeq\E (\frac{1}{\sqrt{n}}\sum_{t=1}^n\psi_{jk,t}^0)^2$.
\end{theorem}
{
\begin{remark}\label{comment.rate2}
The same conclusion as in Theorem \ref{bahadur} can be drawn with assuming stronger exponential moment conditions in \eqref{exp.moment} and using \hyperref[C9e]{(C9')} instead of \hyperref[C6]{(C6)}, \hyperref[C8]{(C8)} and \hyperref[C9]{(C9)}. This is implied by Lemma \ref{ratedeltane}, \ref{max2exp} and \ref{max3exp} in the supplementary material.

We now discuss the rates implication under \hyperref[C9e]{(C9')}. 
Suppose all the dependence adjusted norms are bounded by constant with an appropriately chosen $\nu$, the restrictions in \hyperref[C9e]{(C9')} would imply $n^{-1/2}(\log a_n)^{2/\gamma+1/2}s^{2/\gamma+1}=\smallO(1)$ for the case of smooth score, and $n^{-1/4}(\log a_n)^{3/(2\gamma)}s^{3/(2\gamma)+1/2}=\smallO(1)$ for the non-smooth case, where $\gamma = 2/(2\nu+1)$. For example, when $\nu=1/2, \gamma=1$ the required rates would be $s^6\log^5a_n=\smallO(n)$ and $s^6\log^8a_n=\smallO(n)$ for the smooth and non-smooth cases respectively.


\end{remark}}

The results in Theorem \ref{bahadur} imply the asymptotic normality of the proposed estimator by Algorithm \hyperref[algo1]{1} and \hyperref[algo2]{2} by applying central limit theorems and Gaussian Approximation.

\begin{corollary}\label{asy.norm}
Under conditions \hyperref[A1]{(A1)}-\hyperref[A4]{(A4)} and \hyperref[C10]{(C10)}, 
for any $(j,k)\in G$ the estimators obtained by Algorithm \hyperref[algo1]{1} and \hyperref[algo2]{2} satisfy %
	$$\sigma^{-1}_{jk} n^{1/2}(\hat{\beta}^{[2]}_{jk}- \beta^0_{jk}) \stackrel{\mathcal{L}}{\rightarrow}  \N (0,1).$$
\end{corollary}

\begin{corollary}[Uniform-Dimensional Central Limit Theorem]\label{uninorm}
Under the same conditions as in Theorem \ref{bahadur}, 
assume that $\|\psi^0_{jk,\cdot}\|_{2,\varsigma}<\infty$, we have
$$\sigma^{-1}_{jk} n^{1/2}(\hat{\beta}_{jk}- \beta^0_{jk})  \stackrel{\mathcal{L}}{\rightarrow}  \N (0,1),$$
uniformly over $(j,k) \in G$.
\end{corollary}


Consider the vector $\widetilde\zeta_t \defeq \operatorname{vec}\{(\zeta_{jk,t})_{(j,k)\in G}\}$, $\zeta_{jk,t} \defeq -\sigma^{-1}_{jk}\phi_{j,k}^{-1}\psi_{jk,t}^0$, and define the aggregated dependence adjusted norm as follows:
\begin{equation}\label{dan3}
\|\widetilde\zeta_{\cdot}\|_{q,\varsigma}\defeq\sup_{m\geq 0}(m+1)^{\varsigma} \sum^{\infty}_{t=m}\| |\widetilde\zeta_t- \widetilde\zeta_t^\ast|_\infty \|_q,
\end{equation}
where $q\geq1$, and $\varsigma>0$. Moreover, define the following quantities
\begin{align}
&\Phi^\zeta_{q,\varsigma}\defeq \max_{(j,k)\in G} \|\zeta_{jk,\cdot}\|_{q,\varsigma}, \,\, \Gamma^\zeta_{q, \varsigma} \defeq  \bigg(\sum_{(j,k)\in G} \|\zeta_{jk,\cdot}\|^{q}_{q,\varsigma}\bigg)^{1/q}\notag,\\
&\Theta^\zeta_{q,\varsigma} \defeq \Gamma^\zeta_{q,\varsigma}\wedge \big\{\|\widetilde\zeta_{\cdot}\|_{q,\varsigma}(\log|G|)^{3/2}\big\}.
\end{align}

Define $L_1^\zeta = \{\Phi^\zeta_{2,\varsigma}\Phi^\zeta_{2,0} (\log|G|)^2\}^{1/\varsigma}$, $W_1^\zeta = \{(\Phi^\zeta_{3,0})^6+ (\Phi^\zeta_{4,0})^4\}\{\log(|G|n)\}^7$, $W_2^\zeta = (\Phi^\zeta_{2,\varsigma})^2\{\log(|G|n)\}^4$, $W_3^\zeta = [n^{-\varsigma} \{\log (|G|n)\}^{3/2} \Theta^\zeta_{q,\varsigma}]^{1/(1/2-\varsigma-1/q)}$, $N_1^\zeta =(n/\log|G|)^{q/2} (\Theta^\zeta_{q, \varsigma})^{q}$, $N_2^\zeta=n(\log|G|)^{-2}(\Phi^\zeta_{2,\varsigma})^{-2}$, $N_3^\zeta = \{n^{1/2}(\log|G|)^{-1/2}(\Theta^{\zeta}_{q, \varsigma}\})^{1/(1/2-\varsigma)}$.
\begin{itemize}
	\item[(A6)]\label{A6}
	i) (weak dependency case) Given $\Theta^\zeta_{q,\varsigma} < \infty$ with $q \geq 2$ and $\varsigma > 1/2 - 1/q$, then \\
	$\Theta^\zeta_{q, \varsigma} n^{1/q-1/2}\{\log (|G|n)\}^{3/2} \to 0$ and $L_1^\zeta\max(W_1^\zeta, W_2^\zeta) = \smallO(1) \min (N_1^\zeta,N_2^\zeta)$.\\
	ii) (strong dependency case) Given $0<\varsigma< 1/2 -1/q$, then $\Theta^\zeta_{q,\varsigma}(\log |G|)^{1/2} = \smallO(n^{\varsigma})$ and $L_1^\zeta \max(W_1^\zeta,W_2^\zeta,W_3^\zeta) = \smallO(1)\min(N_2^\zeta,N_3^\zeta)$.
\end{itemize}

\begin{corollary}[Consistency of the Estimated Confidence Interval]\label{cons.boot}
Under \hyperref[A6]{(A6)} and the same conditions as in Theorem \ref{bahadur}, for each $(j,k)\in G$ assume that there exists a constant $c>0$ such that {$\underset{(j,k)\in G}\min\operatorname{avar}\big(n^{-1/2}\sum_{t=1}^n\zeta_{jk,t}\big)\geq c$}, with probability $1-\smallO(1)$, we have 
\begin{equation}
\sup_{\alpha \in (0,1)}|\P({\beta}^0_{jk} \in \widetilde{\operatorname{CI}}_{jk}(\alpha),\,\forall (j,k) \in G) -(1-\alpha)| = \smallO(1),\,\text{ as }n\to \infty,
\end{equation}
where $\widetilde{\operatorname{CI}}_{jk}(\alpha)\defeq \[\hat{\beta}_{jk}\pm \hat\sigma_{jk}n^{-1/2} q(1-\alpha)\]$, and $q(1-\alpha)$ is the $(1-\alpha)$ quantile of the $\underset{(j,k) \in G}{\max} |\mathcal{Z}_{jk}|$, where $\mathcal{Z}_{jk}$'s are the standard normal random variables and $\hat{\sigma}_{jk}$ is a consistent estimator of $\sigma_{jk}$. 
\end{corollary}


Following Theorem \ref{bahadur}, a joint confidence region and the corresponding confidence interval for each component can be constructed via a block bootstrap method. In particular, the bootstrap statistics are defined by
$\frac{1}{\sqrt{n}} \sum_{i=1}^{l_n} e_{j,i}\sum_{l=(i-1)b_n+1}^{ib_n} \hat\zeta_{jk,l}$, where $e_{j,i}$'s are independent and identically distributed draws of standard normal random variables and are independent with respect to the data sample $(Z_{j,t})_{j=1}^J$. Recall that  $\hat\zeta_{jk,t}$ are pre-estimators with a certain range of accuracy. {More details can be found in Comment \ref{zeta.remark} in the supplementary material.}

\begin{corollary}[Validity of Multiplier Bootstrap]\label{proofboot}
Under \hyperref[A6]{(A6)} and the same conditions as in Theorem \ref{bahadur}, assume $\Phi^\zeta_{q,\varsigma}<\infty$ with $q>4$, $b_n = \bigO(n^{\eta})$ for some $0 <\eta< 1$ (the detailed rate is specified in \eqref{ratebn.G}), we have 
\beq
\sup_{\alpha \in (0,1)}|\P({\beta^0_{jk} \in \widetilde{\operatorname{CI}}^\ast_{jk}(\alpha)},\,\forall (j,k) \in G) -(1-\alpha)| = \smallO(1),\,\text{ as }n\to \infty,
\eeq
where $\widetilde{\operatorname{CI}}^\ast_{jk}(\alpha)\defeq \[\hat{\beta}_{jk}\pm \hat{\sigma}_{jk}n^{-1/2} q^\ast(1-\alpha)\]$, and $ q^\ast(1-\alpha)$ is the $(1- \alpha)$ conditional quantile of $\underset{(j,k) \in G}{\max}\frac{1}{\sqrt{n}}|\sum_{i=1}^{l_n} e_{j,i}\sum_{l=(i-1)b_n+1}^{ib_n} \hat\zeta_{jk,l}|$.
\end{corollary}

{\begin{remark}[Admissible rate of $b_n$]\label{bn.remark}
Again, consider the special case of VAR(1) with i.i.d. errors (Example \ref{examp4}, continued), with $\Theta^\zeta_{q,\varsigma}=\bigO(|G|^{1/q})$ and $\Phi^\zeta_{q,\varsigma}=\bigO(1)$, for $\varsigma>1$. Then in Corollary \ref{proofboot}, the restrictions on $b_n$ in \eqref{ratebn.G} along with \hyperref[A6]{(A6)} boil down to a set of simple admissible rates. In particular, letting $\log |G| = \bigO(n^{r})$, we need $2r<\eta<1-5r$ and $|G| (\log|G|)^{3q/2} \vee |G|^2(\log|G|)^{q}c_n^{q/2}=\smallO(n^{q/2-1})$, where $c_n^{-1}=\smallO(1)$. Note that the rate can be further improved by employing the exponential inequality under stronger tail assumptions.
\end{remark}}

\section{Simulation Study}\label{sim}
In this section, we illustrate the performance of our proposed methodology under different simulation scenarios. The first part concerns the performance of the jointly selected penalty level over equations, and the second part discusses the simultaneous inference.
\subsection{Estimation with a Jointly Selected Penalty Level}\label{sim.lambda}

Consider the system of regression equations:
\begin{equation}\label{sim.reg}
Y_{j,t} = X_{t}^\top\beta_{j}^0 + \varepsilon_{j,t},,\quad t=1,\ldots,n, \,j=1,\ldots,J,
\end{equation}
where $X_{t}\in{\R}^K$. We generate $X_{t}$ independently from $\operatorname{N}(0,\Sigma)$, where $\Sigma_{k_1,k_2}=\gamma^{|k_1-k_2|}$, $\gamma=0.5$, $\varepsilon_{j,t}\stackrel{\operatorname{i.i.d.}}{\sim}\operatorname{N}(0,1)$. The coefficient vectors $\beta_j$ are assumed to be sparse. In particular, we divide the indices $\{1,\ldots,K\}$ evenly into blocks with fixed block size 5. $\beta_{jk}^0=10$ if $k$ and $j$ belong to the same block and 0 otherwise.

We take $n=100$, $\#$ of bootstrap replications = 5000. We set $J,K=50, 100$ and $150$. The prediction norm $| \hat \beta_j - \beta^0_{j}|_{j,pr}$ and the Euclidean norm $| \hat \beta_j - \beta^0_{j}|_{2}$ ratios are presented in Table \ref{ratio 2}. The ratios measure the relative difference between the results using the penalty level determined from the equation-by-equation case and from the joint equation case ($\lambda_j$ and $\lambda$ are selected by the multiplier bootstrap procedure). In particular, a ratio smaller than $1$ indicates a better performance of using the jointly selected penalty level.

\begin{table}[H]
	\begin{center}
		\begin{tabular}{p{1cm} ccc}
			\hline\hline
			& \footnotesize{$J=K=50$} &  \footnotesize{$J=K=100$} &  \footnotesize{$J=K=150$} \\
			\cline{2-4}
			& \multicolumn{3}{c}{\small{Prediction norm}}\\
			\hline
			\footnotesize{Mean} & 0.9634 & 0.9474 & 0.9347 \\
			\footnotesize{Median} & 0.9695 & 0.9516 & 0.9371 \\
			\footnotesize{Std.} & 0.0323 & 0.0272 & 0.0254 \\
			\hline
			& \multicolumn{3}{c}{\small{Euclidean norm}}\\
			\hline
			\footnotesize{Mean} & 0.9590 & 0.9429 & 0.9286 \\
			\footnotesize{Median} & 0.9679 & 0.9468 & 0.9316 \\
			\footnotesize{Std.} & 0.0367 & 0.0292 & 0.0286 \\
			\hline\hline
		\end{tabular}
	\end{center}
	\caption{Prediction norm and Euclidean norm ratios (overall $\lambda$ relative to equation-by-equation $\lambda_j$'s, average over equations). Results (mean, median and standard deviation) are computed over $1000$ replications.}
	\label{ratio 2}
\end{table}

It is evident from Table \ref{ratio 2} that the proposed estimation procedure delivers much better performance in terms of the two measures. In particular, the superiority tends to be more evident (more than $10 \% $) with higher dimension of the covariates and more equations.


Still consider the system of regression equations as in \eqref{sim.reg}, but here we generate the data with dependency by following the Appendix D in \citet{ZW15_sup}. In particular, assume the linear process such that $X_t=\sum_{\ell=0}^{\infty}A_\ell\xi_{t-\ell}$, with $A_\ell=(\ell+1)^{-\rho-1}M_\ell$, where $M_\ell$ are independently drawn from Ginibre matrices, i.e. all the entries of $M_\ell$ are i.i.d. $\operatorname{N}(0,1)$, and in practice the sum is truncated to $\sum_{\ell=0}^{1000}$. We set $\rho$ to be 1.0 for the weaker dependence and 0.1 for the stronger dependence cases respectively. Let $\xi_{k,t}=e_{k,t}(0.8e_{k,t-1}^2+0.2)^{1/2}$ where $e_{k,t}$ are i.i.d. distributed as $t(d)/\sqrt{d/(d-2)}$ and $t(d)$ is the Student's $t$ with degree of freedom $d$ (take $d=8$ for example). $\eps_t$ are generated by following the same fashion independently.


We take $n=100$, $\#$ of bootstrap replications = 5000, $J,K=50, 100$ and $150$. 
{Based on bias-variance trade-off, several approaches were suggested to determine the optimal choice of $b_n$ for univariate case. Concerning the high-dimensional case, we propose to take the one which gives the lowest prediction norm as the optimal choice. Below we report the average prediction norm $J^{-1}\sum_{j=1}^J|\hat \beta_j - \beta^0_{j}|_{j,pr}$ with several block sizes $b_n$ under different settings and the minimal ones are in bold. }

\begin{table}[H]
	\begin{center}
		\begin{tabular}{p{1.3cm} ccccccc}
			\hline\hline
			& \multicolumn{3}{c}{\small{$\rho=0.1$ (stronger dependency)}} && \multicolumn{3}{c}{\small{$\rho=1.0$ (weaker dependency)}}\\
			\cline{2-4}\cline{6-8}
			& \footnotesize{$J=K=50$} &  \footnotesize{$J=K=100$} &  \footnotesize{$J=K=150$} && \footnotesize{$J=K=50$} &  \footnotesize{$J=K=100$} &  \footnotesize{$J=K=150$} \\
			\hline
			$b_n=2$ &  2.0721 & 2.9122 & 3.5932 & & \bf{2.0165} & 2.6270 & 3.2286 \\
			\hline
			$b_n=4$ & 2.0627 & 2.8924 & 3.5617 & & 2.0303 & \bf{2.6183} & 3.2225 \\
			\hline
			$b_n=6$ & 2.0487 & 2.9007 & 3.5235 & & 2.0834 & 2.6288 & \bf{3.2198} \\
			\hline
			$b_n=8$ & \bf{2.0388} & 2.8841 & \bf{3.5073} & & 2.2149 & 2.6502 & 3.2320 \\
			\hline
			$b_n=10$ & 2.0521 & \bf{2.8836} & 3.5268 & & 2.3576 & 2.7099 & 3.2975 \\
			\hline
			$b_n=12$ & 2.0581 & 2.9065 & 3.5687 & & 2.5592 & 2.8310 & 3.3895 \\
			\hline\hline
		\end{tabular}
	\end{center}
	\caption{The prediction norm (average over equations) using several choices of $b_n$. Results are computed over $1000$ simulations.}
	\label{bn}
\end{table}

{From Table \ref{bn}, it is apparent that a larger block size is required for the stronger dependency case. Moreover, the choice also depends on the dimensionality, which is more evident for relatively weaker dependent data. We note that when $J=K=50$, $\rho=1.0$ the ordinary multiplier bootstrap (with $b_n=1$) produces 2.1003 as the average prediction norm, therefore we suggest $b_n=2$ for this case.}

The prediction norm $| \hat \beta_j - \beta^0_{j}|_{j,pr}$ and the Euclidean norm $| \hat \beta_j - \beta^0_{j}|_{2}$ ratios {(using the optimal $b_n$ suggested in Table \ref{bn} for each case correspondingly)} are presented in Table \ref{ratio 3}. Again we report the results with the jointly estimated $\lambda$ (selected by the algorithm proposed in section 3.2 based on multiplier block bootstrap) relative to using the single equation $\lambda_j$'s. 

\begin{table}[H]
	\begin{center}
		\begin{tabular}{p{1.1cm} ccccccc}
			\hline\hline
			& \multicolumn{3}{c}{\small{$\rho=0.1$ (stronger dependency)}} && \multicolumn{3}{c}{\small{$\rho=1.0$ (weaker dependency)}}\\
			\cline{2-4}\cline{6-8}
			& \footnotesize{$J=K=50$} &  \footnotesize{$J=K=100$} &  \footnotesize{$J=K=150$} && \footnotesize{$J=K=50$} &  \footnotesize{$J=K=100$} &  \footnotesize{$J=K=150$} \\
			\cline{2-8}
			& \multicolumn{7}{c}{\small{Prediction norm}}\\
			\hline
			\footnotesize{Mean} & 0.9141 & 0.8534 & 0.8250 & & 0.9356 & 0.8786 & 0.8326 \\
			\footnotesize{Median} & 0.9165 & 0.8532 & 0.8255 & & 0.9384 & 0.8792 & 0.8330 \\
			\footnotesize{Std.} & 0.0436 & 0.0377 & 0.0326 & & 0.0380 & 0.0338 & 0.0296 \\
			\hline
			& \multicolumn{7}{c}{\small{Euclidean norm}}\\
			\hline
			\footnotesize{Mean} & 0.9017 & 0.8447 & 0.8114 & & 0.9251 & 0.8648 & 0.8154 \\
			\footnotesize{Median} & 0.9062 & 0.8453 & 0.8135 & & 0.9290 & 0.8652 & 0.8157 \\
			\footnotesize{Std.} & 0.0515 & 0.0401 & 0.0348 & & 0.0453 & 0.0368 & 0.0317 \\
			\hline\hline
		\end{tabular}
	\end{center}
	\caption{Prediction norm and Euclidean norm ratios (overall $\lambda$ relative to equation-by-equation $\lambda_j$'s, average over equations). Results (mean, median and standard deviation) are computed over $1000$ replications.}
	\label{ratio 3}
\end{table}

The results show that the coefficient estimation performance measured by both the prediction norm and the Euclidean norm is in favor of the joint penalty level with multiplier block bootstrap approach. The results are robust over different dimension cases with stronger or weaker dependency. 

\subsection{Simultaneous Inference}

In this subsection we consider the following regression model for the purpose of simultaneous inference on the parameters within a system of equations

\begin{equation}
Y_{j,t} = d_{j,t}\alpha^0_{j} + X_{t}^\top\beta^0_{j} +\varepsilon_{j,t}, \,\, d_{j,t} = X_t^\top\theta^0_j + v_{j,t},\,\,t=1,\ldots,n, \,\, j=1,\ldots,J,
\end{equation}
where $\alpha^0_j=\alpha^0$ for all $j$. Also, $\beta^0_j,\theta^0_j\in{\R}^K$ are assumed to be sparse. In particular, we divide the indices $1,\ldots,K$ evenly into blocks with a fixed block size $5$, $\beta_{jk}^0$ and $\theta_{jk}^0$ are independently drawn from $\operatorname{Unif}[0,5]$ and $\operatorname{Unif}[0,0.25]$ respectively, if $k$ and $j$ belong to the same block and $0$ otherwise. The way to generate $X_t$, $\eps_{t}$ and $v_t$ is same as the dependent data setting above.


We consider the sample size $n=100$. Our goal is to estimate and make inferences on the target variables $d_{j,t}$'s based on the procedure proposed in Section \ref{inference}. We evaluate and compare the empirical power and size performance of the confidence intervals constructed by the asymptotic distribution theory \eqref{CI.asy}, block bootstrap \eqref{CI.boot.ind} and the simultaneous confidence regions via block bootstrap \eqref{CI.boot.sim}. The bootstrap statistics are computed based on 5000 replications and we also take the optimal block size according to the numerical comparison conducted above. 
Note that the case of $\alpha^0=0$ gives the size performance under the null hypothesis, while $\alpha^0$ uniformly lies in $[0,2.5]$ and $[0,5]$ illustrate the power results.


Table \ref{reject} shows the average rejection rate of $H_0^j: \alpha_j^0=0$ over $j$ for individual (or multiple) inference and the rejection rate of $H_0: \alpha_1^0=\cdots=\alpha_J^0=0$ for simultaneous inference under different settings of $J,K$ and $\rho$. {Multiple testing procedure via step-down method (see e.g. \citet{romano2005exact,CCK13AoS}), is considered to control the false positives in evaluating the power performance.} The rejection rates are computed over $1000$ simulation samples.

\begin{table}[H]
	\begin{center}
		\begin{tabular}{p{2.1cm} ccccccc}
			\hline\hline
			& \multicolumn{3}{c}{\small{$\rho=0.1$ (stronger dependency)}} && \multicolumn{3}{c}{\small{$\rho=1.0$ (weaker dependency)}}\\
			\cline{2-4}\cline{6-8}
			& \footnotesize{$J=K=50$} &  \footnotesize{$J=K=100$} &  \footnotesize{$J=K=150$} && \footnotesize{$J=K=50$} &  \footnotesize{$J=K=100$} &  \footnotesize{$J=K=150$} \\
			\cline{2-8}
			& \multicolumn{7}{c}{\small{$\alpha^0=0$}}\\
			\hline
			\footnotesize{Ind. Asym.} &  0.0166 & 0.0126 & 0.0126 & & 0.0242 & 0.0148 & 0.0119 \\
			\hline
			\footnotesize{Ind. Boot.} &  0.0303 & 0.0202 & 0.0155 & & 0.0224 & 0.0169 & 0.0141 \\
			\hline
			\footnotesize{Simult. Boot.} & 0.0260 & 0.0473 & 0.0527 & & 0.0520 & 0.0547 & 0.0587 \\
			\hline
			& \multicolumn{7}{c}{\small{$\alpha^0\sim\operatorname{Unif}[0,2.5]$}}\\
			\hline
			\footnotesize{Ind. Asym.} &  0.8714 & 0.8558 & 0.8553 & & 0.8763 & 0.8622 & 0.8572 \\
			\hline
			\footnotesize{Ind. Boot.} & 0.8746 & 0.8573 & 0.8566 & & 0.8761 & 0.8629 & 0.8578 \\
			\hline
			\footnotesize{Mult. Boot.} & 0.8413 & 0.8027 & 0.8004 & & 0.8438 & 0.8249 & 0.8091 \\
			\hline
			& \multicolumn{7}{c}{\small{$\alpha^0\sim\operatorname{Unif}[0,5]$}}\\
			\hline
			\footnotesize{Ind. Asym.} &  0.9376 & 0.9247 & 0.9282 & & 0.9380 & 0.9319 & 0.9269 \\
			\hline
			\footnotesize{Ind. Boot.} & 0.9390 & 0.9254 & 0.9331 & & 0.9288 & 0.9325 & 0.9273 \\
			\hline
			\footnotesize{Mult. Boot.} & 0.9282 & 0.9070 & 0.9072 & & 0.9262 & 0.9182 & 0.9082 \\
			\hline\hline
		\end{tabular}
	\end{center}
	\caption{Average rejection rate of $H_0^j: \alpha_j^0=0$ over $j$ for the individual (or multiple) inference and the rejection rate of $H_0: \alpha_1^0=\cdots=\alpha_J^0=0$ for simultaneous inference under several true $\alpha^0$ values (given the significance level = 0.05).}
	\label{reject}
\end{table}

{It is shown that for individual inference our proposed individual bootstrap approach provides a closer size control to the nominal $\alpha$ and more powerful empirical rejection probabilities compared to constructing the confidence intervals by asymptotic normality in most of the cases. Moreover, the simultaneous inference outperforms the individual inference in size accuracy and in terms of the power performance, the multiple testing is relatively conservative after controlling the false positives. Overall, we observe that the results using bootstrap approach are robust over different dimension settings under either stronger or weaker dependency cases. }

\section{Empirical Analysis: Textual Sentiment Spillover Effects}\label{app}
Financial markets are driven by information, and this is a well-known phenomenon among investors. More frequent news and availability of sentiment data allows study of the impact of firm-specific investor sentiment on market behavior such as stock returns, volatility and liquidity; see \citealp{baker2006, tetlock2007}, among others. Moreover, powerful statistical tools (e.g. LASSO-type estimators) are being used to model complex relationships among individuals. For example, \citet{audrino_sentiment} analyze the influence of news on US and European companies by constructing a sparse predictive network via adaptive LASSO and related testing procedures. In this section the developed technology is applied to study textual sentiment spillover effects across individual stocks. This is different from the "equation-by-equation" analysis in \citet{audrino_sentiment}, since we build up a system of regression equations and implement the estimation and the inference of the network jointly.

\subsection{Data Source}
The empirical study in this paper is carried out based on the financial news articles published on the NASDAQ community platform from January 2, 2015 to December 29, 2015 (252 trading days). The data were gathered via a self-written web scraper to automate the downloading process. The dataset is available at the Research Data Centre (RDC), Humboldt-Universit\"at zu Berlin. Moreover, unsupervised learning approaches are employed to extract sentiment variables from the articles. Two sentiment dictionaries: the BL option lexicon \citep{BL2004text} and the LM financial sentiment dictionary \citep{LM2011text} were used in \citet{zhang2016text}. For each article $i$ (published on day $t$), the average proportion of positive/negative words using BL or LM lexica - $Pos_{j,i,t}^\text{BL}$, $Neg_{j,i,t}^\text{BL}$, $Pos_{j,i,t}^\text{LM}$, $Neg_{j,i,t}^\text{LM}$ - are considered as the text sentiment variables. Furthermore, the bullishness indicator for stock $j$ on day $t$ with the related articles $i=1,\ldots,m$ (based on a particular lexicon) is constructed by following \citet{af2004text}
\begin{equation}
B_{j,t}=\log\Big(\Big\{1+m^{-1}\sum_{i=1}^m\IF(Pos_{j,i,t}>Neg_{j,i,t})\Big\}\Big/\Big\{1+m^{-1}\sum_{i=1}^m\IF(Pos_{j,i,t}<Neg_{j,i,t})\Big\}\Big).
\end{equation}
We refer to \citet{zhang2016text} for more details about the data gathering and processing procedure.
63 individual stocks which are S\&P 500 component stocks from 9 Global Industrial Classification Standard (GICS) sectors are considered. They are traded at NSDAQ Stock Exchange or NYSE. The list of the stock symbols and the corresponding company names can be found in Table \ref{table:List_of_Firms} in Appendix \ref{table} in the supplementary materials.

The daily log returns $R_{j,t}$ and log volatilities $\log(\sigma_{j,t}^2)$ for the stocks over the same time span are taken as response variables. More precisely, the \citet{gk1980vola} range-based measure to represent the volatility level is employed:
\begin{equation}
\sigma_{j,t}^2 = 0.511 (u_{j,t}-d_{j,t})^2 - 0.019\{r_{j,t}(u_{j,t}+d_{j,t})-2u_{j,t}d_{j,t}\} - 0.383r_{j,t}^2,
\end{equation}
where $u_{j,t} = \log(P_{j,t}^H) - \log(P_{j,t}^O), d_{j,t} = \log(P_{j,t}^L) - \log(P_{j,t}^O), r_{j,t} = \log(P_{j,t}^C) - \log(P_{j,t}^O)$, with $P_{j,t}^H, P_{j,t}^L, $, $P_{j,t}^O$, and $P_{j,t}^C$ denote the highest, lowest, opening and closing prices, respectively. In addition, the S\&P 500 index returns and Chicago Board Options Exchange volatility index (VIX) are included as the state variables. The financial time series data were originally obtained from Datastream, and GICS sector information was found at Compustat.

\subsection{Model Setting and Results}
We now construct a network model to detect the spillover effects from sentiment variables to financial variables by
\begin{align}\label{ret}
r_{j,t} &= c_j + B_{t}^\top\beta_j + z_{t}^\top\gamma_j + r_{j,t-1}\delta_j + \varepsilon_{j,t},\notag\\
\log\sigma^2_{j,t} &= c_j + B_{t}^\top\beta_j + z_{t}^\top\gamma_j + \log\sigma^2_{j,t-1}\delta_j + \varepsilon_{j,t},
\end{align}
where $j=1,\ldots,J$ indicate the stock symbols, $B_{t}=(B_{1,t},\ldots,B_{J,t})^\top$ and $z_t$ includes the state variables. 

It is of interest to make inferences on the parameters $\beta_j\in {\R}^J$, $j=1,\ldots J$. Following the framework introduced in Section \ref{inference}, an estimation procedure with three steps needs to be implemented.

\begin{enumerate}
	\item[S1] For each $j$, run LASSO on \eqref{ret} and keep the estimator $\hat{\beta}^{[1]}_{j(-j)}$, $\hat{\gamma}^{[1]}_j$, $\hat{\delta}^{[1]}_j$ and $\hat{c}^{[1]}_j$.
	\item[S2] For each $j$, run LASSO on $B_{j,t} = (B_{-j,t}^\top,z_{t}^\top,r_{j,t-1})^\top\theta_{j} + v_{j,t}$ to model the dependence among sentiment variables. In particular, we propose to take the joint penalty level obtained via block multiplier bootstrap (discussed in Section \ref{mblambda}) for this regression system. Keep the residuals as $\hat{v}_{j,t} = B_{j,t} - (B_{-j,t}^\top,z_{t}^\top,r_{j,t-1})^\top\hat{\theta}_{j}$.
	\item[S3] For each $(j,k)$, run IV regression of $r_{j,t} - \hat{c}^{[1]}_j - B_{-j,t}^{\top}\hat{\beta}^{[1]}_{j(-j)} - z_{t}^\top\hat{\gamma}^{[1]}_j- r_{j,t-1}\hat{\delta}^{[1]}_j$ on $B_{k,t}$ using $\hat{v}_{k,t}$ as an instrument variable. Then we obtain the final estimator $\hat{\beta}^{[2]}_{jk}$.
\end{enumerate}

If for stock $j$, the sentiment variable of firm $k$ is selected into the active set after the individual significance test i.e., the null hypothesis $H_0^{jk}: \beta_{jk}=0$ is rejected under the block multiplier bootstrap procedure {(as discussed in Section \ref{sim.lambda} we pre-determine $b_n=5$ by choosing the one gives the lowest prediction norm in the LASSO estimation in S1 on a grid search)}, then we put a directional edge from $k$ to $j$. As a result, we achieve a $0-1$ adjacency matrix describing the dependency network from sentiment variable to financial variable. Note that the diagonal elements in the matrix show the self-effect of stocks.

The graphical network for stock returns and volatility modelled by \eqref{ret} based on BL and LM lexica (from 01/02/15 to 12/29/15) is depicted in Figures \ref{graph}-\ref{graph2}.

\begin{figure}[H]\begin{centering}
		\begin{subfigure}[t]{1\textwidth}
			\centering
			\includegraphics[scale=0.5]{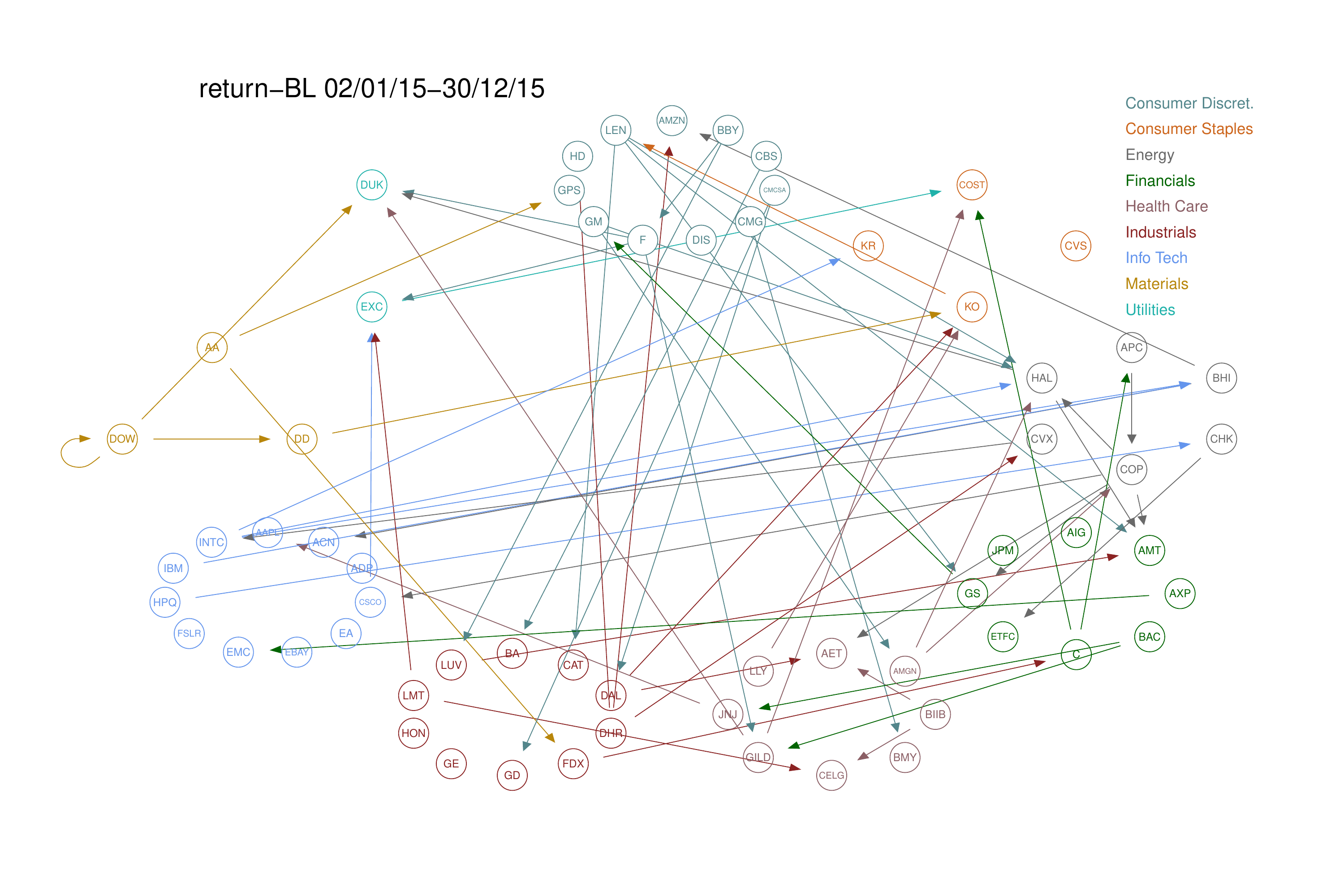}
		\end{subfigure}
		\vspace{-1.2cm}
		\begin{subfigure}[t]{1\textwidth}
			\centering
			\includegraphics[scale=0.5]{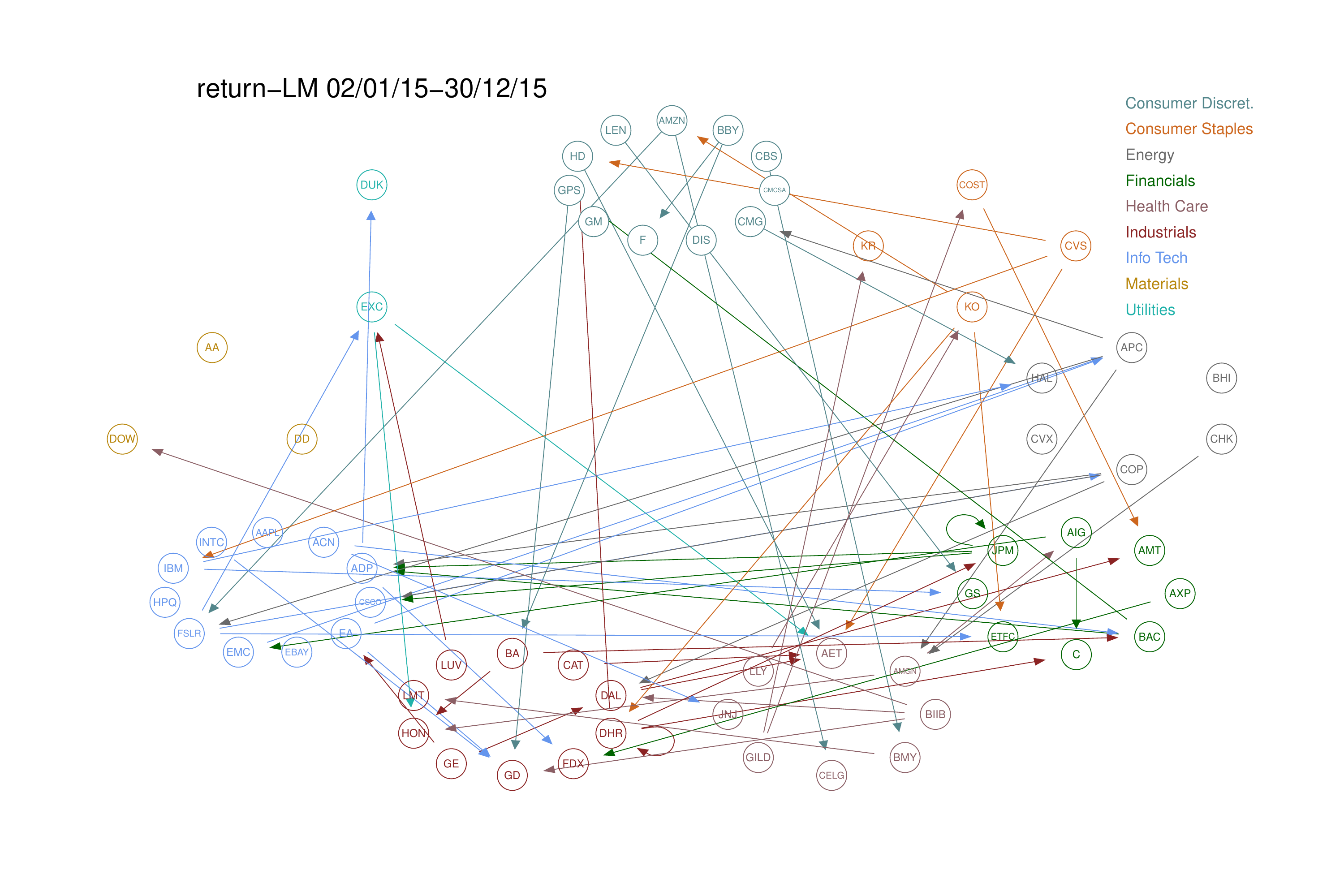}
		\end{subfigure}		
		\caption{The dependency network among individual stocks from sentiment variables to return.}
		\label{graph}
	\end{centering}
\end{figure}

\begin{figure}[H]\begin{centering}
		\begin{subfigure}[t]{1\textwidth}
			\centering
			\includegraphics[scale=0.5]{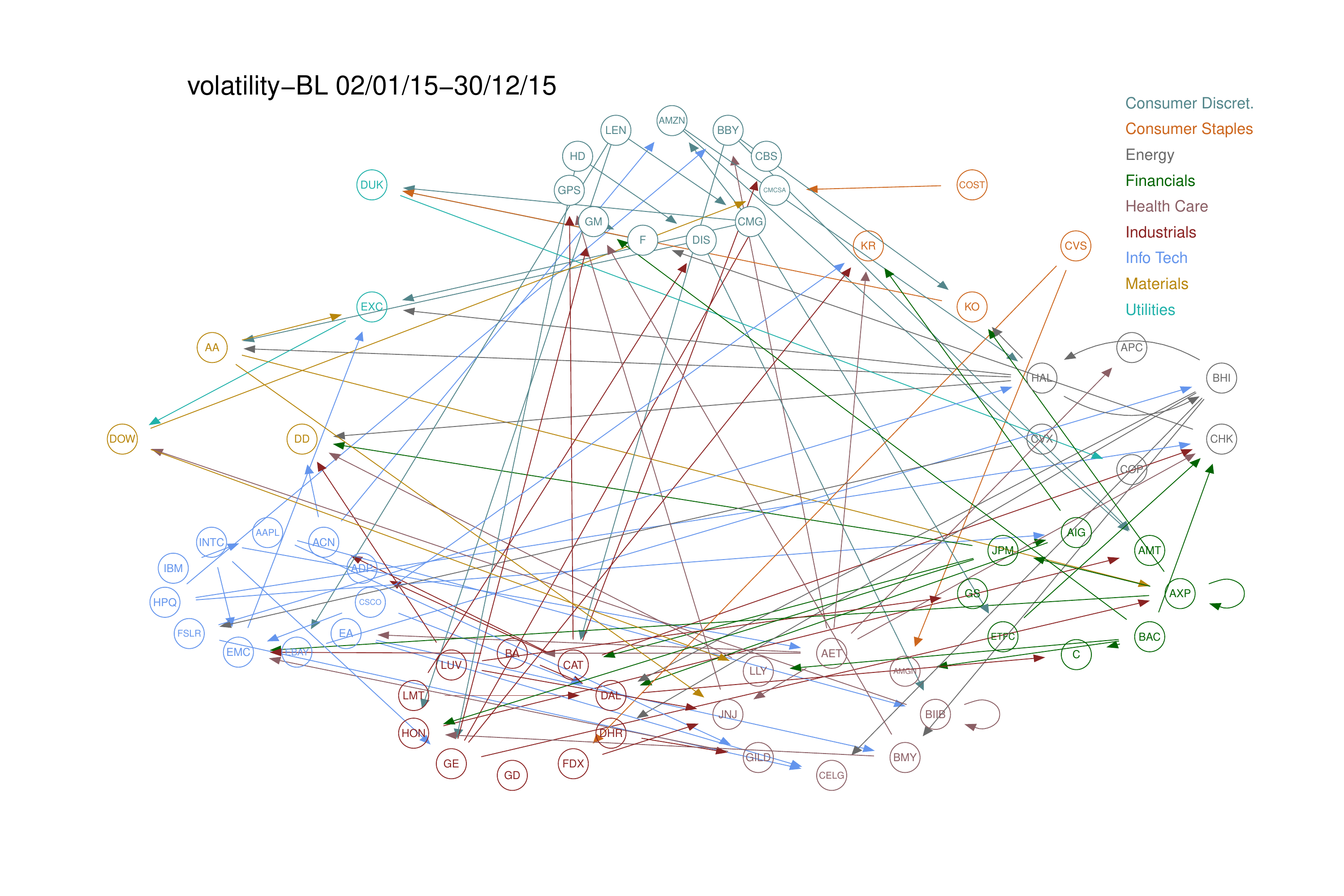}
		\end{subfigure}
		\vspace{-1.2cm}
		\begin{subfigure}[t]{1\textwidth}
			\centering
			\includegraphics[scale=0.5]{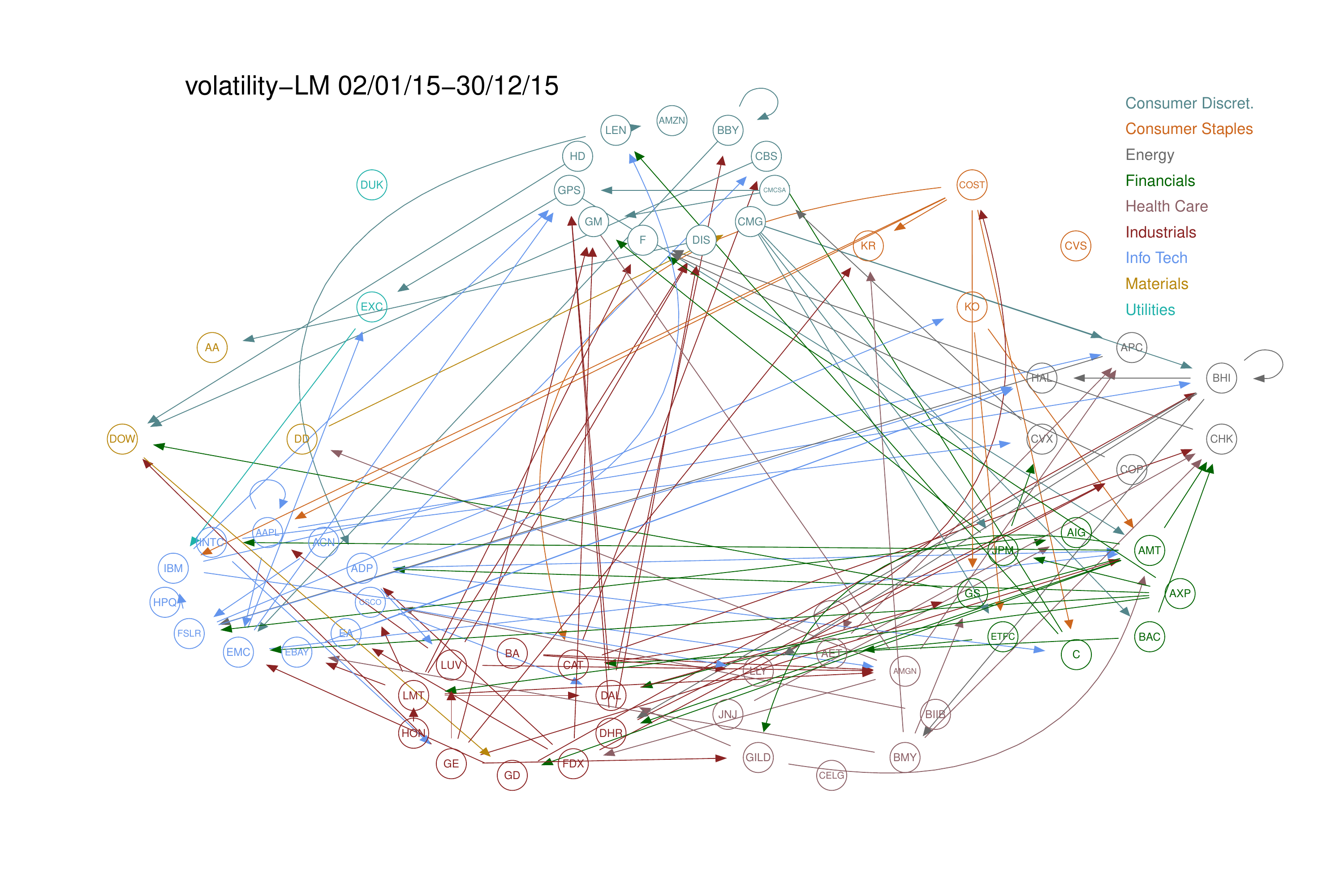}
		\end{subfigure}
		\caption{The dependency network among individual stocks from sentiment variables to volatility.}
		\label{graph2}
	\end{centering}
\end{figure}

Figures \ref{graph}-\ref{graph2} depict the dependency networks among individual stocks. Given that the time series of returns and volatility are scaled and centered before implementing the estimation procedure, we find even denser spillover effects in the volatility analysis. This indicates the stock volatility is more sensitive to sentiment than returns. Moreover, the relationships between sectors are also of interest. The simultaneous confidence region constructed via the bootstrap approach introduced in Section \ref{scr} may help us to detect whether the sentiment information from one sector has joint influence on the returns of the stocks in another sector. In particular, we look at the null hypothesis: $H_0^{S_1,S_2}:\, \beta_{jk}=0,\,\,\forall j\in S_1,\,k\in S_2$, where $S_1$ and $S_2$ represent two groups of stocks that belong to two sectors, respectively. The conclusion that the sentiment from sector $S_2$ has a joint effect on the returns or volatility of sector $S_1$ can be drawn if the null hypothesis is rejected with the simultaneous confidence region \eqref{CI.boot.sim} under the significance level = 0.05.

\begin{figure}[H]\begin{centering}
		\begin{subfigure}[t]{0.45\textwidth}
			\centering
			\includegraphics[scale=0.47]{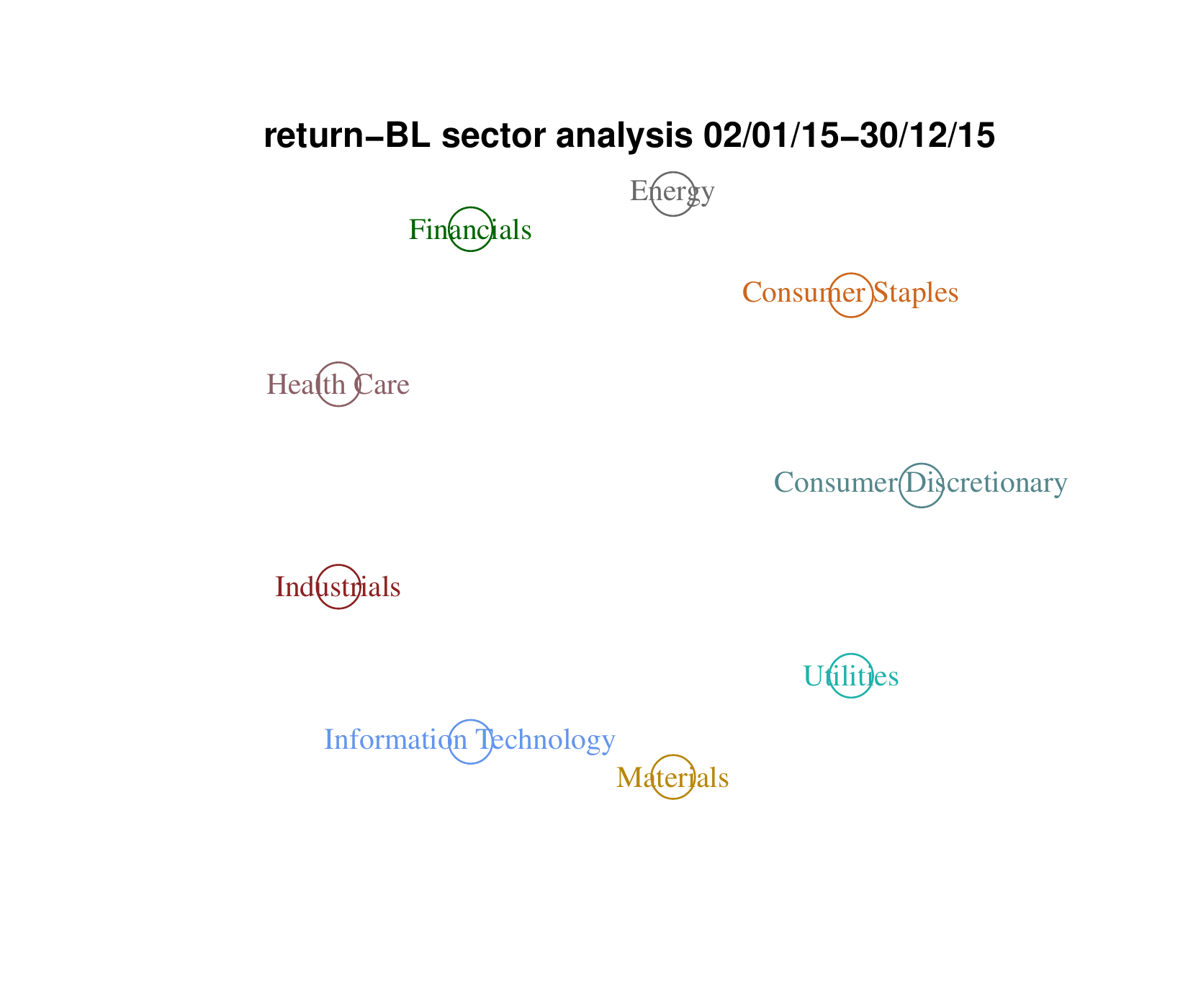}
		\end{subfigure}
		\begin{subfigure}[t]{0.54\textwidth}
			\centering
			\includegraphics[scale=0.47]{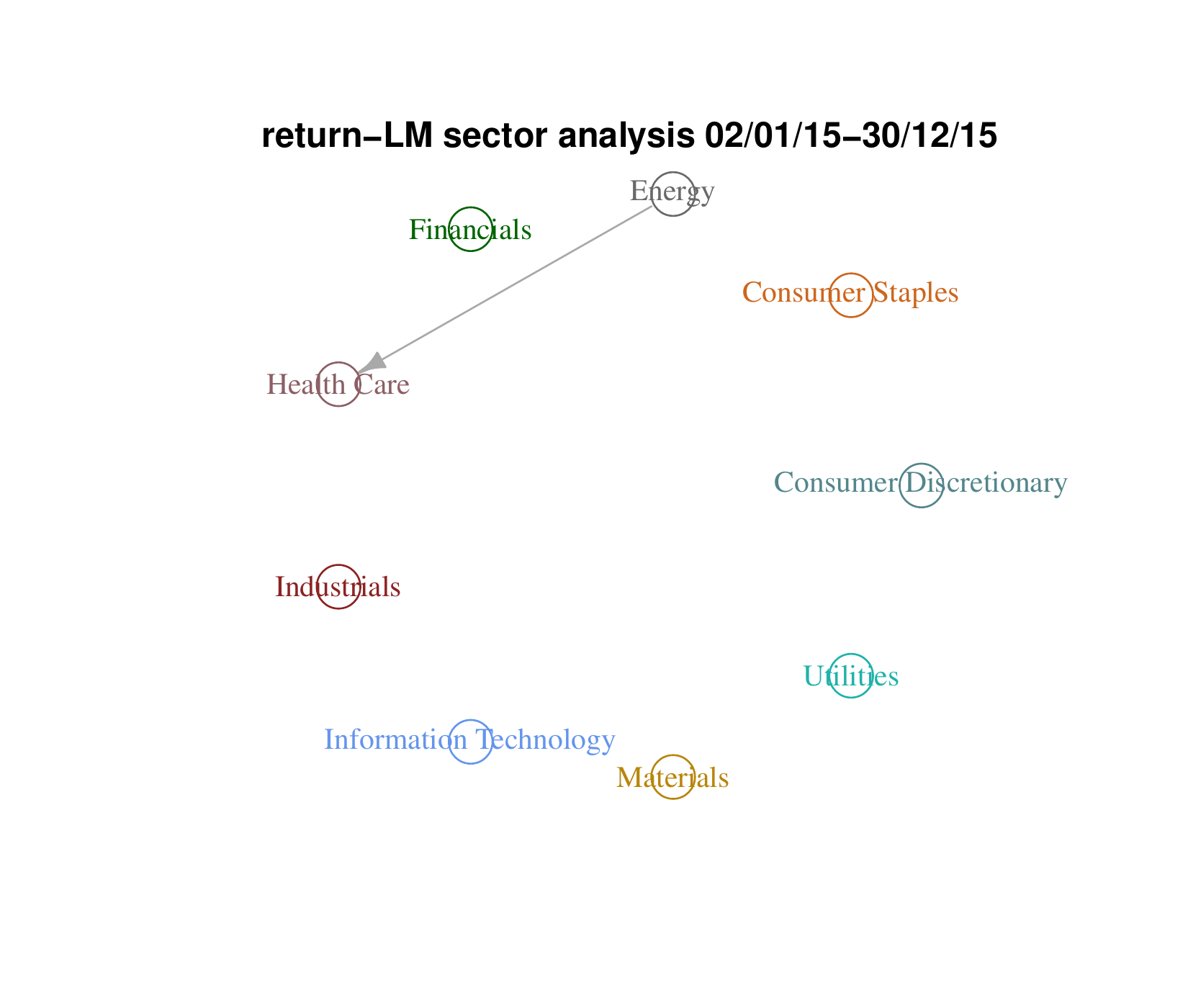}
		\end{subfigure}
		\vspace{-1cm}
		\begin{subfigure}[t]{0.45\textwidth}
			\centering
			\includegraphics[scale=0.47]{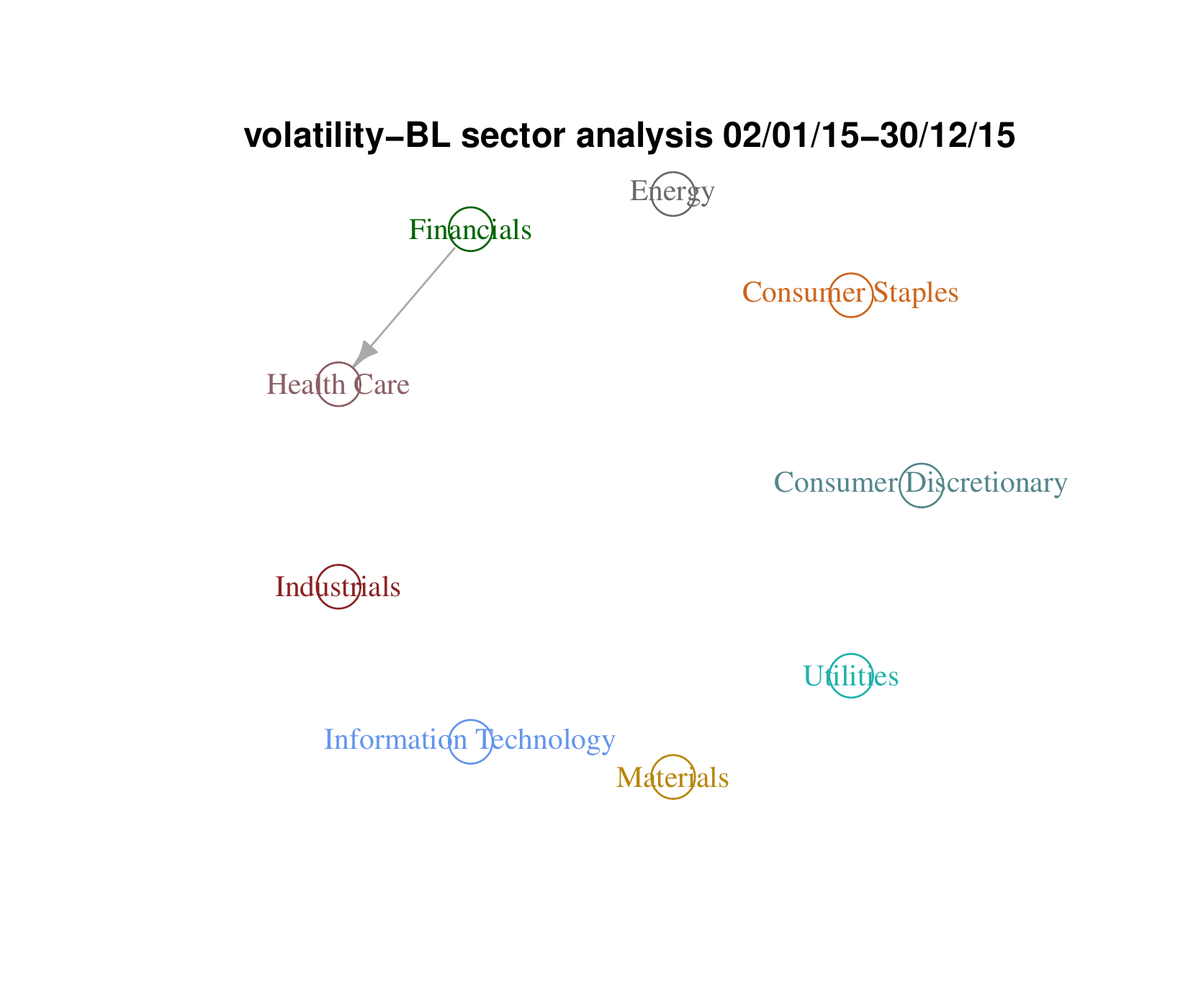}
		\end{subfigure}
		\begin{subfigure}[t]{0.54\textwidth}
			\centering
			\includegraphics[scale=0.47]{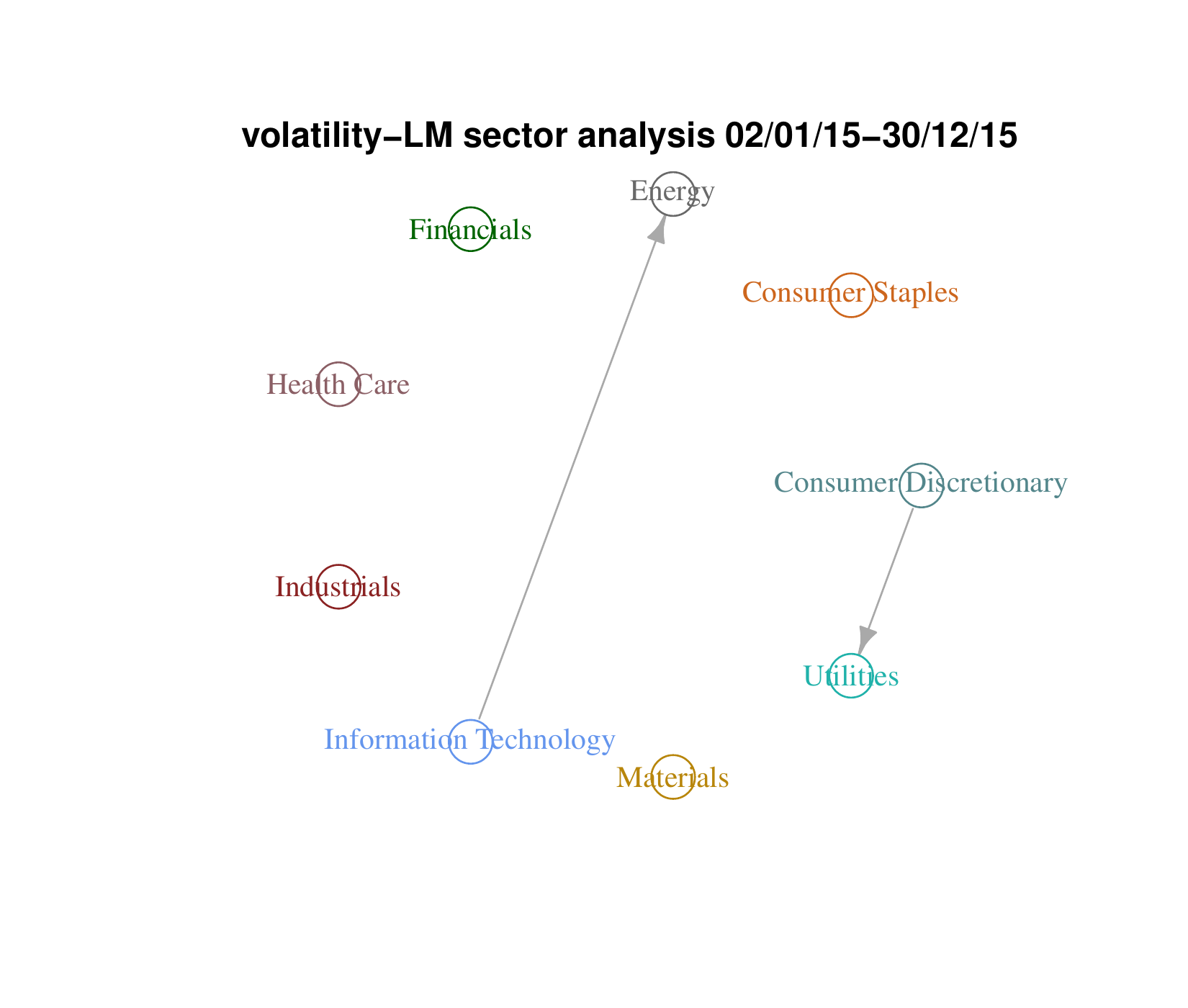}
		\end{subfigure}
		\caption{The dependency network among sectors from sentiment variables to financial variables.}
		\label{sector}
	\end{centering}
\end{figure}

Figure \ref{sector} describes the spillover effect network from sentiment to financial variables on the sector levels. In particular, the connections from energy to health care is found to be significant in the analysis of stock returns; while if volatility is focused on then the spillover effects from financials to health care, from information technology to energy, also from consumer discretionary to utilities are detected.

{
	\begin{remark}[Link to GGM]
		Another popular way to conduct the network analysis in the literature is the GGM, which is corresponding to the estimation of a high dimensional precision matrix. And under the Gaussian assumption our SRE can be linked to a nodal wise GGM. In particular, one can estimate the coefficients in each equation of SRE by using a sparse Graphical model estimation, for example the LASSO type estimation as in \citet{yuan2007model}, and thus we build the link equation-by- equation.
		
		Consider a high-dimensional VAR(1) model as in Example \ref{examp4}, the $j$th equation in the SRE is given by $Y_{j,t}=\Phi_{j\cdot}Y_{t-1}+\varepsilon_{j,t}$, where $Y_t$ is covariance stationary with $\operatorname{Var}(Y_{t})=\Gamma$ (p.d.). Correspondingly, we look at the vector $\widetilde Y_{j,t}=(Y_{j,t},Y_{1,t-1},\ldots,Y_{J,t-1})^\top$ belonging to an undirected graph $(V_j,E_j)$ with vertex set $(1,\ldots,J+1)$. Suppose $\widetilde{Y}_{j,t}\sim\operatorname{MVN}(0,\Sigma_j)$, $\Sigma_j=\begin{bmatrix}
		\Gamma_{jj} & \Phi_{j\cdot} \Gamma \\
		(\Phi_{j\cdot} \Gamma)^{\top} & \Gamma \end{bmatrix}$. Define $C_j\defeq\Phi_{j\cdot}\Gamma\Phi_{j\cdot}^\top$, then we have the precision matrix as $\Theta_j=\Sigma_j^{-1}=\begin{bmatrix}
		(\Gamma_{jj}-C_j)^{-1} & -(\Gamma_{jj}-C_j)^{-1}\Phi_{j\cdot} \\
		-\Phi_{j\cdot}^{\top}(\Gamma_{jj}-C_j)^{-1} & \Gamma^{-1}+\Phi_{j\cdot}^{\top}(\Gamma_{jj}-C_j)^{-1} \Phi_{j\cdot}
		\end{bmatrix} $. It can be seen that $\Phi_{jk}=0$ would imply that the $(1,k+1)$th element of $\Theta_j$ is zero and vice versa. In addition, a LASSO type estimator proposed in \citet{yuan2007model} can be obtained by solving $$\widehat\Theta_j= \arg \max_{\Theta}\{-\log \det(\Theta)+ \operatorname{trace}( S_j\Theta)+ \lambda_j \sum_{\ell k}|\Theta_{\ell k}|\},$$
		where $S_j\defeq n^{-1}\sum^n_{t=1} \widetilde Y_{j,t} \widetilde Y_{j,t}^{\top}$.
		
		In an unreported simulation study we compare the estimation performance between our proposed approach and the nodal wise GGM under the VAR(1) model. The results show that the nodal wise GGM which is approximated to SRE has worse prediction performance than our method, which can be obtained from the authors upon request.
	\end{remark}
	
}

\newpage%

\vskip 2em \centerline{\Large \bf Supplementary Material} \vskip -1em
\setcounter{subsection}{0}
\vskip 2em



\begin{appendices}
\renewcommand{\thesubsection}{A.\arabic{subsection}}
\setcounter{equation}{0}
\renewcommand{\theequation}{A.\arabic{equation}}
\setcounter{theorem}{0}
\renewcommand{\thetheorem}{A.\arabic{theorem}}
\setcounter{lemma}{0}
\renewcommand{\thelemma}{A.\arabic{lemma}}
\setcounter{table}{0}
\renewcommand{\thetable}{A.\arabic{table}}
\setcounter{corollary}{0}
\renewcommand{\thecorollary}{A.\arabic{corollary}}

\section{Theorems for Joint Penalty over Equations}\label{app.joint}

Recall that the theoretical choice $\lambda^0(1-\alpha)$ is defined as the $(1-\alpha)$ quantile of \\
$\underset{1\leq k\leq K, 1\leq j\leq J}{\max}2c\sqrt{n}|S_{jk}{/\Psi_{jk}}|$. First, we provide the analogue results of Theorem \ref{gammabound} and Corollary \ref{betabound2}.

\begin{theorem}\label{gammabound.joint}
	Under \hyperref[A1]{(A1)} and \hyperref[A3]{(A3)}, we have
	\begin{align}
	\P(2c\sqrt{n}\max_{1\leq k\leq K, 1\leq j\leq J} |S_{jk}{/\Psi_{jk}}|\geq r) \leq  &C_1\varpi_nnr^{-q}\sum_{j=1}^J\sum_{k=1}^K\frac{\|X_{jk,\cdot}\vps_{j,\cdot}\|^q_{q,\varsigma}}{\Psi_{jk}^q}\notag\\
	&+C_2\sum_{j=1}^J\sum_{k=1}^K\exp\Big(\frac{-C_3r^2\Psi_{jk}^2}{n \|X_{jk,\cdot}\vps_{j,\cdot}\|^2_{2,\varsigma}}\Big),
	\end{align}
	where for $\varsigma > 1/2-1/q$ (weak dependence case), $\varpi_n = 1$; for $\varsigma< 1/2-1/q$  (strong dependence case), $\varpi_n = n^{q/2-1- \varsigma q}$. $C_1,C_2,C_3$ are constants depending on $q$ and $\varsigma$.
\end{theorem}

\begin{corollary}[Bound for $\lambda^0(1-\alpha)$ and Oracle Inequalities under IC]\label{lambdabound.joint}
	Under \hyperref[A1]{(A1)} and \hyperref[A3]{(A3)}, given $\lambda^0(1-\alpha)$ satisfies
	\beq\label{lambdabound.eq}
	\lambda^0(1-\alpha) \lesssim \max_{1\leq k\leq K, 1\leq j\leq J}\bigg\{\|X_{jk,\cdot}\vps_{j,\cdot}\|_{2,\varsigma}\sqrt{n\log(KJ/\alpha) }\vee\|X_{jk,\cdot}\vps_{j,\cdot}\|_{q,\varsigma}(n\varpi_nKJ/\alpha)^{1/q}\bigg\},
	\eeq
	additionally assume that the RE condition \hyperref[A2]{(A2)} holds uniformly over equations $j=1,\ldots,J$ with probability $1-\smallO(1)$, and under the exact sparsity assumption \eqref{sparse}, then $\hat\beta_j$ obtained from \eqref{beta} under IC satisfy
	\beq
	|\hat \beta_j - \beta_{j}^0|_{j, pr} \lesssim C\sqrt{s}\max_{1\leq k \leq K}\Psi_{jk}\max_{1\leq j\leq J}\bigg\{\|X_{jk,\cdot}\vps_{j,\cdot}\|_{2,\varsigma}\sqrt{\log (KJ/\alpha) /n}\vee\|X_{jk,\cdot}\vps_{j,\cdot}\|_{q,\varsigma}n^{1/q-1}(\varpi_nKJ/\alpha)^{1/q} \bigg\},
	\eeq
	with probability $1- \alpha - \smallO(1)$, where for $\varsigma > 1/2-1/q$ (weak dependence case), $\varpi_n = 1$; for $\varsigma< 1/2-1/q$  (strong dependence case), $\varpi_n = n^{q/2-1- \varsigma q}$, and the constant $C$ depends on the RE constants.
\end{corollary}


The other empirical choices of the joint penalty level can be:
\begin{itemize}	
	\item[a)]\label{lambda_a} $Q(1-\alpha)$:  the $(1-\alpha)$ quantile of $ 2 c \underset{1\leq k\leq K, 1\leq j\leq J}{\max} \sqrt{n}|Z_{jk}{/\Psi_{jk}}|$. 
	In practice, one can take an alternative choice such that $\widetilde{Q}(1-\alpha) \defeq 2 c \sqrt{n}\Phi^{-1}\{1-\alpha/(2KJ)\}$.
	\item[b)]\label{lambda_b} $\Lambda(1-\alpha)\defeq 2c\sqrt{n}q_{(1-\alpha)}^{[B]}$, where $q_{(1-\alpha)}^{[B]}$ is the $(1-\alpha)$ quantile of $\underset{1\leq k\leq K, 1\leq j\leq J}{\max} |Z_{jk}^{[B]}{/\Psi_{jk}}|$.
\end{itemize}

For \hyperref[lambda_a]{a)} again we need the Gaussian approximation results for the vectorized process $\widetilde{\mathcal{S}}\defeq \operatorname{vec}[\{(S_{jk})_{k=1}^K\}_{j=1}^J]=\frac{1}{\sqrt{n}}\sum_{t=1}^{n}\widetilde{\mathcal{X}}_{t}$, where  $\widetilde{\mathcal{X}}_t \defeq \operatorname{vec}[\{(X_{jk,t}\vps_{j,t})_{k=1}^K\}_{j=1}^J]$ similar to Theorem \ref{gausappro} and Corollary \ref{gausappro.c} to justify the choice of $\lambda$ as $Q(1-\alpha)$.

{Let $\mathcal{X}_t \defeq \operatorname{vec}[\{(X_{jk,t})_{k=1}^K\}_{j=1}^J]$. 
	We first aggregate the dependence adjusted norm over $j=1,\ldots,J$ and $k=1,\ldots,K$:
	\begin{equation}\label{dan2}
	\||\mathcal{X}_{\cdot}|_\infty\|_{q,\varsigma}\defeq\sup_{m\geq 0}(m+1)^{\varsigma} \sum^{\infty}_{t=m} \delta_{q,t}, \,\, \delta_{q,t} \defeq \| |\mathcal{X}_t- \mathcal{X}_t^\ast|_\infty \|_q,
	\end{equation}
	where $q\geq1$, and $\varsigma>0$. Moreover, define the following quantities
	\begin{align}
	&\Phi_{q,\varsigma}\defeq 2\max_{1\leq k\leq K, 1\leq j\leq J} \|X_{jk,\cdot}\|_{q,\varsigma}\|\vps_{j,\cdot}\|_{q,\varsigma}, \,\, \Gamma_{q, \varsigma} \defeq  2\bigg(\sum_{j=1}^J\|\vps_{j,\cdot}\|^{q/2}_{q,\varsigma}\bigg)^{2/q}\bigg(\sum_{k=1}^K\sum_{j=1}^J \|X_{jk,\cdot}\|^{q/2}_{q,\varsigma}\bigg)^{2/q}\notag\\
	&\Theta_{q,\varsigma} \defeq \Gamma_{q,\varsigma}\wedge \big\{\||\mathcal{X}_{\cdot}|_\infty\|_{q,\varsigma}\|\vps_{j,\cdot}\|_{q,\varsigma}(\log KJ)^{3/2}\big\}.
	\end{align}}

Let $L_1 = [\Phi_{4,\varsigma}\Phi_{4,0} \{\log (KJ)\}^2]^{1/\varsigma}$, $W_1 = (\Phi^6_{6,0}+ \Phi^4_{8,0})\{\log(KJn)\}^7$, $W_2 = \Phi^2_{4,\varsigma}\{\log(KJn)\}^4$, $W_3 = [n^{-\varsigma} \{\log (KJn)\}^{3/2} \Theta_{2q,\varsigma}]^{1/(1/2-\varsigma-1/q)}$, $N_1$ $=$ $\{n/\log (KJ)\}^{q/2} \Theta_{2q, \varsigma}^{q}$, $N_2=n\{\log (KJ)\}^{-2}\Phi_{4,\varsigma}^{-2}$, $N_3$ $=$ $[n^{1/2}\{\log (KJ)\}^{-1/2}\Theta^{-1}_{2q, \varsigma}]^{1/(1/2-\varsigma)}$.
\begin{itemize}
	\item[(A4')]\label{A4'}
	i) (weak dependency case) Given $\Theta_{2q,\varsigma} < \infty$ with $q \geq 4$ and $\varsigma > 1/2 - 1/q$, then \\
	$\Theta_{2q, \varsigma} n^{1/q-1/2}\{\log (KJn)\}^{3/2} \to 0$ and $L_1\max(W_1, W_2) = \smallO(1) \min (N_1,N_2)$.\\
	ii) (strong dependency case) Given $0<\varsigma< 1/2 -1/q$, then $\Theta_{2q,\varsigma}\{\log (KJ)\}^{1/2} = \smallO(n^{\varsigma})$ and $L_1 \max(W_1,W_2,W_3) = \smallO(1)\min(N_2,N_3)$.
\end{itemize}


\begin{theorem}\label{gausappro.joint}
	Under \hyperref[A1]{(A1)}, \hyperref[A3]{(A3)} and \hyperref[A4]{(A4')}, for each $k=1,\ldots,K$, $j=1,\ldots,J$ assume that there exists a constant $c>0$ such that {$\underset{1\leq k\leq K,1\leq j\leq J}\min\operatorname{avar}(S_{jk})\geq c$}, then we have
	\beq
	\rho\big(D^{-1}\widetilde{\mathcal{S}}, D^{-1}\widetilde{\mathcal{Z}}\big)\rightarrow 0, \quad \text{as } n \to \infty,
	\eeq
	where $\widetilde{\mathcal{Z}}\sim \operatorname{N} (0, \Sigma_{\widetilde{\mathcal{X}}})$, $\Sigma_{\widetilde{\mathcal{X}}}$ is the $JK \times JK$ long-run variance-covariance matrix of $\widetilde{\mathcal{X}}_{t}$, and $D$ is a diagonal matrix with the square root of the diagonal elements of $\Sigma_{\widetilde{\mathcal{X}}}$, namely {$$\bigg\{\sum^{\infty}_{\ell=-\infty} \E (X_{jk,t}X_{jk,(t-\ell)}\vps_{j,t} \vps_{j,(t-\ell)} )\bigg\}^{1/2}=\sqrt{\operatorname{avar}(S_{jk})}, \text{ for  } k=1,\ldots,K, j=1,\ldots,J.$$}
\end{theorem}

\begin{corollary}\label{gausappro.c2}
	Under the conditions of Theorem \ref{gausappro.joint}, we have
	\beq
	\sup_{\alpha \in(0,1)}|\P\{\max_{1\leq k\leq K,1\leq j\leq J} 2c\sqrt{n} |S_{jk}/\Psi_{jk}| \geq  Q(1-\alpha) \} - \alpha| \to 0, \quad   \text{as } n \to \infty.
	\eeq
\end{corollary}

\begin{corollary}\label{gausappro.est2}
	Under the conditions of Theorem \ref{gausappro.joint}, and assume $\Phi_{2q,\varsigma}<\infty$ with $q>4$, $b_n = \bigO(n^{\eta})$ for some $0 <\eta< 1$. Let $F_{\varsigma} = n$, for $\varsigma >1-2/q$; $F_{\varsigma} = l_nb_n^{q/2-\varsigma q/2}$, for $1/2-2/q<\varsigma<1-2/q$; $F_{\varsigma} = l_n^{q/4-\varsigma q/2}b_n^{q/2-\varsigma q/2} $, for $\varsigma<1/2-2/q$. Given $n^{-1}\{\log(KJ)\}^2\max\big\{n^{1/2}b_n^{1/2}\Phi^2_{2q,\varsigma},\\
	n^{1/2} b_n^{1/2} \sqrt{\log (KJ)}\Phi_{8,\varsigma}^2,F^{2/q}_{\varsigma} \Gamma^2_{2q, \varsigma}(KJ)^{2/q}, \Phi_{2,0}\Phi_{2,\varsigma}v'(b_n)n/\sqrt{\log (KJ)}\big\}=\smallO(1)$, where $v'(b_n) = (b_n+1)^{-\varsigma}+2 v_{n,2}/b_n$,
	$v_{n,2} = \log b_n $ (resp. $b_n^{-\varsigma+1}$ or 1) for $\varsigma = 1$ (resp. $\varsigma < 1$ or $\varsigma>1$), then we have	
	\beq
	\rho\big(\hat D^{-1}\widetilde{\mathcal{S}}, D^{-1}\widetilde{\mathcal{Z}}\big)\rightarrow 0, \quad \text{as } n \to \infty,
	\eeq
	where $\hat D=\{\diag(\hat\Sigma_{\widetilde{\mathcal{X}}})\}^{1/2}$, $\hat \Sigma_{\widetilde{\mathcal{X}}} = \frac{1}{b_nl_n}\sum_{i=1}^{l_n}\big(\sum_{l=(i-1)b_n+1}^{ib_n}\widetilde{\mathcal{X}}_{l}\big)\big(\sum_{l=(i-1)b_n+1}^{ib_n}\widetilde{\mathcal{X}}_{l}\big)^\top$.
\end{corollary}

{
	
	Similar to Corollary \ref{betabound4}, we can provide a refined bound for $\lambda^0(1- \alpha)$ and also the oracle inequalities under IC as follows.
	\begin{corollary}[Bounds for $\lambda^0(1-\alpha)$ and Oracle Inequalities under IC with Gaussian Approximation Results]\label{betabound4.joint}
		Under the conditions of Theorem \ref{gausappro.joint},
		suppose $2\{\log (KJ)\}^{-1/2} + \rho(D^{-1}\widetilde{\mathcal{S}}, D^{-1}\widetilde{\mathcal{Z}}) = \smallO(\alpha)$ and let $Z_{\alpha}= 2\tilde c\sqrt{n \log(KJ)}$, where $\tilde c$ is no less than the $c$ in the definition of $\lambda^0(1-\alpha)$, then we have $\lambda^0(1-\alpha)$ satisfying
		\beq
		\lambda^0(1-\alpha) \leq Z_{\alpha},
		\eeq
		additionally assume that the RE condition \hyperref[A2]{(A2)} holds uniformly over equations $j=1,\ldots,J$ with probability $1-\smallO(1)$, and given the exact sparsity assumption \eqref{sparse}, then $\hat\beta_j$ obtained from \eqref{beta} under IC satisfies
		\beq
		|\hat \beta_j - \beta_{j}^0|_{j, pr} \lesssim C\sqrt{s}\max_{1\leq k \leq K}\Psi_{jk}\sqrt{\log (KJ) /n},
		\eeq
		with probability $1 - \alpha - \smallO(1)$, where the constant $C$ depends on the RE constants.
		
	\end{corollary}
}

Next, we need to show the validity of \hyperref[lambda_b]{b)}. Let $\widetilde{\mathcal{Z}}^{[B]}\defeq \operatorname{vec}[\{(Z^{[B]}_{jk})_{k=1}^K\}_{j=1}^J]$ and $\widetilde{\Psi}\defeq \operatorname{vec}[\{(\Psi_{jk})_{k=1}^K\}_{j=1}^J]$. Similarly to Theorem \ref{validboot} we have the following results:
\begin{theorem}\label{validboot.joint}
	Under \hyperref[A1]{(A1)}, \hyperref[A3]{(A3)} and \hyperref[A4]{(A4')}, assume $\Phi_{2q,\varsigma}<\infty$ with $q>4$, $b_n = \bigO(n^{\eta})$ for some $0 <\eta< 1$ {(the detailed rate is calculated in \eqref{ratebn.joint})}, then
	\begin{equation}
	\tilde{\rho}_n \defeq \sup_{r\in{\R}}| \operatorname{P} (|\widetilde{\mathcal{Z}}^{[B]}/\widetilde{\Psi}|_{\infty} \leq r|\mathcal{X}_{\cdot},\vps_{\cdot}) -\operatorname{P}(|\widetilde{\mathcal{Z}}/\widetilde{\Psi}|_{\infty}\leq r)| \to 0,\,\text{as } n \to \infty,
	\end{equation}
	and
	\begin{equation}
	\underset{\alpha \in(0,1)}{\sup}\big|\operatorname{P}(|\widetilde{\mathcal{S}}{/\widetilde{\Psi}}|_{\infty} \geq q^{[B]}_{(1-\alpha)} 
	\big) - \alpha\big| \to 0,\, \text{as } n \to \infty.
	\end{equation}
\end{theorem}

{Lastly, we show the performance bounds for the OLS post-LASSO estimator in the following theorem.
	
	For each $p\leq n$, $\widetilde{T}_{j} \subset \{1, \cdots, K\}$, $|\widetilde{T}_j \backslash T_j|\leq p$, we define the class of functions \\
	$\mathcal{G}_{\widetilde{T}_j} = \{\varepsilon_{j,t} X_{j,t}^{\top}\delta/|\delta|_{j,pr}, \supp(\delta) \subseteq \widetilde{T}_j, |\delta|_2 = 1\}$. The covering number of the function class is given by $\sup_{\mathcal Q}\mathcal{N}(\epsilon ,\mathcal{G}_{\widetilde{T}_j}, \|\cdot\|_{\mathcal Q,1})$. Also define $\mathcal{F}_{j,p} = \{\mathcal{G}_{\widetilde{T}_j}: \widetilde{T}_j \subset \{1, \cdots, K\}, |\widetilde{T}_j\backslash T_j|\leq p\}$. For any $f\in\mathcal{F}_{j,p}$, there exists a set $F_{j,p}$ such that $\min_{f'\in F_{j,p}}\|f-f'\|_{\mathcal Q,1} \leq \epsilon$, and the cardinality of the set is denoted by $|F_{j,p}|$. Consider the vector $\vartheta_t$ of length $|F_{j,p}|$, such that for $l=1,\ldots,|F_{j,p}|$, there is $\vartheta_{l,t}=(f-\E f)/\psi_f$ with $\psi_f=\{\operatorname{avar}(G_n(f))\}^{1/2}$, corresponding to each $f\in F_{j,p}$. The aggregated dependence adjusted norm is given by \begin{equation}
	\|\vartheta_{\cdot}\|_{q,\varsigma}\defeq\sup_{m\geq 0}(m+1)^{\varsigma} \sum^{\infty}_{t=m}\| |\vartheta_t- \vartheta_t^\ast|_\infty \|_q,
	\end{equation}
	where $q\geq1$, and $\varsigma>0$. Moreover, define the following quantities (for simplicity we drop the subscripts $j,p$)
	\begin{align}
	&\Phi^\vartheta_{q,\varsigma}\defeq \max_{1\leq l\leq|F_{j,p}|} \|\vartheta_{l,\cdot}\|_{q,\varsigma}, \,\, \Gamma^\vartheta_{q, \varsigma} \defeq  \bigg(\sum_{l=1}^{|F_{j,p}|} \|\vartheta_{l,\cdot}\|^{q}_{q,\varsigma}\bigg)^{1/q}\notag,\\
	&\Theta^\vartheta_{q,\varsigma} \defeq \Gamma^\vartheta_{q,\varsigma}\wedge \big\{\|\vartheta_{\cdot}\|_{q,\varsigma}(\log|F_{j,p}|)^{3/2}\big\}.
	\end{align}
	To evoke the Gaussian approximation on $G_n(f)/\psi_f$, we need to impose the following assumptions additionally.
	Define $L_1^\vartheta = \{\Phi^\vartheta_{2,\varsigma}\Phi^\vartheta_{2,0} (\log|F_{j,p}|)^2\}^{1/\varsigma}$, $W_1^\vartheta = \{(\Phi^\vartheta_{3,0})^6+ (\Phi^\vartheta_{4,0})^4\}\{\log(|F_{j,p}|n)\}^7$, $W_2^\vartheta = (\Phi^\vartheta_{2,\varsigma})^2\{\log(|F_{j,p}|n)\}^4$, $W_3^\vartheta = [n^{-\varsigma} \{\log (|F_{j,p}|n)\}^{3/2} \Theta^\vartheta_{q,\varsigma}]^{1/(1/2-\varsigma-1/q)}$, \\
	$N_1^\vartheta =(n/\log|F_{j,p}|)^{q/2} (\Theta^\vartheta_{q, \varsigma})^{q}$, $N_2^\vartheta=n(\log|F_{j,p}|)^{-2}(\Phi^\vartheta_{2,\varsigma})^{-2}$, $N_3^\vartheta = \{n^{1/2}(\log|F_{j,p}|)^{-1/2}(\Theta^{\vartheta}_{q, \varsigma}\})^{1/(1/2-\varsigma)}$.
	\begin{itemize}
		\item[(A7)]\label{A7}
		i) (weak dependency case) Given $\Theta^\vartheta_{q,\varsigma} < \infty$ with $q \geq 2$ and $\varsigma > 1/2 - 1/q$, then \\
		$\Theta^\vartheta_{q, \varsigma} n^{1/q-1/2}\{\log (|F_{j,p}|n)\}^{3/2} \to 0$ and $L_1^\vartheta\max(W_1^\vartheta, W_2^\vartheta) = \smallO(1) \min (N_1^\vartheta,N_2^\vartheta)$.\\
		ii) (strong dependency case) Given $0<\varsigma< 1/2 -1/q$, then $\Theta^\vartheta_{q,\varsigma}(\log |F_{j,p}|)^{1/2} = \smallO(n^{\varsigma})$ and $L_1^\vartheta \max(W_1^\vartheta,W_2^\vartheta,W_3^\vartheta) = \smallO(1)\min(N_2^\vartheta,N_3^\vartheta)$.
	\end{itemize}
	
	\begin{remark}
		For a random vector $z_t\in R^K$, suppose there exist constants $C,D>0$, such that $\max_k \E \{\exp(|z_{k,t}/D|^q)\} \leq C$. Then by Jensen's inequality it follows that  $\||z_{t}|_\infty\|_q\leq D (\log K+ \log C)^{1/q}$. In particular, for the case of sub-Gaussian random variables, there exists constant $D>0$ such that  $\E \{\exp(|z_{k,t}/D|^2)\}  -1 \leq 1$, which implies $\||z_{t}|_\infty\|_2\lesssim D \sqrt{\log K}$.	
		
		Similar to the discussion in Comment \ref{garate}, consider the case with $\Theta_{q, \varsigma}^\vartheta=\bigO((\log|F_{j,p}|)^{1/q})$ and $\Phi_{q,\varsigma}^\vartheta=\bigO(1)$, where $\varsigma>1/2-1/q$. Then $\Theta_{q, \varsigma}^\vartheta n^{1/q-1/2}\{\log (|F_{j,p}|n)\}^{3/2} \to 0$ becomes $\log|F_{j,p}|\{\log(n|F_{j,p}|)\}^{3q/2}=\smallO(n^{q/2-1})$, which implies that $L_1^\vartheta\max(W_1^\vartheta, W_2^\vartheta) = \smallO(1) \min (N_1^\vartheta,N_2^\vartheta)$. 
		
		As shown in the proof of Theorem \ref{post}, $|F_{j,p}|\lesssim K^p(6 \mu_j(p)\sigma/ \epsilon)^{s+p}$ with \\
		$\epsilon=\sqrt{p\log K+(p+s)\log  (6 \mu_j(p) \sigma)}(4 \sqrt{n})^{-1}$. This means with \hyperref[A7]{(A7)}, the dimension K has to satisfy the condition $\{p\log K + (s+p)\log(\sqrt{n})\}^{1+3q/2}=\smallO(n^{q/2-1})$, where we consider the case such that $|F_{j,p}|$ is larger than $n$.
	\end{remark}
}

\begin{theorem}[Prediction Performance Bounds for OLS Post-LASSO]\label{post}
	Given \hyperref[A1]{(A1)}, \hyperref[A3]{(A3)} and \hyperref[A7]{(A7)}, suppose \hyperref[A2]{(A2)} (with $\bar c=\frac{c+1}{c-1}, c>1$) and \hyperref[A5]{(A5)} (with $\hat p_j=|\hat T_j\setminus T_j|$) hold uniformly over equations with probability $1-\smallO(1)$, then under the exact sparsity assumption \eqref{sparse}, for any $\tau>0$, there is a constant $C_\tau$ independent of $n$, for all $j =1,\ldots, J$ we have
	\begin{align}
	&|\hat \beta_j^{[P]} - \beta_{j}^0|_{j, pr} \leq C_\tau{ \max_{1\leq k \leq K}\Psi_{jk}\sqrt{\frac{p\log K+(p+s)\{\log  (6 \mu_j(p) \sigma)+\log n/2\}}{n}}}\notag\\
	&+\IF(T_j\nsubseteq \hat T_j)C\sqrt{s}\underset{1\leq k \leq K}{\max}\Psi_{jk}\max_{1\leq j\leq J}\big\{\|X_{jk,\cdot}\vps_{j,\cdot}\|_{2,\varsigma}\sqrt{\frac{\log (KJ/\alpha) }{n}}\vee\|X_{jk,\cdot}\vps_{j,\cdot}\|_{q,\varsigma}n^{1/q-1}(\varpi_nKJ/\alpha)^{1/q} \big\},
	\end{align}
	with probability $1-\alpha -\tau - \smallO(1)$, where for $\varsigma > 1/2-1/q$ (weak dependence case), $\varpi_n = 1$; for $\varsigma< 1/2-1/q$ (strong dependence case), $\varpi_n = n^{q/2-1- \varsigma q}$. {$\sigma=\max\limits_{1\leq j\leq J}\{\operatorname{avar}(n^{-1/2}\sum_{t=1}^n\vps_{j,t})\}^{1/2}$} and the constant $C$ depends on the RE constants.
\end{theorem}

{In particular, suppose the Gaussian approximation results hold for $\lambda^0(1-\alpha)$, the bound for it can be replaced according to Corollary \ref{betabound4.joint}.}


\renewcommand{\thesubsection}{B.\arabic{subsection}}
\setcounter{equation}{0}
\renewcommand{\theequation}{B.\arabic{equation}}
\setcounter{theorem}{0}
\renewcommand{\thetheorem}{B.\arabic{theorem}}
\setcounter{lemma}{0}
\renewcommand{\thelemma}{B.\arabic{lemma}}
\setcounter{table}{0}
\renewcommand{\thetable}{B.\arabic{table}}
\setcounter{remark}{0}
\renewcommand{\theremark}{B.\arabic{remark}}

\section{Detailed Proofs}

\subsection{Proofs of Single Equation Estimation}
%

\begin{proof}[Proof of Theorem \ref{gammabound}]
For each $j=1,\ldots J$, $k=1,\ldots,K$, applying Theorem 2 of \cite{wu2016performance} gives
\beq
\P(\sqrt{n}|S_{jk}| \geq x) \leq \frac{C'_1\varpi_nn \|X_{jk,\cdot}\vps_{j,\cdot}\|^q_{q,\varsigma}}{x^q}+C'_2 \exp\big(\frac{-C_3x^2}{n \|X_{jk,\cdot}\vps_{j,\cdot}\|^2_{2,\varsigma} }\big),\notag
\eeq
where for $\varsigma > 1/2-1/q$, $\varpi_n = 1$; for $\varsigma< 1/2-1/q$, $\varpi_n = n^{q/2-1- \varsigma q}$. $C'_1,C'_2,C_3$ are three constants depending on $q$ and $\varsigma$. 
It follows that the conclusion holds if we set $x=(2c)^{-1}\Psi_{jk}r$ and apply the Bonferroni inequality. 
\end{proof}


\begin{proof}[Proof of Theorem \ref{gausappro}]
According to the Minkowski's inequality and H\"{o}lder's inequality, we have
\begin{align*}
\sum^{\infty}_{t=m}\|X_{jk,t}\vps_{j,t} - X^\ast_{jk,t} \vps^\ast_{j,t}\|_{q} &\leq \sum^{\infty}_{t=m}\big\{\|X_{jk,t}(\vps_{j,t}-\vps^\ast_{j,t})\|_{q}+ \|(X_{jk,t} - X_{jk,t}^\ast)\vps^\ast_{j,t}\|_{q}\big\}\notag\\
&\leq \sum^{\infty}_{t=m}\big\{\|X_{jk,t}\|_{2q} \|\vps_{j,t} -\vps^\ast_{j,t}\|_{2q} + \| X_{jk,t} - X^\ast_{jk,t}\|_{2q} \|\vps_{j,t}\|_{2q}\big\}.
\end{align*}
Thus, it is easy to see that
\beq
\|X_{jk,\cdot} \vps_{j,\cdot}\|_{q, \varsigma} \leq \|X_{jk,\cdot}\|_{2q,0} \|\vps_{j,\cdot}\|_{2q,\varsigma}+ \|X_{jk,\cdot}\|_{2q,\varsigma} \|\vps_{j,\cdot}\|_{2q,0}
\leq 2  \|X_{jk,\cdot}\|_{2q,\varsigma} \|\vps_{j,\cdot}\|_{2q,\varsigma}.\notag
\eeq
Consequently, we have the following relationships:
\begin{align*}
&\max_{1\leq k\leq K} \|X_{jk,\cdot}\vps_{j,\cdot}\|_{q,\varsigma}\leq2\max_{1\leq k\leq K} \|X_{jk,\cdot}\|_{2q,\varsigma}\|\vps_{j,\cdot}\|_{2q,\varsigma},\notag\\
& (\sum_{k=1}^K \|X_{jk,\cdot}\vps_{j,\cdot}\|^{q}_{q,\varsigma})^{1/q} \leq2 \|\vps_{j,\cdot}\|_{2q,\varsigma} (\sum_{k=1}^K \|X_{jk,\cdot}\|^{q}_{2q,\varsigma})^{1/q},\notag\\
&\|X_{j,\cdot}\vps_{j,\cdot}\|_{q,\varsigma} \leq 2\|X_{j,\cdot}\|_{2q,\varsigma} \|\vps_{j,\cdot}\|_{2q,\varsigma}.
\end{align*}
Therefore, the conditions in Theorem 3.2 of \cite{ZW15gaussian} can be verified for the $K$-dimensional stationary process $X_{j,t}\vps_{j,t}$. Finally, applying that theorem yields the Gaussian approximation results.
\end{proof}

\begin{proof}[Proof of Corollary \ref{gausappro.c}]
It follows directly from the Gaussian approximation results in Theorem \ref{gausappro}.
\end{proof}

{
\begin{proof}[Proof of Corollary \ref{gausappro.est}]
{The proof follows that of Corollary 5.4 in \citet{ZW15gaussian}.}\\
For $w>0$, we have
\begin{align*}
\rho(\hat D_{j}^{-1}S_{j\cdot}, D_{j}^{-1}Z_j)&=\sup_{r\geq0}\big|\P(|\hat D_j^{-1}S_{j\cdot}|_{\infty} \geq r) - \P(|D_j^{-1}Z_j|_{\infty} \geq r)\big|\\
&\leq \rho(D_{j}^{-1}S_{j\cdot}, D_{j}^{-1}Z_j) + \sup_{r\geq0}\P(||D_j^{-1}Z_j|_{\infty}-r|\leq w) + \P(|(D_j^{-1}-\hat D_j^{-1})S_{j\cdot}|_\infty\geq w)\\
&\lesssim \rho(D_{j}^{-1}S_{j\cdot}, D_{j}^{-1}Z_j) + w\sqrt{\log K} + \P(|(D_j^{-1}-\hat D_j^{-1})S_{j\cdot}|_\infty\geq w),
\end{align*}
where the last line uses the arguments of Theorem 3 in \citet{chernozhukov2015comparison}. Let $V_{n,j}\defeq\underset{1\leq k\leq K}{\max}|\Psi_{jk}/\hat\Psi_{jk}-1|$ and $L_{n,j}\defeq\underset{1\leq k\leq K}{\max}|\Psi_{jk}^2-\hat\Psi_{jk}^2|$. Then $|(D_j^{-1}-\hat D_j^{-1})S_{j\cdot}|_\infty\leq V_{n,j}|D_j^{-1}S_{j\cdot}|_\infty$. As $\underset{1\leq k\leq K}{\min}\Psi_{jk}^2\geq c_j$, let $w=xy$, $0<x<c_j/2$, $y>0$, then
\begin{align*}
\P(|(D_j^{-1}-\hat D_j^{-1})S_{j\cdot}|_\infty\geq w)&\leq\P(V_{n,j}\geq2x/c_j) + \P(|D_j^{-1}S_{j\cdot}|_\infty\geq c_jy/2)\\
&\leq \P(L_{n,j}\geq x) + \rho(D_{j}^{-1}S_{j\cdot}, D_{j}^{-1}Z_j) + \P(|D_j^{-1}Z_{j}|_\infty\geq c_jy/2).
\end{align*}
It follows that
$$
\rho(\hat D_{j}^{-1}S_{j\cdot}, D_{j}^{-1}Z_j) \leq \rho(D_{j}^{-1}S_{j\cdot}, D_{j}^{-1}Z_j) + xy\sqrt{\log K} + \P(L_{n,j}\geq x) + \P(|D_j^{-1}Z_{j}|_\infty\geq c_jy/2).
$$

In particular, $L_{n,j} \leq  L_{n,j,1}+ L_{n,j,2}$, with $L_{n,j,1} = \max_{1\leq k\leq K}| \E\hat{\Psi}^2_{jk} - \hat{\Psi}^2_{jk}| $ and $L_{n,j,2}=\max_{1\leq k\leq K}|\Psi^2_{jk} - \E\hat{\Psi}^2_{jk}|$.

As for $L_{n,j,1}$, applying Theorem 5.1 of \cite{ZW15gaussian}, for $u\geq n^{1/2}b_n^{1/2}\Phi^2_{j,2q,\varsigma}$, we have
\beq
\P(nL_{n,j,1}\geq u) \lesssim \frac{F_{\varsigma}\Gamma_{j,2q,\varsigma}^{q}}{u^{q/2}} + K\exp\bigg(- \frac{C_{j} u^2}{n b_n \Phi^4_{j, 8,\varsigma}}\bigg),\notag
\eeq
where the constants $C_j$ depend on $\eta$, $q$, and $\varsigma$. Then we have $\P(L_{n,j,1}>x)\to0$, as $n\to\infty$, {if we set $x\geq\frac{\sqrt{\log K}}{n}\max\big\{n^{1/2}b_n^{1/2}\Phi^2_{j,2q,\varsigma}, cn^{1/2} b_n^{1/2} \sqrt{\log K}\Phi_{j,8,\varsigma}^2, cF^{2/q}_{\varsigma} \Gamma^2_{j,2q, \varsigma}\big\}$, for sufficiently large $c$.}

For $L_{n,j,2}$, define $v'(b_n) = (b_n+1)^{-\varsigma}+ 2v_{n,2}/b_n$, $v_{n,2} = \log b_n $ (resp. $b_n^{-\varsigma+1}$ or 1) for $\varsigma = 1$ (resp. $\varsigma < 1$ or $\varsigma>1$).
It can be shown that $L_{n,j,2}\leq \Phi_{j,2,0}\Phi_{j,2,\varsigma}v'(b_n)$. Note that $v'(b_n)$ is a special case of $v(b_n)$ in the proof of Theorem \ref{validboot} given $n\to\infty$, and the conclusion follows similarly.

It follows that $\P(L_{n,j}>x)\to0$, as $n\to\infty$, {if we set
\beq
x\geq\frac{\sqrt{\log K}}{n}\max\big\{n^{1/2}b_n^{1/2}\Phi^2_{j,2q,\varsigma}, cn^{1/2} b_n^{1/2} \sqrt{\log K}\Phi_{j,8,\varsigma}^2, cF^{2/q}_{\varsigma} \Gamma^2_{j,2q, \varsigma}, \Phi_{j,2,0}\Phi_{j,2,\varsigma}v'(b_n)n/\sqrt{\log K}\big\},\notag
\eeq
where $c$ is sufficiently large.} Moreover, given Theorem \ref{gausappro} and choosing $y=C\sqrt{\log K}$ (the constant $C>0$ is sufficiently large), the conclusion can be obtained.
\end{proof}
}

{
\begin{proof}[Proof of Corollary \ref{betabound4}]
Let $\tilde\rho_n\defeq\rho(D_{j}^{-1}S_{j\cdot}, D_{j}^{-1}Z_j)$ and by its definition, we have
\begin{align*}
\P(2c \sqrt{n} \max_{1\leq k\leq K} |S_{jk}/\Psi_{jk}|\leq Z_\alpha)&\geq\P(2c \sqrt{n} \max_{1\leq k\leq K} |Z_{jk}/\Psi_{jk}||\leq Z_\alpha) - \tilde\rho_n\\
&\geq1- \sum_{k=1}^K \P\{|Z_{jk}/\Psi_{jk}|\geq Z_{\alpha}/(2c\sqrt{n})\} - \tilde\rho_n\\
&\geq 1- \sum_{k=1}^K2\{Z_{\alpha}/(2c\sqrt{n})\}^{-1} \exp[- Z_{\alpha}^2/\{2(2c\sqrt{n})^2\}] - \tilde\rho_n\\
&\geq 1-2(\log K)^{-1/2} - \tilde\rho_n,
\end{align*}
where we have applied the union bound, the tail probability of Gaussian random variable and the condition that $Z_\alpha=2\tilde c\sqrt{n\log K}\geq2\sqrt{2}c\sqrt{n\log K}$.

It follows that $\lambda^0_j(1-\alpha)\leq Z_\alpha$ as $
1- \alpha = \P\{2c \sqrt{n} \underset{1\leq k\leq K}{\max} |S_{jk}/\Psi_{jk}|\leq \lambda_j^0(1-\alpha)\} \leq  \P(2c \sqrt{n} \underset{1\leq k\leq K}{\max} |S_{jk}/\Psi_{jk}|\leq Z_{\alpha})$, given $2(\log K)^{-1/2} + \tilde\rho_n=\smallO(\alpha)$ (note that Theorem \ref{gausappro} ensures that $\tilde\rho_n\to0$ with a polynomial rate as $n\to\infty$).

\end{proof}

}

\begin{proof}[Proof of Theorem \ref{validboot}]
Let $S_{jk,i}=\frac{1}{\sqrt{n}}\sum_{l=(i-1)b_n+1}^{ib_n}X_{jk,l}\varepsilon_{j,l}$, we first need to prove that 
\begin{align*}\label{sup2}
\rho_{n,j} &\defeq \sup_{r\in{\R}}\big| \P \big\{\max_{1\leq k\leq K}(Z_{jk}^{[B]}/\Psi_{jk})\leq r| X_{j,\cdot}, \vps_{j,\cdot}\big\} -\P\big\{\max_{1\leq k\leq K} (\widetilde{Z}_{jk}/\Psi_{jk})\leq r\big\}\big|\notag\\
&= \sup_{r\in{\R}}\big| \P \big\{\max_{1\leq k\leq K}\big(\sum_{i=1}^{l_n}e_{j,i}S_{jk,i}/\Psi_{jk}\big)\leq r| X_{j,\cdot}, \vps_{j,\cdot}\big\} -\P\big\{\max_{1\leq k\leq K} (\widetilde{Z}_{jk}/\Psi_{jk})\leq r\big\}\big|\to 0,\, \text{as } n \to \infty.
\end{align*}
{Given the sample variance covariance matrix ($K\times K$) $\Sigma_{j,n} = \sum^{n}_{\ell = -n}(1-|\ell|/n) \Gamma_j(\ell)$, where $\Gamma_j(\ell) = \E(X_{j,t}\vps_{j,t}X_{j,t-\ell}^{\top} \vps_{j,t-\ell})$, 
let $\widetilde{Z}_{j} = (\widetilde{Z}_{jk})_{k=1}^K\sim \N(0,\Sigma_{j,n})$. In addition, define $\Sigma_{j,b_n} = \sum^{b_n}_{\ell = -b_n}(1- |\ell|/b_n) \Gamma_{j}(\ell)$ and $\widehat\Sigma_j = \sum^{l_n}_{i=1}  S_{j,i}S_{j,i}^{\top}$, where $S_{j,i} = (S_{jk,i})_{k=1}^K$. Let $\Psi_j = \diag (\Psi_{jk})$, $\delta_j = \delta_{j1} + \delta_{j2}$, with $\delta_{j1} = |\Psi_j^{-1} \widehat\Sigma_j \Psi_j^{-1} - \Psi_j^{-1} \Sigma_{j,b_n}\Psi_j^{-1}|_{\max}$ and $\delta_{j2} = |\Psi_j^{-1} \Sigma_{j,b_n}\Psi_j^{-1} - \Psi_j^{-1} \Sigma_{j,n}\Psi_j^{-1}|_{\max}$, where $|\cdot|_{\max}$ is the maximum norm of a matrix.}
According to Theorem 2 of \cite{chernozhukov2015comparison}, 
$\rho_{n,j}$ is bounded by {$\pi(\delta_{j1})\vee\pi(\delta_{j2})$, with $\pi(\delta_j) \defeq C\delta_j^{1/3}\{1\vee a_K^2 \vee \log(1/\delta_j)\}^{1/3}(\log K)^{1/3}$}, where $a_K=\E(\underset{1\leq k\leq K}{\max}Z_{jk}/\Psi_{jk})\leq\sqrt{2\log K}$.

For the first part,
\begin{align*}
\delta_{j1} &=\max_{1\leq k_1,k_2\leq K}\bigg|\frac{\sum_{i=1}^{l_n}S_{jk_1,i}S_{jk_2,i}}{\Psi_{jk_1}\Psi_{jk_2}} - \frac{l_n\E (S_{jk_1,i}S_{jk_2,i})}{\Psi_{jk_1}\Psi_{jk_2}}\bigg|\notag\\
&\leq\frac{\underset{1\leq k_1,k_2\leq K}{\max}\big|\sum_{i=1}^{l_n}S_{jk_1,i}S_{jk_2,i}-l_n\E (S_{jk_1,i}S_{jk_2,i})\big|}{\underset{1\leq k_1,k_2\leq K}{\min}\Psi_{jk_1}\Psi_{jk_2}}.
\end{align*}
We need to analyze the tail probability of {$\delta_{j1}$}. Applying Theorem 5.1 of \cite{ZW15gaussian}, for $x\geq n^{1/2} b_n^{1/2}\Phi^2_{j,2q,\varsigma}$, we have
\beq
\P\bigg(n \delta_{j1}\geq \frac{x}{\underset{1\leq k_1,k_2\leq K}{\min}\Psi_{j_1k_1}\Psi_{j_2k_2}}\bigg) \lesssim \frac{K F_{\varsigma}\Gamma_{j,2q,\varsigma}^{q}}{x^{q/2}} + K^2 \exp\bigg(- \frac{C_{j} x^2}{n b_n \Phi^4_{j, 8,\varsigma}}\bigg),\notag
\eeq
for all large $n$, where $F_{\varsigma} = n$, for $\varsigma >1-2/q$; $F_{\varsigma} = l_nb_n^{q/2-\varsigma q/2}$, for $1/2-2/q<\varsigma<1-2/q$; $F_{\varsigma} = l_n^{q/4-\varsigma q/2}b_n^{q/2-\varsigma q/2} $, for $\varsigma<1/2-2/q$. The constants $C_j$ depend on $\eta$, $q$, and $\varsigma$. This ensures that {when $x\geq\max\big\{n^{1/2} b_n^{1/2}\Phi^2_{j,2q,\varsigma},cn^{1/2} b_n^{1/2} (\log K)^{1/2}\Phi_{j,8,\varsigma}^2,cK^{2/q} F^{2/q}_{\varsigma} \Gamma^2_{j,2q, \varsigma}\big\}$, 
the tail probability tends to 0, as $n\to\infty$, for sufficiently large $c$.}

It follows that {$\pi(\delta_{j1})\to0$} as $n\to\infty$, given $x=\smallO\{n(\log K)^{-2}\}$, 
which implies the following conditions on $b_n$:
\begin{align*}
&b_n=\smallO\{n(\log K)^{-4}\Phi^{-4}_{j,2q,\varsigma}\wedge n(\log K)^{-5}\Phi_{j,8,\varsigma}^{-4}\},\quad F_\varsigma=\smallO\{n^{q/2}(\log K)^{-q}K^{-1}\Gamma_{j,2q,\varsigma}^{-q}\}.
\end{align*}

{
For the second part, by defining $\psi_j \defeq \underset{1\leq k_1,k_2\leq K}{\min} \Psi_{jk_1} \Psi_{jk_2}$, we have $$\delta_{j2} \leq \bigg|\psi_j^{-1}\bigg\{\sum_{b_n< |\ell|\leq n} (1- |\ell|/n) \Gamma_{j}(\ell)+ \sum_{\ell=-b_n}^{b_n} |\ell| (-1/n + 1/b_n) \Gamma_j(\ell)\bigg\}\bigg|_{\max}.$$
Recall that
\begin{align*}
|\Gamma_{j,k_1,k_2}(\ell)|&= \bigg|\sum^{\infty}_{h = 0}\E\{(\mathcal{P}_h (X_{jk,0} \vps_{j0}) \mathcal{P}_h( X_{jk_2,\ell}\vps_{j,\ell})\}\bigg|\\
&\leq \sum^{\infty}_{h=0} 
\|X_{jk_1,h}\vps_{j,h} - X_{jk_1,h}^\ast\vps_{j,h}^\ast\|_2\|X_{jk_2,h+\ell}\vps_{j,h+\ell} - X_{jk_2,h+\ell}^\ast\vps_{j,h+\ell}^\ast\|_2,
\end{align*}
where the operator is given by $\mathcal{P}_h(\cdot) \defeq \E(\cdot|\mF_{h}) - \E(\cdot|\mF_{h-1})$. It follows that
\begin{eqnarray}\label{delta2.bound}
&&\bigg|\sum_{b_n<|\ell|\leq n} (1-|\ell|/n) \Gamma_{j,k_1,k_2}(\ell)+\sum_{\ell=-b_n}^{b_n} |\ell|(-1/n + 1/b_n) \Gamma_{j,k_1,k_2}(\ell)\bigg|\notag\\
&\leq& \Delta_{0,2,j,k_1}\Delta_{b_n+1, 2, j,k_2} + \frac{2}{n}\Delta_{0,2,j,k_1}\sum^n_{\ell= b_n+1}\Delta_{\ell,2,j,k_2}+ 2\frac{n-b_n}{n b_n}  \Delta_{0,2,j,k_1}\sum_{\ell=1}^{b_n}\Delta_{\ell,2,j,k_2},
\end{eqnarray}
where $\Delta_{m,2,j,k}=\sum_{t=m}^\infty\|X_{jk,t}\vps_{j,t} - X_{jk,t}^\ast\vps_{j,t}^\ast\|_2$.
Given the fact that $\Delta_{0,2,j,k} \leq \Phi_{j,4,0}$, $\Delta_{\ell,2,j,k}\leq  \Phi_{j,4,\varsigma}\ell^{-\varsigma}$, 
\eqref{delta2.bound} is bounded by $\Phi_{j,4,0}\Phi_{j,4,\varsigma} \{(b_n+ 1)^{-\varsigma}+ 2n^{-1}\sum^{n}_{\ell = b_n+1} \ell^{-\varsigma}+ 2\frac{n-b_n}{n b_n} \sum^{b_n}_{\ell=1} \ell^{-\varsigma}\}=\Phi_{j,4,0}\Phi_{j,4,\varsigma} v(b_n)$ for any $k_1,k_2$, where $v(b_n)$ is a function of $b_n$. Note that $v(b_n)\lesssim (b_n+1)^{-\varsigma}+ 2v_{n,1}/n + 2(n-b_n) v_{n,2}/(n b_n)$, where $v_{n,1} = \log \{n/(b_n+1)\}$ (resp. $n^{-\varsigma+1}$ or $(b_n+1)^{-\varsigma+1}$) for $\varsigma = 1$ (resp. $\varsigma<1$ or $\varsigma>1$), $v_{n,2} = \log b_n $ (resp. $b_n^{-\varsigma+1}$ or 1) for $\varsigma = 1$ (resp. $\varsigma < 1$ or $\varsigma>1$).
Therefore, the bound of $\delta_{j2}$ 
would decrease as $b_n$ increases. In particular, we need to impose an addition assumption such that 
$\Phi_{j,4,0}\Phi_{j,4,\varsigma} v(b_n) = \smallO\{(\log K)^{-2}\}$ to guarantee $\pi(\delta_{j2})\to0$.

The results for the two parts above ensure that $\rho_{n,j}\to0$ as $n\to\infty$, given $x=\smallO\{n(\log K)^{-2}\}$ and $\Phi_{j,4,0}\Phi_{j,4,\varsigma} v(b_n) = \smallO\{(\log K)^{-2}\}$, which imply the following conditions on $b_n$:
	\begin{align}\label{ratebn}
	&b_n=\smallO\{n(\log K)^{-4}\Phi^{-4}_{j,2q,\varsigma}\wedge n(\log K)^{-5}\Phi_{j,8,\varsigma}^{-4}\},\,F_\varsigma=\smallO\{n^{q/2}(\log K)^{-q}K^{-1}\Gamma_{j,2q,\varsigma}^{-q}\}.\notag\\
	&\Phi_{j,4,0}\Phi_{j,4,\varsigma}\{b_n^{-1}+\log(n/b_n)/n+(n-b_n)\log b_n/(nb_n)\}(\log K)^2=\smallO(1),\,\text{if } \varsigma=1;\notag\\
	&\Phi_{j,4,0}\Phi_{j,4,\varsigma}\{b_n^{-1}+n^{-\varsigma}+(n-b_n)b_n^{-\varsigma+1}/(nb_n)\}(\log K)^2=\smallO(1),\,\text{if } \varsigma<1;\notag\\
	&\Phi_{j,4,0}\Phi_{j,4,\varsigma}\{b_n^{-1}+n^{-1}b_n^{-\varsigma+1}+(n-b_n)/(nb_n)\}(\log K)^2=\smallO(1),\,\text{if } \varsigma>1.
	\end{align}

At last, combining the Gaussian approximation results for $S_{jk}/\Psi_{jk}$ 
and applying Theorem 3.1 in \cite{CCK13AoS}, we have
\beq\label{sup1s}
\sup_{\alpha \in(0,1)}\big|\P\big(\max_{1\leq k\leq K} |S_{jk}{/\Psi_{jk}}| \geq q^{[B]}_{j,(1-\alpha)}
	\big) - \alpha\big|\lesssim \rho_{n,j}+ \pi'(z)+ \P(\delta_j \geq z),\notag
\eeq
where $\pi'(z) = z^{1/3}\{1\vee \log(K/z)\}^{2/3}$. We need to pick $z$ such that $\pi'(z)+ \P(\delta_j \geq z)\to 0$ as $n\to \infty$ and it can be obtained by taking $z = R^{1/2}_n/(\log K)$, 
with \\
{$R_n = n^{-1} \max\big\{n^{1/2} b_n^{1/2}\Phi^2_{j,2q,\varsigma}, cn^{1/2} b_n^{1/2} (\log K)^{1/2}\Phi_{j,8,\varsigma}^2, cK^{2/q} F^{2/q}_{\varsigma} \Gamma^2_{j,2q, \varsigma}, n \Phi_{j,2,0}\Phi_{j,2,\varsigma} v(b_n) \big\}$, with sufficiently large $c$.}
}

\end{proof}

\begin{remark}[Admissible rate of $b_n$]
Consider the special case with $\Phi_{j,2q,\varsigma}=\bigO(1)$ and $\Gamma_{j,2q,\varsigma}=\bigO(1)$, for $q>4$. Let $\log K = \bigO(n^{r})$, 
and assume $1/2-2/q<\varsigma<1-2/q$. Then \eqref{ratebn} implies an admissible rate of $b_n=\bigO(n^\eta)$ such that 
$2r/\varsigma<\eta<\max\{1-5r,(q/2-qr-r-1)/(q/2-\varsigma q/2-1)\}$.
\end{remark}

\begin{remark}[Validity of multiplier block bootstrap under stronger tail assumptions]\label{validboot.exp}
Note that in case with stronger exponential moment conditions on the underlying processes, we shall change the tail probabilities to bound $\delta_{j1}$.

Let $\Phi_{j,\psi_\nu,\varsigma}= \underset{1\leq k\leq K}{\max}\,\underset{q\geq 2}{\sup}\, q^{-\nu}\|X_{jk,\cdot}\vps_{j,\cdot}\|_{q,\varsigma} < \infty $, then according to Theorem 5.2 of \cite{ZW15gaussian}, for all $x>0$, we have
\begin{equation*}
\P(n\delta_{j1} \geq  \frac{x}{\underset{1\leq k_1,k_2\leq K}{\min}\Psi_{j_1k_1}\Psi_{j_2k_2}} )\lesssim K^2 \exp\bigg(- \frac{x^{\gamma}}{4e\gamma(\sqrt{n}b_n\Phi^2_{j,\psi_\nu,0})^\gamma }\bigg),
\end{equation*}
where $\gamma = 1/(2\nu+ 1)$.
This implies that {when 
$x\geq c(\log K)^{1/\gamma}\sqrt{n}b_n\Phi_{j,\psi_\nu,0}^2$, with sufficiently large $c$, the tail probability tends to 0, as $n\to\infty$.} It follows that $\pi(\delta_{j1})\to0$ as $n\to\infty$, given $x= \smallO\{n(\log K)^{-2}\}$. As a result, \eqref{ratebn} will be replaced by
\begin{align*}
&b_n=\smallO\{n^{-1/2}(\log K)^{-2-1/\gamma}\Phi_{j,\psi_\nu,0}^{-2}\}.\notag\\
&\Phi_{j,4,0}\Phi_{j,4,\varsigma}\{b_n^{-1}+\log(n/b_n)/n+(n-b_n)\log b_n/(nb_n)\}(\log K)^2=\smallO(1),\,\text{if } \varsigma=1;\notag\\
&\Phi_{j,4,0}\Phi_{j,4,\varsigma}\{b_n^{-1}+n^{-\varsigma}+(n-b_n)b_n^{-\varsigma+1}/(nb_n)\}(\log K)^2=\smallO(1),\,\text{if } \varsigma<1;\notag\\
&\Phi_{j,4,0}\Phi_{j,4,\varsigma}\{b_n^{-1}+n^{-1}b_n^{-\varsigma+1}+(n-b_n)/(nb_n)\}(\log K)^2=\smallO(1),\,\text{if } \varsigma>1.
\end{align*}
\end{remark}

\begin{remark}[Consistency of the bootstrap statistics with pre-estimated residuals]\label{residuals.remark}
We note that the errors $\varepsilon_{j,t}$ in $Z_{jk}^{[B]}$ (defined in \eqref{mbb2}) are always unobservable. In practice, one can pre-estimate them using a conservative choice of penalty levels and loadings. It is needed to discuss the consistency rate of the bootstrap statistics with the generated errors. Let $\hat Z_{jk}^{[B]}=\frac{1}{\sqrt{n}}\sum_{i=1}^{l_n} e_{j,i}\sum_{l=(i-1)b_n+1}^{ib_n} \hat\vps_{j,l} X_{jk,l}$ denote the feasible bootstrap statistics. We need to show that $\sup_{r\in{\R}}\big| \P \big\{\max\limits_{1\leq k\leq K}(\hat Z_{jk}^{[B]}/\Psi_{jk})\leq r| X_{j,\cdot}, \vps_{j,\cdot}\big\} -\P\big\{\max\limits_{1\leq k\leq K} (\widetilde{Z}_{jk}/\Psi_{jk})\leq r\big\}\big|\to0$, as $n\to\infty$. For $w>0$, it can be decomposed as follows
\begin{eqnarray*}\label{residuals.eq}
&&\sup_{r\in\R}\big| \P \big(\max_{1\leq k\leq K}|\hat Z_{jk}^{[B]}/\Psi_{jk}|\geq r| X_{j,\cdot}, \vps_{j,\cdot}\big) -\P\big(\max_{1\leq k\leq K} |\widetilde{Z}_{jk}/\Psi_{jk}|\geq r\big)\big|\\
&\leq& \P\big(\max_{1\leq k\leq K}|\hat{Z}_{jk}^{[B]}/\Psi_{jk}- Z_{jk}^{[B]}/\Psi_{jk}|\geq w |X_{j,\cdot}, \vps_{j,\cdot}\big) + \sup_{r\in \R}\P\big(\big|\max_{1\leq k\leq K} |\widetilde{Z}_{jk}/\Psi_{jk}| -r\big|\leq w\big)\\
&& + \sup_{r\in\R}\big|\P\big(\max_{1\leq k\leq K}|Z_{jk}^{[B]}/\Psi_{jk}|\geq r|X_{j,\cdot}, \vps_{j,\cdot}\big)- \P\big(\max_{1\leq k\leq K} |\widetilde{Z}_{jk}/\Psi_{jk}|\geq r\big)\big|,
\end{eqnarray*}
where the second term $\sup_{r\in \R}\P\big(\big|\max\limits_{1\leq k\leq K} |\widetilde{Z}_{jk}/\Psi_{jk}| -r\big|\leq w\big)\lesssim w\sqrt{\log K}$ by the anti-concentration bound. Besides, in the proof of Theorem \ref{validboot}, we have demonstrated that under some conditions the third term tends to 0 as $n\to\infty$. Theorem \ref{residuals.theorem} below presents the rate of the first term when exponential moment conditions are satisfied.
\end{remark}

\begin{theorem}\label{residuals.theorem}
Assume $\Phi^X_{j,\psi_\nu}\defeq\underset{1\leq k\leq K}{\max}\|X^2_{jk,\cdot}\|_{\psi_\nu,0}=\underset{1\leq k\leq K}{\max}\underset{q\geq 2}{\sup}\, q^{-\nu}\|X^2_{jk,\cdot}\|_{q,0} < \infty$. Given the exact sparsity assumption \eqref{sparse}, suppose the LASSO estimator $\hat\beta_j$ satisfies $|\hat\beta_j-\beta^0_j|_1\lesssim_{\P}\sqrt{s_j}\rho_n$. Then, for $w\geq c\big[n^{-1/2}s_j\rho_n^2b_n(\Phi^X_{j,\psi_\nu})^2\{\log (K^2)\}^{1/\gamma}c + s_j\rho^2_nb_n\max\limits_{1\leq k\leq K}\|X_{jk,t}\|_4^4\big]^{1/2}$, with sufficiently large $c$, $\gamma=2/(2\nu+1)$, we have
$$\P\big(\max_{1\leq k\leq K}|\hat{Z}_{jk}^{[B]}/\Psi_{jk}- Z_{jk}^{[B]}/\Psi_{jk}|\geq w |X_{j,\cdot}, \vps_{j,\cdot}\big)\to0,\,\text{ as }n\to\infty,$$
\end{theorem}

\begin{proof}[Proof of Theorem \ref{residuals.theorem}]
Observe that
\begin{eqnarray*}
I_n&:=&\P\big(\max_{1\leq k\leq K}|\hat{Z}_{jk}^{[B]}/\Psi_{jk}- Z_{jk}^{[B]}/\Psi_{jk}|\geq w |X_{j,\cdot}, \vps_{j,\cdot}\big)\\
&\leq&\P\big(\max_{1\leq k\leq K}|\hat{Z}_{jk}^{[B]}- Z_{jk}^{[B]}|\geq u |X_{j,\cdot}, \vps_{j,\cdot}\big)\qquad(u:=w\min_{1\leq k\leq K}\Psi_{jk})\\
&=&\P\Big(n^{-1/2} \max_{1\leq k\leq K}\Big|\sum^{l_n}_{i=1} e_{j,i} \sum^{ib_n}_{l= (i-1)b_n+1} (\hat{\vps}_{j,l} - \vps_{j,l})X_{jk,l} \Big|\geq u | X_{j,\cdot}, \vps_{j,\cdot}\Big)\\
&=&\P\Big(n^{-1/2} \max_{1\leq k\leq K}\Big|\sum^{l_n}_{i=1} e_{j,i} \sum^{ib_n}_{l= (i-1)b_n+1}X_{j,l}^\top (\hat\beta_j - \beta^0_j)X_{jk,l} \Big|\geq u | X_{j,\cdot}, \vps_{j,\cdot}\Big)\\
&=&\P\Big(n^{-1/2} |\hat\beta_j - \beta^0_j|_1\max_{1\leq k,k'\leq K}\Big|\sum^{l_n}_{i=1} e_{j,i} \sum^{ib_n}_{l= (i-1)b_n+1}X_{j,l}X_{jk,l} \Big|\geq u | X_{j,\cdot}, \vps_{j,\cdot}\Big).
\end{eqnarray*}
Let $\sigma_{i,kk'}( X_{j,\cdot} ,\vps_{j,\cdot})= n^{-1}s_j\rho^2_n(\sum^{ib_n}_{l = (i-1)b_n+1} X_{jk',l} X_{jk,l})^2$. Applying the Markov inequality yields $I_n\lesssim\underset{1\leq k,k'\leq K}{\max}\sum^{l_n}_{i=1}\sigma_{i,kk'}( X_{j,\cdot} ,\vps_{j,\cdot})/u^2$.

Next, we define $T_{j,kk'}\defeq\sum^{l_n}_{i=1}(\sum^{ib_n}_{l = (i-1)b_n+1}X_{jk,l}X_{jk',l})^2$ and let $\E_0(T_{j,kk'})=T_{j,kk'}-\E(T_{j,kk'})$. Thus we have $\E_0(T_{j,kk'})=\sum_{v=-\infty}^{l_nb_n}\mP_v(T_{j,kk'})$, where the operator $\mP_v(\cdot) \defeq \E(\cdot|\mathcal{F}_v) -\E(\cdot|\mathcal{F}_{v-1})$ produces martingale difference sequences. In addition, by using the Jensen inequality we have $(\sum^{ib_n}_{l = (i-1)b_n+1} X_{jk',l} X_{jk,l})^2\leq b_n\sum^{ib_n}_{l = (i-1)b_n+1}X_{jk',l}^2 X_{jk,l}^2$. Therefore, by applying Theorem 2.1 of \citet{rio2009moment} and Lemma \ref{buck}, it follows that
\begin{align*}
\|{\E}_0(T_{j,kk'})\|_{q/2}^2&\leq(q/2-1)\sum_{v=-\infty}^{l_nb_n}\|\mP_v(T_{j,kk'})\|_{q/2}^2\notag\\
&\leq (q/2-1)^{3/2}b_n^2\sum_{v=-\infty}^{l_nb_n}\sum^{l_n}_{i=1}\Big\|\mP_v\Big(\sum^{ib_n}_{l = (i-1)b_n+1}X_{jk',l}^2 X_{jk,l}^2\Big)\Big\|_{q/2}^2\notag\\
&\leq (q/2-1)^{2}b_n^2\sum_{v=-\infty}^{l_nb_n}\sum^{l_n}_{i=1}\sum^{ib_n}_{l = (i-1)b_n+1}\|\mP_v(X_{jk',l}^2 X_{jk,l}^2)\|_{q/2}^2\notag\\
&\leq (q-2)^{2}l_nb_n^3(2q)^{2\nu}\|X_{jk',l}^2\|_{\psi_\nu,0}^2 \|X_{jk,l}^2\|_{\psi_\nu,0}^2.
\end{align*}
Following the similar argument as in the proof of Theorem 3 in \citet{wu2016performance}, we have
$$\E[\exp(\tau|{\E}_0(T_{j,kk'})/(l_n^{1/2}b_n^{3/2})|^\gamma]\leq1+C_\gamma(1-\tau/\tau_0)^{-1/2}\tau/\tau_0,$$
with $\tau_0=(2e\gamma\|X_{jk',l}^2\|_{\psi_\nu,0}^\gamma \|X_{jk,l}^2\|_{\psi_\nu,0}^\gamma)^{-1}$ and $\gamma=2/(2\nu+1)$. Finally, by letting $\tau=\tau_0/2$, and applying the Markov inequality and the Bonferroni inequality, we obtain
$$\P\big(\max_{1\leq k,k'\leq K}|{\E}_0(T_{j,kk'})|\geq x\big)\lesssim K^2\exp\bigg(-\frac{x^\gamma}{4e\gamma\{l_n^{1/2}b_n^{3/2}(\Phi^X_{j,\psi_\nu})^2\}^\gamma}\bigg) .$$
It follows that $I_n\lesssim_{\P} (n^{-1}s_j\rho_n^2x+ s_j\rho^2_nb_n\max\limits_{1\leq k\leq K}\|X_{jk,t}\|_4^4)/u^2$ (the $\lesssim_{\P}$ only depends on $q$ and $\gamma$), for $x\geq cl_n^{1/2}b_n^{3/2}(\Phi^X_{j,\psi_\nu})^2\{\log (K^2)\}^{1/\gamma}$ with sufficiently large $c$, and thus the desired conclusion holds.
\end{proof}

\subsection{Proofs of Joint Equation Estimation}

\begin{proof}[Proof of Theorem \ref{validboot.joint}] 
Analogue to the proof of Theorem \ref{validboot}, the conclusions are implied by
\beq
\P\bigg(n { \delta_1}\geq \big(\underset{1\leq k_1,k_2\leq K, 1\leq j_1,j_2\leq J}{\min}\Psi_{j_1k_1}\Psi_{j_2k_2}\big)^{-1}x\bigg) \lesssim \frac{JK F_{\varsigma}\Gamma_{2q,\varsigma}^{q}}{x^{q/2}} + (JK)^2 \exp\bigg(-\frac{C x^2}{n b_n \Phi^4_{8,\varsigma}}\bigg),\notag
\eeq
for $x\geq  n^{1/2} b_n^{1/2}\Phi_{2q,\varsigma}^2$ and all large $n$, where
\beq
{\delta_1}\defeq \max_{1\leq k_1,k_2\leq K, 1\leq j_1,j_2\leq J} \bigg|\frac{\sum_{i=1}^{l_n}S_{j_1k_1,i}S_{j_2k_2,i}}{\Psi_{j_1k_1}\Psi_{j_2k_2}} - \frac{l_n\E (S_{j_1k_1,i}S_{j_2k_2,i})}{\Psi_{j_1k_1}\Psi_{j_2k_2}}\bigg|.\notag
\eeq
In particular, {when $x\geq\max\big\{{n^{1/2} b_n^{1/2}\Phi_{2q,\varsigma}^2, cn^{1/2} b_n^{1/2} \{\log (JK)\}^{1/2}\Phi_{8,\varsigma}^2}, c(JK)^{2/q} F^{2/q}_{\varsigma} \Gamma^2_{2q, \varsigma}\big\}$, where $c$ is sufficiently large, 
the tail probability tends to 0, as $n\to\infty$.}

{By similar proof to that of Theorem \ref{validboot}, it follows that $\tilde\rho_n\to0$ as $n\to\infty$, given $x=\smallO[n\{\log(KJ)\}^{-2}]$ and $\Phi_{4,0}\Phi_{4,\varsigma} v(b_n) = \smallO\{(\log KJ)^{-2}\}$, 
which imply the following conditions on $b_n$:
	\begin{align}\label{ratebn.joint}
	&b_n=\smallO[n\{\log(KJ)\}^{-4}\Phi^{-4}_{2q,\varsigma}\wedge n\{\log(KJ)\}^{-5}\Phi_{8,\varsigma}^{-4}],\, F_\varsigma=\smallO[n^{q/2}\{\log(KJ)\}^{-q}(KJ)^{-1}\Gamma_{2q,\varsigma}^{-q}].\notag\\
	&\Phi_{4,0}\Phi_{4,\varsigma}\{b_n^{-1}+\log(n/b_n)/n+(n-b_n)\log b_n/(nb_n)\}\{\log (KJ)\}^2=\smallO(1),\,\text{if } \varsigma=1;\notag\\
	&\Phi_{4,0}\Phi_{4,\varsigma}\{b_n^{-1}+n^{-\varsigma}+(n-b_n)b_n^{-\varsigma+1}/(nb_n)\}\{\log (KJ)\}^2=\smallO(1),\,\text{if } \varsigma<1;\notag\\
	&\Phi_{4,0}\Phi_{4,\varsigma}\{b_n^{-1}+n^{-1}b_n^{-\varsigma+1}+(n-b_n)/(nb_n)\}\{\log (KJ)\}^2=\smallO(1),\,\text{if } \varsigma>1.
	\end{align}
Recall that $F_{\varsigma} = n$, for $\varsigma >1-2/q$; $F_{\varsigma} = l_nb_n^{q/2-\varsigma q/2}$, for $1/2-2/q<\varsigma<1-2/q$; $F_{\varsigma} = l_n^{q/4-\varsigma q/2}b_n^{q/2-\varsigma q/2} $, for $\varsigma<1/2-2/q$.}

The rest of the proof is similar to that of Theorem \ref{validboot} and thus is omitted.

\end{proof}


{
	\begin{proof}[Proof of Theorem \ref{post}]
		For any $\delta,\tilde{\delta}\in\R^K$ in $\mathcal{G}_{\widetilde{T}_j}$, we have
		\begin{align*}
		\bigg|{\E}_n \bigg\{\vps_{j,t}\bigg(\frac{X_{j,t}^{\top}\delta}{|\delta|_{j,pr}}-\frac{X_{j,t}^{\top}\tilde{\delta}}{|\tilde{\delta}|_{j,pr}}\bigg)\bigg\}\bigg| &= \bigg|{\E}_n \bigg[\vps_{j,t}\bigg\{\frac{X_{j,t}^{\top}(\delta-\tilde\delta)}{|\delta|_{j,pr}} + \frac{X_{j,t}^{\top}\tilde{\delta}}{|\delta|_{j,pr}} -\frac{X_{j,t}^{\top}\tilde{\delta}}{|\tilde{\delta}|_{j,pr}}\bigg\}\bigg]\bigg|\\
		&\leq \bigg|{\E}_n \bigg[\vps_{j,t}\bigg\{\frac{X_{j,t}^{\top}(\delta-\tilde\delta)}{|\delta|_{j,pr}}\bigg\}\bigg]\bigg| + \bigg|{\E}_n \bigg[\vps_{j,t}\bigg\{\frac{X_{j,t}^{\top}\tilde{\delta}}{|\delta|_{j,pr}} -\frac{X_{j,t}^{\top}\tilde{\delta}}{|\tilde{\delta}|_{j,pr}}\bigg\}\bigg]\bigg|\\
		&\leq ({\E}_n\vps_{j,t}^2)^{1/2}\bigg\{{\E}_n\bigg|\frac{X_{j,t}^{\top}(\delta-\tilde\delta)}{|\delta|_{j,pr}}\bigg|^2\bigg\}^{1/2} + ({\E}_n\vps_{j,t}^2)^{1/2}\bigg(\frac{|\tilde\delta|_{j,pr} - |\delta|_{j,pr}}{|\delta|_{j,pr}}\bigg)\\
		&\leq 2\sigma \mu_j(p)|\delta- \tilde{\delta}|_2.
		\end{align*}
		Then by following the proof of Lemma 5 (Step 2) in \cite{belloni2009least}, we have $\sup_{\mathcal Q}\mathcal{N}(\epsilon,\mathcal{G}_{\widetilde{T}_j}, \|\cdot\|_{\mathcal Q,1})\lesssim (6 \mu_j(p)\sigma/ \epsilon)^{s+p}$. And it follows that $|F_{j,p}|\lesssim {{K}\choose {p}}(6 \mu_j(p)\sigma/ \epsilon)^{s+p}$.
		
		Moreover, it is not hard to see that $\sup_{f \in \mathcal{F}_{j,p}} |G_n(f)|\leq 2\sqrt{n} \epsilon + \sup_{f \in F_{j,p}} |G_n(f)|$. Let $\psi=\max_{f\in F_{j,p}}\psi_f$ (assume $\psi$ is bounded by constant) and applying the Gaussian approximation results on the vector $G_n(f)/\psi_f$ (given \hyperref[A6]{(A6)}), we have
		\begin{align*}
		\P\big\{\sup_{f\in F_{j,p}} |G_n(f)| \geq \kappa_n/2\big\} &\leq \P\big\{\sup_{f\in F_{j,p}} |G_n(f)/\psi_f| \geq \kappa_n/(2\psi)\big\}\\
		&\leq 2|F_{j,p}|\{1-\Phi(\kappa_n/(2\psi))\} + d_n\\
		&\leq 2K^p(6 \mu_j(p)\sigma/ \epsilon)^{s+p}\exp\{-\kappa_n^2/(8\psi^2)\}\{\kappa_n/(2\psi)\}^{-1} + d_n,
		\end{align*}
		as ${{K}\choose {p}}\leq K^p$. Therefore, for $\kappa_n=\psi\sqrt{p\log K+(p+s)\{\log  (6 \mu_j(p) \sigma)+\log n/2\}}$ and {$\epsilon=c\sqrt{p\log K+(p+s)\log  (6 \mu_j(p) \sigma)}(4 \sqrt{n})^{-1}$, it follows that $\sup_{f\in \mathcal F_{j,p}} |G_n(f)| \lesssim_{\P} \kappa_n$ with sufficiently large $c$} (note that $d_n\to0$ with a polynomial rate as $n\to\infty$).
		
		The rest of the proof is a direct application of Theorem 5 of \cite{belloni2009least} by inserting the bound for $\lambda^0(1-\alpha)$ \eqref{lambdabound.eq} provided in Corollary \ref{lambdabound.joint}, and thus is omitted.	
	\end{proof}
}

\subsection{Plausibility of RE and RSE Conditions}
{Define the $s$-sparse sphere as $F_{\delta}= \{\delta: |\delta|_0\leq s, |\delta|_2 = 1\}$.
According to \cite{rudelson2012reconstruction}, the $\epsilon$-covering number of $F_{\delta}$ w.r.t. the Euclidean metric is $l = \exp(s\log(3eK/m\epsilon))$, with $m\geq1$. This is the cardinality of the $\epsilon$-cover set $\Pi_{\delta}$ of $F_{\delta}$. Moreover, for any point $\delta \in F_{\delta}$, let $\pi_{\delta}$ denote the closest point to $\delta$ within $\Pi_{\delta}$. Let $\breve X^{\pi(\delta)}_{j,t}\defeq \{\widetilde X_{j,t}^\top\pi(\delta)\}^2-n^{-1}\pi(\delta)^\top\E\{X_{j,t}X_{j,t}^\top\}\pi(\delta)$, where $\widetilde{X}_{j}\defeq n^{-1/2}X_{j}$ and $X_{j}(n\times K)$ is a matrix of $X_{j,t}$. Note that $\breve X^{\pi(\delta)}_{j,t}$ is a vector of the cardinality of $\Pi_{\delta}$.}

\begin{theorem}[Plausibility of RE and RSE]\label{ratere}
	For any $j=1,\ldots,J$, suppose the vectors $X_{j,t}$ of length $K$ satisfy
	\begin{equation*}
	0< \kappa \leq \min_{|\delta|_0\leq s, |\delta|_1 = 1}\delta^{\top}\E(X_{j,t}X_{j,t}^{\top})\delta \leq \max_{|\delta|_0\leq s, |\delta|_1 = 1}\delta^{\top}\E(X_{j,t}X_{j,t}^{\top})\delta \leq \psi < \infty,
	\end{equation*}
	where $\psi$ and $\kappa$ are positive constants. {Given $\breve\Phi_{2,\varsigma} \defeq \underset{\pi(\delta) \in \Pi_{\delta}}{\max}\|\breve{X}^{\pi(\delta)}_{j,\cdot}\|_{2,\varsigma}<\infty$, and for $q>2$, $\big\|\underset{\pi(\delta) \in \Pi_{\delta}}{\max}|\breve X^{\pi(\delta)}_{j,\cdot}|\big\|_{q,\varsigma}<\infty$,
	$$n^{-1/2}(\log l)^{1/2}\breve\Phi_{2,\varsigma}+n^{-1}r_\varsigma(\log l)^{3/2}\big\|\underset{\pi(\delta) \in \Pi_{\delta}}{\max}|\breve X^{\pi(\delta)}_{j,\cdot}|\big\|_{q,\varsigma}=\smallO(1),$$}
	where $r_{\varsigma} = n^{1/q}$ for $\varsigma > 1/2-1/q$ and $r_{\varsigma} = n^{1/2-\varsigma}$ for $\varsigma < 1/2-1/q$, then {the RE and RSE conditions hold with probability $1- \smallO(1)$, with $p+s_j \leq s$}.
\end{theorem}

\begin{proof}[Proof of Theorem \ref{ratere}]
	
	Firstly, we need to check the implication of the population matrix. We know that $\delta^{\top}X_{j}^{\top}X_{j}\delta/n = |\widetilde{X}_{j}\delta|_2^2$.
	Then we have the following inequalities for any point $\delta\in F_{\delta}$,
	\begin{equation}\label{inq}
	-|\widetilde{X}_{j}\{\delta- \pi(\delta)\}|_2 + |\widetilde{X}_{j}\pi(\delta)|_2 \leq |\widetilde{X}_{j}\delta|_2\leq |\widetilde{X}_{j}\{\delta- \pi(\delta)\}|_2 + |\widetilde{X}_{j}\pi(\delta)|_2.
	\end{equation}
	We first check the right hand side of \eqref{inq}. Define $\|\widetilde{X}_{j}\|_{2, F_{\delta}} \defeq \underset{\delta\in F_{\delta}}{\sup}|\widetilde{X}_{j}\delta|_2$. As indicated in the proof of Theorem 16 in \citet{rudelson2012reconstruction}, we have $|\widetilde{X}_{j}\{\delta - \pi(\delta)\}|_2 \leq \epsilon \|\widetilde{X}_{j}\|_{2, F_{\delta}}$. To bound $\underset{\pi(\delta) \in \Pi_{\delta}}{\max} |\widetilde{X}_{j} \pi(\delta)|_2$, we invoke the tail probability inequality in Lemma \ref{tail}, which gives 
	\begin{align*}
	\P\big(\max_{\pi(\delta) \in \Pi_{\delta}}\big|\sum_{t=1}^n\breve X^{\pi(\delta)}_{j,t}\big|\geq x\big)&=\P\big[\max_{\pi(\delta) \in \Pi_{\delta}}\big||\widetilde{X}_{j}\pi(\delta)|_2^2 - \pi(\delta)^\top\E\{X_{j,t}X_{j,t}^\top\}\pi(\delta)\big|\geq x\big]\to 0, \text{as } n\to\infty,
	\end{align*}
	{if $x\geq c\sqrt{n \log l}\breve\Phi_{2,\varsigma} + cr_{\varsigma}(\log l)^{3/2}\big\|\underset{\pi(\delta) \in \Pi_{\delta}}{\max}|\breve X^{\pi(\delta)}_{j,\cdot}|\big\|_{q,\varsigma}$, for sufficiently large $c$.}
	
	Therefore, given $\kappa,\psi>0$, $\kappa- x_n \leq |\widetilde{X}_{j}\pi(\delta)|_2^2 \leq x_n +\psi$ holds with probability $1-\smallO(1)$ for all $\pi(\delta) \in \Pi_{\delta}$, where 
	{$x_n\defeq cn^{-1/2}(\log l)^{1/2}\breve\Phi_{2,\varsigma}+ cn^{-1}r_\varsigma(\log l)^{3/2}\big\|\underset{\pi(\delta) \in \Pi_{\delta}}{\max}|\breve X^{\pi(\delta)}_{j,\cdot}|\big\|_{q,\varsigma}=\smallO(1)$.}
	
	Hence, the right inequality in \eqref{inq} leads to $|\widetilde{X}_{j}\delta|_2 \leq \epsilon\|\widetilde{X}_{j}\|_{2, F_{\delta}}+ \sqrt{x_n}+ \sqrt{\psi}$. Taking the supremum over all $\delta\in F_\delta$ on both sides shows that $\underset{\delta \in F_{\delta}}{\sup} |\widetilde{X}_{j}\delta|_2 \leq (\sqrt{x_n}+ \sqrt{\psi})/ (1-\epsilon)$ with probability $1- \smallO(1)$. Moreover, by the left hand side of \eqref{inq}, we have $|\widetilde{X}_{j}\delta|_2\geq \sqrt{\kappa- x_n} - \epsilon(\sqrt{x_n}+ \sqrt{\psi})/ (1-\epsilon)$, with probability $1- \smallO(1)$.
	
	Collecting the results together, we have shown that for all $\delta\in F_\delta$,
	\beq
	\sqrt{\kappa- x_n} - \frac{\epsilon(\sqrt{x_n}+ \sqrt{\psi})}{(1-\epsilon)}\leq|\widetilde{X}_{j}\delta|_2\leq \frac{\sqrt{x_n}+ \sqrt{\psi}}{(1-\epsilon)},
	\eeq
	with probability $1- \smallO(1)$.

{Let $c^*(s)= \max_{\delta \in F_{\delta}} |\widetilde{X}_{j}\delta|_2$, $c_*(s)= \min_{\delta \in F_{\delta}} |\widetilde{X}_{j}\delta|_2$, with properly chosen $\epsilon$, $c^*(s), c_*(s)$ are bounded from above and below, and the desired results follow by the fact
$\kappa^2_{j}(p) \geq c_*(s_j+p)$, $\phi_j(p)\leq c^*(s_j +p)$, with $s_j + p \leq s$.}
\end{proof}

\subsection{Proofs of Simultaneous Inference}

\subsubsection{Some Useful Lemmas}

\begin{lemma}[\cite{burkholder1988sharp,rio2009moment}]\label{buck} Let $q>1$, $q'=\min(q,2)$. Let $M_n = \sum^n_
{t=1} \xi_t$; where $\xi_t \in \mathcal{L}^{q}$ (i.e., $\|\xi_t\|_q<\infty$) are martingale differences. Then
\begin{equation*}
\|M_n\|_q^{q'} \leq K^{q'}_q\sum_{t=1}^n\|\xi_t\|^{q'}_q \quad \text{where} \quad K_q = \max((q -1)^{-1},\sqrt{q-1}).
\end{equation*}
\end{lemma}

\begin{lemma}[Freedman's inequality]\label{freedman}
Let $\{\xi_{a,i}\}_{i=1}^n$ be a martingale difference sequence w.r.t. the filtration $\{\mathcal{F}_i\}_{i=1}^n$. Let $V_a = \sum^n_{i=1}\E(\xi^2_{a,i}| \mathcal{F}_{i-1})$ and $M_a = \sum^n_{i=1} \xi_{a,i} $. Then for all $x,u,v>0$, we have,
\begin{equation}
\P\big(\max_{a \in \mathcal{A}} |M_a|\geq x\big) \leq \sum^{n}_{i=1}\P\big(\max_{a\in \mathcal{A}} |\xi_{a,i}|\geq u\big)+ 2 \P\big(\max_{a \in \mathcal{A} } V_a \geq v \big)+ 2|\mathcal{A}| e^{-x^2/(2zu+ 2v)},
\end{equation}
where $\mathcal A$ is an index set with $|\mathcal A|<\infty$.
\end{lemma}
Lemma \ref{freedman} is a maximal form of Freedman's inequality \citet{freedman1975tail}.

\begin{lemma}[Theorem 6.2 of \cite{ZW15gaussian} Tail probabilities for high dimensional partial sums]\label{tail}
For a zero-mean $p$-dimensional random variable $X_t\in {\R}^p$, let $S_n = \sum^{n}_{t=1}X_t$ and assume that $\||X_\cdot|_{\infty}\|_{q,\varsigma}< \infty, $ where $q>2$ and $\varsigma\geq0$, and $\Phi_{2, \varsigma} = \underset{1\leq j \leq p}{\max} \|X_{j,\cdot}\|_{2, \varsigma}<\infty$. \\
i) If $\varsigma> 1/2 -1/q$, then for $x \gtrsim \sqrt{n \log p} \Phi_{2,\varsigma}+ n^{1/q}(\log p)^{3/2}\||X_\cdot|_{\infty}\|_{q,\varsigma}$,
$$\P(|S_{n}|_{\infty}\geq x )\leq \frac{C_{q,\varsigma}n (\log p)^{q/2} \||X_{\cdot}|_{\infty}\|^q_{q,\varsigma}}{x^q}+ C_{q,\varsigma}\exp\bigg(\frac{-C_{q,\varsigma}x^2}{n\Phi^2_{2,\varsigma}}\bigg).$$
ii) If $0<\varsigma< 1/2 -1/q$, then for $x \gtrsim \sqrt{n \log p} \Phi_{2,\varsigma}+ n^{1/2-\varsigma}(\log p)^{3/2}\||X_\cdot|_{\infty}\|_{q,\varsigma}$,
$$\P(|S_{n}|_{\infty}\geq x )\leq \frac{C_{q,\varsigma}n^{q/2-\varsigma q} (\log p)^{q/2} \||X_{\cdot}|_{\infty}\|^q_{q,\varsigma}}{x^q}+ C_{q,\varsigma}\exp\bigg(\frac{-C_{q,\varsigma}x^2}{n\Phi^2_{2,\varsigma}}\bigg).$$
\end{lemma}

\begin{lemma}[Tail probabilities for high dimensional partial sums with strong tail assumptions] \label{lemma.exp}
For a zero-mean $p$-dimensional random variable $X_t\in{\R}^p$, let $S_n = \sum^{n}_{t=1}X_t$ and assume that $\Phi_{\psi_\nu,\varsigma}= \underset{1\leq j\leq p}{\max}\,\underset{q\geq 2}{\sup}\, q^{-\nu}\|X_{j,\cdot}\|_{q,\varsigma} < \infty $ for some $\nu\geq0$, and let $\gamma = 2/(1+ 2\nu)$. Then for all $x>0$, we have
$$\P(|S_n|_\infty\geq x) \lesssim p \exp\{-C_{\gamma}x^{\gamma}/(\sqrt{n}\Phi_{\psi_\nu,0})^{\gamma}\},$$
where $C_\gamma$ is a constant only depends on $\gamma$.
\end{lemma}
{Lemma \ref{lemma.exp} follows from Theorem 3 of \citet{wu2016performance} and applying the Bonferroni inequality. }

\begin{lemma}[Theorem 1 of \cite{MR2988107}]\label{lem:cltwbw}
Denote $Y_t=f(\F_t),$ where $f$ is some measurable function. Let $S_n=\sum_{t=1}^n Y_t,$ and $\delta_{\varsigma,t}=\|Y_t-Y_{t}^\ast\|_\varsigma$.	If $\E(Y_i)=0,$ $\sum_{t=0}^\infty\delta_{\varsigma,t} <\infty,$ some $\varsigma\geq 2$, and $\sigma_n^2\defeq\E(S_n^2)\rightarrow\infty,$
then
$$\sigma_n^{-1} S_n \stackrel{\mathcal{L}}{\rightarrow} \N(0,1).$$
\end{lemma}

\begin{lemma}\label{ratedeltan}
Under the same conditions as in Theorem \ref{bahadur}, let $\tilde{\beta}_{jk}$ be any estimator such that $|\tilde{\beta}_{jk}- \beta^0_{jk}|\leq C \rho_n$ with probability $1-\smallO(1)$. Then we have
\beq\label{bound.D}
n^{-1}\max_{(j,k)\in G}\Delta_n \lesssim \smallO(n^{-1/2}g_n^{-1}),
\eeq
holds with probability $1-\smallO(1)$, where $\Delta_n\defeq n^{1/2}G_n\{\psi_{jk}(Z_{j,t},\tilde\beta_{jk},\hat h_{jk})-\psi_{jk}(Z_{j,t},\beta^0_{jk},h^0_{jk})\}$.
\end{lemma}

\begin{proof}[Proof of Lemma \ref{ratedeltan}]
As indicated in the proof of Theorem 2 in \cite{BCK15Bio}, the entropy $\operatorname{ent}(\epsilon,\widetilde{\mathcal F})\leq cs\log(a_n/\epsilon)$ for the function class $\widetilde{\mathcal{F}} = \{z\mapsto\psi_{jk}\{z,\beta,\tilde h(x_{j(-k)})\} - \psi_{jk}\{z,\beta^0_{jk},h^0_{jk}(x_{j(-k)})\}: (j,k)\in G, \beta\in\mathcal{B}_{jk},|\beta-\beta^0_{jk}|\leq C\rho_n,\tilde h\in\mathcal{H}_{jk}\}$, which has $2F$ as the envelope (the definition of $F$ is given in \hyperref[C6]{(C6)}). Therefore, for any $f\in\widetilde{\mathcal F}$, there exists a set $F_n$ such that $\min_{f'\in F_n}\|f-f'\|_{\mathcal Q,2}\leq\tilde\epsilon$, where $\tilde\epsilon\defeq\epsilon\|2F\|_{\mathcal Q,2}$, and the cardinality of the set $|F_n|=(a_n/\epsilon)^{cs}$. Then we have
\begin{equation*}
\sup_{f\in\widetilde{\mathcal F}}\bigg|\sum_{t=1}^n\big[f-\pi(f)-\E\{f-\pi(f)\}\big]\bigg|\leq2\tilde\epsilon n,
\end{equation*}
where $\pi(f)\defeq\arg\,\underset{f'\in F_n}{\min}\|f-f'\|_{\mathcal Q,2}$. Hence, with probability $1-\smallO(1)$,
\begin{align}\label{decomp}
\max_{(j,k)\in G}\Delta_n&\leq n^{1/2}\sup_{f\in\widetilde{\mathcal{F}}}|G_n(f)|\notag\\
&=n\sup_{f\in\widetilde{\mathcal F}}\big|[\E{_n}(f)-\E{_n}\{\pi(f)\}-\E (f)+\E\{\pi(f)\}]+[\E{_n}\{\pi(f)\}-\E\{\pi(f)\}]\big|\notag\\
&\leq2n\tilde\epsilon + n\max_{f\in F_n}|\E{_n}(f)-\E(f)|\notag\\
&\leq2n\tilde\epsilon + n\max_{f\in F_n}|\E{_n}(f) - \E{_n}\E(f|\mF_{t-1},X_{j(-k),t})| + n\max_{f\in F_n}|\E{_n}\E(f|\mF_{t-1},X_{j(-k),t}) - \E(f)|\notag\\
&=:2n\tilde\epsilon + K_n + N_n
\end{align}

Next, we look for the bounds for $K_n$ and $N_n$, respectively. Note the summands of $K_n$ form martingale differences. Consider the function set $F_n$, for each $f\in F_n$, let $\varphi_{l,t}\defeq f(z_t)$ and $\tilde \varphi_{l,t}\defeq \varphi_{l,t} - \E(\varphi_{l,t}|\mF_{t-1},X_{j(-k),t})$. Note that $\varphi_t$ and $\tilde \varphi_{t}$ are vectors of length $|F_n|= (a_n/\epsilon)^{cs}$. For $l=1,\ldots,|F_n|$, the dependence adjusted norm of $\tilde\varphi_{l,t}$ obeys that $\|\tilde\varphi_{l,\cdot}\|_{2,\varsigma} \leq 2\|\tilde\varphi_{l,t}\|_2\lesssim 8\|\varphi_{l,t}\|_2$. Moreover, by \hyperref[C3]{(C3)} and \hyperref[C5]{(C5)}, 
we have $\|\varphi_{l,t}\|_2^2 \lesssim L_{2n} \rho^2_{n,\upsilon}$. In particular, for the mean regression case $\rho_{n,\upsilon}=\rho_ns$, while $\rho_{n,\upsilon}=\rho_n^{1/2}$ for the median regression case.

Apply the tail inequality as in Lemma \ref{tail} to the vector $\tilde{\varphi}_{t}$. As $\underset{1\leq l\leq|F_n|}{\max} \|\tilde\varphi_{l,\cdot}\|_{2,\varsigma} \lesssim \sqrt{L_{2n}} \rho_{n,\upsilon}$ and $\|\underset{1\leq l\leq|F_n|}{\max}\tilde\varphi_{l,\cdot}\|_{q,\varsigma} \lesssim \|4F(z_t)\|_q$ (by \hyperref[C6]{(C6)}), 
 then we can see that with probability greater than $1- \bigO(|F_n|^{-1} + (\log|F_n|)^{-q})$,
{\begin{eqnarray*}
	K_n& \lesssim & \sqrt{n s\log (a_n/\epsilon)}\max_{1\leq l\leq|F_n|}\|\tilde{\varphi}_{l,\cdot}\|_{2,\varsigma }+ r_{\varsigma} \{s\log (a_n/\epsilon)\}^{3/2}
	\|\underset{1\leq l\leq|F_n|}{\max}\tilde\varphi_{l,\cdot}\|_{q,\varsigma}\notag\\
	&\leq&  \sqrt{n L_{2n} s\log (a_n/\epsilon)}\rho_{n,\upsilon}+  r_{\varsigma}\{s\log(a_n/\epsilon)\}^{3/2}\|8F(z_t)\|_q,
\end{eqnarray*}}
{given $\epsilon$ is sufficiently small.}
Hence, we have
\beq\label{bound.K}
K_n \lesssim_{\P} \rho_{Kn},
\eeq
where $\rho_{K_n} \defeq r_{k1} + r_{\varsigma} r_{k2}$ with {$r_{k1} \defeq \sqrt{n L_{2n} s\log (a_n/\epsilon)}\rho_{n,\upsilon}$, $r_{k2} \defeq \{s\log (a_n/\epsilon)\}^{3/2}\| 8F(z_t)\|_q $} and $r_{\varsigma} = n^{1/q}$ for $\varsigma > 1/2-1/q$ and $r_{\varsigma} = n^{1/2-\varsigma}$ for $\varsigma < 1/2-1/q$.

Then we handle the term $N_n$. Again consider the function set $F_n$, for each $f\in F_n$, let $\breve \varphi_{l,t}\defeq \E(\varphi_{l,t}|\mF_{t-1},X_{j(-k),t}) - \E(\varphi_{l,t})$, where $\varphi_{l,t}= f(z_t)$. Then
\begin{equation*}
N_n \leq\max_{1\leq l\leq|F_n|}|\sum_{t=1}^n\breve{\varphi}_{l,t}|.
\end{equation*}
Moreover, for $l=1,\ldots,|F_n|$, there is a function $g$ corresponding to each $f\in F_n$ such that $\breve\varphi_{l,t}=g(z_t,\beta,\tilde h)$, where $\beta\in\mathcal{B}_{jk}, |\beta-\beta^0_{jk}|\leq C\rho_n, \tilde h\in\mathcal{H}_{jk}, (j,k)\in G$. By the mean value theorem and the continuity of the function $g$, we have
\begin{align*}
g(Z_{j,t},\beta,\tilde h)=&\,\partial_\beta g(Z_{j,t},\bar\beta,\tilde h)(\beta-\beta^0_{jk})\\
&+\sum_{m=1}^2\partial_{h_m} g(Z_{j,t},\beta,\bar h)\{\tilde h_m(X_{j(-k),t})-h^0_{jk,m}(X_{j(-k),t})\},
\end{align*}
where $(\bar\beta,\bar h(\cdot))$ is the corresponding point which joins the line segment between $(\beta,\tilde h(\cdot))$ and $(\beta^0_{jk},h^0_{jk}(\cdot))$.
Then
\begin{align*}
\underset{1\leq l\leq|F_n|}{\max}\sum_{t=1}^n\breve{\varphi}_{l,t}=&\max_{\bar\beta\in F^{\beta}_n }\sum_{t=1}^{n}\partial_\beta g(Z_{j,t},\bar\beta,\tilde h)(\beta-\beta^0_{jk})\\
&+ \max_{\bar h\in F^h_n }\sum_{m=1}^2\sum_{t=1}^n\partial_{h_m} g(Z_{j,t},\beta,\bar h)\{\tilde h_m(X_{j(-k),t})-h^0_{jk,m}(X_{j(-k),t})\},
\end{align*}
where $F^\beta_n$ and $F^{\tilde h}_n$ collect all the points of $\beta$ and $\tilde h$ according to $F_n$, respectively.

Recall that in our linear model setting, $h^0_{jk}(X_{j(-k),t})=(X_{j(-k),t}^\top\beta^0_{j(-k)},X_{j(-k),t}^\top\gamma^0_{j(-k)})^\top=(X_{j(-k),t}^\top\theta^0_{jk,1},X_{j(-k),t}^\top\theta^0_{jk,2})^\top$, and $\tilde h(X_{j(-k),t})=(X_{j(-k),t}^\top\tilde\theta_{jk,1},X_{j(-k),t}^\top\tilde\theta_{jk,2})^\top$, where $\theta^0_{jk,m}$ and $\tilde\theta_{jk,m}$ ($m=1,2$) are vectors of length $K-1$. Let $T^0_{jk}\defeq\{1\leq \ell\leq K-1:\theta^0_{jk,1,\ell}\neq0,\theta^0_{jk,2,\ell}\neq0\}$, $\widetilde T_{jk}\defeq\{1\leq \ell\leq K-1:\tilde\theta_{jk,1,\ell}\neq0,\tilde\theta_{jk,2,\ell}\neq0\}$, and $\breve X_{t}^{jk}\defeq\operatorname{vec}\{(X_{j(-k),t,\ell})_{\ell\in T^0_{jk}\bigcup\widetilde T_{jk}}\}$. Now we apply Lemma \ref{tail} on $\sum_{t=1}^n\partial_{h_m} g(Z_{j,t},\beta,\bar h)\{\tilde h_m(X_{j(-k),t})-h^0_{jk,m}(X_{j(-k),t})\}$ and \\
$\sum_{t=1}^n\partial_\beta g(Z_{j,t},\bar\beta,\tilde h)(\beta-\beta^0_{jk})$. To this end, we define the following quantities:
{\begin{align}\label{Phi.h}
&\Phi^h_{m,2,\varsigma}\defeq \max_{\bar h\in F_n^{\tilde h}}\big\||\breve X_{\cdot}^{jk}\partial_{h_m} g(Z_{j,\cdot},\beta,\bar h)|_\infty\big\|_{2,\varsigma},\,
\Omega^h_{m,q,\varsigma}\defeq\big\|\max_{\bar h\in F_n^{\tilde h}}|\breve X_{\cdot}^{jk}\partial_{h_m} g(Z_{j,\cdot},\beta,\bar h)|_\infty\big\|_{q,\varsigma}.
\end{align}}

Let $\chi_t^m\defeq\partial_{h_m} g(Z_{j,t},\beta,\bar h)\{\tilde h_m(X_{j(-k),t})-h^0_{jk,m}(X_{j(-k),t})\}$ and define the projector operator $\mP_l(\chi^m_t) \defeq \E(\chi^m_t|\mathcal{F}_l) -\E(\chi^m_t|\mathcal{F}_{l-1})$. According to Theorem 1(i) of \citet{wu2005nonlinear}, it is not hard to see that $\|\chi^m_\cdot\|_{q, \varsigma} \lesssim \sup_{d\geq0} (d+1)^{\varsigma}\sum^{\infty}_{t=d} \|\mathcal{P}_0(\chi^m_t)\|_q$, for $m=1,2$. {Moreover, as $|\tilde{\theta}_{jk,m}-\theta^0_{jk,m}|_1\lesssim \sqrt{s_j} \rho_n\leq\sqrt{s}\rho_n$, we have
\begin{align*}
\|\mathcal{P}_0(\chi^m_t)\|_q &\leq \big(\E[\mathcal{P}_0\{|\partial_{h_m} g(Z_{j,t},\beta,\bar h)\breve X_{t}^{jk}|_{\infty}\}|\tilde{\theta}_{jk,m}-\theta^0_{jk,m}|_1]^q\big)^{1/q}\\
&\lesssim\sqrt{ s} \rho_n \big(\E[\mathcal{P}_0\{|\partial_{h_m} g(Z_{j,t},\beta,\bar h)\breve X_{t}^{jk}|_{\infty}\}]^{q}\big)^{1/q}.
\end{align*}}
It follows that $\|\chi^m_\cdot\|_{q, \varsigma}  \lesssim \sqrt{ s} \rho_n\big\||\breve X_{\cdot}^{jk}|_\infty|\partial_{h_m} g(Z_{j,\cdot},\beta,\bar h)|\big\|_{q, \varsigma}$. Then applying the tail probability bounds in Lemma \ref{tail} yields with probability approaching $1$,
\begin{align*}
\max_{\bar h\in F_n^{\tilde h}}\big|\sum_{t=1}^n\partial_{h_m} g(Z_{j,t},\beta,\bar h)\{\tilde h_m(X_{j(-k),t})-h^0_{jk,m}(X_{j(-k),t})\}\big|\lesssim r_{N1,m}+ r_{\varsigma} r_{N2,m},
\end{align*}
where $r_{N1,m}= \sqrt{n} s \rho_n \{\log(a_n/\epsilon)\}^{1/2} \Phi^h_{m,2,\varsigma}$, $r_{N2,m} = s^{2}\rho_n\{\log(a_n/\epsilon)\}^{3/2}\Omega^h_{m,q,\varsigma}$, {with a sufficiently small $\epsilon$}. The rates of $\Phi^h_{m,2,\varsigma}$ and $\Omega^h_{m,q,\varsigma}$ are restricted in \hyperref[C9]{(C9)}.

Similarly, by defining
\begin{align}\label{Phi.beta}
\Phi^\beta_{2,\varsigma}\defeq \max_{\bar\beta\in F_n^{\beta}}\big\|\partial_{\beta} g(Z_{j,\cdot},\bar\beta,\tilde h)\big\|_{2,\varsigma},\,\Omega^\beta_{q,\varsigma}\defeq\big\|\max_{\bar\beta\in F_n^{\beta}}|\partial_{\beta} g(Z_{j,\cdot},\bar\beta,\tilde h)|\big\|_{q,\varsigma},
\end{align}
we have with probability tending to 1
\begin{align*}
\max_{\bar\beta\in F_n^{\beta}}\big|\sum_{t=1}^n\partial_{\beta} g(Z_{j,t},\bar\beta,\tilde h)(\beta-\beta^0_{jk})\big|\lesssim r_{N1,0}+ r_{\varsigma} r_{N2,0},
\end{align*}
where $r_{N1,0}= \rho_n\sqrt{n s\log(a_n/\epsilon)}\Phi^\beta_{2,\varsigma}$, $r_{N2,0} = \rho_n\{s\log(a_n/\epsilon)\}^{3/2}\Omega^\beta_{q,\varsigma}$, {with a sufficiently small $\epsilon$}. {And \hyperref[C9]{(C9)} constrains the rates of $\Phi^\beta_{2,\varsigma}$ and $\Omega^\beta_{q,\varsigma}$.}

As a result, with probability $1-\smallO(1)$,
\begin{equation}\label{bound.N}
N_n\lesssim \rho_{N_n},
\end{equation}
by letting $\underset{m\in\{0,1,2\}}{\max}\{r_{N1,m} + r_{\varsigma}r_{N2,m}\}=\bigO(\rho_{N_n})$.

As $\P(K_n+N_n \geq x) \leq \P(K_n \geq x/2)+ \P(N_n \geq x/2)$ and collecting the results from \eqref{decomp}, \eqref{bound.K}, and \eqref{bound.N}, we have shown that $\Delta_n$ satisfies
\beq
n^{-1}\max_{(j,k)\in G}\Delta_n \lesssim \rho_{\Delta_n},\notag
\eeq
with probability $1- \smallO(1)$, where $\rho_{\Delta_n} = n^{-1}(\rho_{K_n} + \rho_{N_n})=\smallO(n^{-1/2}g_n^{-1})$ (given $\tilde\epsilon$ is sufficiently small, and using \hyperref[C8]{(C8)} and \hyperref[C9]{(C9)}).

\end{proof}

\begin{remark}[The rates of $\Omega^h_{m,q,\varsigma}$ and $\Omega^\beta_{q,\varsigma}$]\label{omega.rate} 
It is worth discussing the rates of $\Omega^h_{m,q,\varsigma}$ and $\Omega^\beta_{q,\varsigma}$ by the definition under some special cases. For example, consider the VAR(1) model as in Comment \ref{var} 
given by $Y_t=AY_{t-1}+\vps_t$, where $Y_t,\vps_t \in {\R}^J$, and $\vps_t\sim\mbox{i.i.d.}\N(0, \Sigma)$. 
At first, as shown in the proof of Theorem \ref{gausappro}, we have
$$\Omega^h_{m,q,\varsigma}=\big\|\max_{\bar h\in F_n^{\tilde h}}|\breve X_{\cdot}^{jk}\partial_{h_m} g(Z_{j,\cdot},\beta,\bar h)|_\infty\big\|_{q,\varsigma}\lesssim\big\|\max_{(j,k)\in G}|\breve X_{\cdot}^{jk}|_\infty\big\|_{2q,\varsigma}\big\|\max_{\bar h\in F_n^{\tilde h}}\partial_{h_m} g(Z_{j,\cdot},\beta,\bar h)\big\|_{2q,\varsigma}.$$
For the first term, it is not hard to see that
$$\big\|\max_{(j,k)\in G}\{|\breve X_{t}^{jk}|_\infty-|(\breve X_{t}^{jk})^\ast|_\infty\}\big\|_{2q}\lesssim |A|^{t-1}_{\infty}\||\vps_0|_\infty\|_{2q}\lesssim J^{1/(2q)},$$
where the last inequality is by the union bound, assuming $|A|_{\infty}< 1$,
and the $q$th moments of $\vps_{j,0}$ ($\forall j$) are bounded by a constant $\mu_q$. As for the second term, let $d_n\defeq|G|\vee J$. In the mean regression case, for $f\in\widetilde{\mathcal F}$, $\E(f(z_t)|\mathcal F_{t-1})=\{X_{jk,t}(\beta_{jk}^0-\beta)+h_1^0-\tilde h_1\}(v_{jk,t}+h^0_2-\tilde h_2)$, it can be seen that
\begin{eqnarray*}
&&\big\|\max_{\bar h\in F_n^{\tilde h}}\{\partial_{h_1} g(Z_{j,t},\beta,\bar h)-\partial_{h_1} g(Z_{j,t}^\ast,\beta,\bar h)\}\big\|_{2q}\\
&\leq&\big\|\max_{(j,k)\in G}|v_{jk,t}-v_{jk,t}^\ast|\big\|_{2q} + \big\|\max_{(j,k)\in G}\big\{|X^\top_{j(-k),t}-(X^\top_{j(-k),t})^\ast|\max_{\bar\gamma_{j(-k)}}|\gamma_{j(-k)}^0-\bar\gamma_{j(-k)}|\big\}\big\|_{2q}\\
&\lesssim& d_n^{1/(2q)}(1\vee s^{1/2}\rho_n),
\end{eqnarray*}
while in the median regression case, for $f\in\widetilde{\mF}$, $\E(f(z_t)|\mathcal F_{t-1})=[\frac{1}{2} - F_{\vps_{j,t}|\mathcal F_{t-1}}\{X_{jk,t}(\beta_{jk}^0-\beta)+h_1^0-\tilde h_1\}](v_{jk,t}+h^0_2-\tilde h_2)$,
\begin{eqnarray*}
&&\big\|\max_{\bar h\in F_n^{\tilde h}}\{\partial_{h_1} g(Z_{j,t},\beta,\bar h)-\partial_{h_1} g(Z_{j,t}^\ast,\beta,\bar h)\}\big\|_{2q}\\
&\lesssim&\big\|\max_{(j,k)\in G}|v_{jk,t}-v_{jk,t}^\ast|\big\|_{4q} + \big\|\max_{(j,k)\in G}\big\{|X^\top_{j(-k),t}-(X^\top_{j(-k),t})^\ast|\max_{\bar\gamma_{j(-k)}}|\gamma_{j(-k)}^0-\bar\gamma_{j(-k)}|\big\}\big\|_{4q}\\
&\lesssim& d_n^{1/(4q)}(1\vee s^{1/2}\rho_n),
\end{eqnarray*}
where we use the assumption such that the $4q$th moment of the conditional density is bounded. Moreover, we have
\begin{eqnarray*}
&&\big\|\max_{\bar h\in F_n^{\tilde h}}\{\partial_{h_2} g(Z_{j,t},\beta,\bar h)-\partial_{h_2} g(Z_{j,t}^\ast,\beta,\bar h)\}\big\|_{2q}\\
&\leq&\big\|\max_{(j,k)\in G}|(X_{j(-k),t}-X_{j(-k),t}^\ast)(\beta_{jk}^0-\beta)|\big\|_{2q} \\
&& +\,\big\|\max_{(j,k)\in G}\big\{|X^\top_{j(-k),t}-(X^\top_{j(-k),t})^\ast|\max_{\bar\beta_{j(-k)}}|\beta_{j(-k)}^0-\bar\beta_{j(-k)}|\big\}\big\|_{2q}\\
&\lesssim& d_n^{1/(2q)}(1\vee s^{1/2}\rho_n),
\end{eqnarray*}
or $\big\|\max_{\bar h\in F_n^{\tilde h}}\{\partial_{h_2} g(Z_{j,t},\beta,\bar h)-\partial_{h_2} g(Z_{j,t}^\ast,\beta,\bar h)\}\big\|_{2q}=\bigO(1)$ for the two cases. Therefore, we are able to conclude that $\Omega_{m,q,\varsigma}^h\lesssim d_n^{1/q}(1\vee s^{1/2}\rho_n)$ or $\Omega_{m,q,\varsigma}^h\lesssim d_n^{3/(4q)}(1\vee s^{1/2}\rho_n)$, respectively.

Similarly, it can be shown that $\Omega_{q,\varsigma}^\beta\lesssim d_n^{1/q}s^{1/2}\rho_n$ or $\Omega_{q,\varsigma}^\beta\lesssim d_n^{1/(2q)}s^{1/2}\rho_n$ for the two cases, since
\begin{eqnarray*}
&&\big\|\max_{\bar\beta\in F_n^{\beta}}|\partial_{\beta} g(Z_{j,\cdot},\bar\beta,\tilde h)|\big\|_{q}\\
&\lesssim&\big\|\max_{(j,k)\in G}|X^\top_{j(-k),t}-(X^\top_{j(-k),t})^\ast|\big\|_{2q} \big\|\max_{(j,k)\in G}|\{X^\top_{j(-k),t}-(X^\top_{j(-k),t})^\ast\}\{\gamma_{j(-k)}^0-\bar\gamma_{j(-k)}\}|\big\|_{2q}\\
&\lesssim& d_n^{1/q}s^{1/2}\rho_n,
\end{eqnarray*}
or
\begin{eqnarray*}
&&\big\|\max_{\bar\beta\in F_n^{\beta}}|\partial_{\beta} g(Z_{j,\cdot},\bar\beta,\tilde h)|\big\|_{q}\\
&\lesssim&\big\|\max_{(j,k)\in G}|X^\top_{j(-k),t}-(X^\top_{j(-k),t})^\ast|\big\|_{4q} \big\|\max_{(j,k)\in G}|\{X^\top_{j(-k),t}-(X^\top_{j(-k),t})^\ast\}\{\gamma_{j(-k)}^0-\bar\gamma_{j(-k)}\}|\big\|_{4q}\\
&\lesssim& d_n^{1/(2q)}s^{1/2}\rho_n.
\end{eqnarray*}
In addition, a similar derivation can show that $\|F(z_t)\|_q\lesssim d_n^{1/q}(1\vee\rho_n)$ and $\big\|\underset{(j,k)\in G}{\max}|\psi^0_{jk,\cdot}|\big\|_{q,\varsigma} \lesssim d_n^{1/q}(1\vee\rho_n)$.
\end{remark}

\begin{lemma}\label{max2}
Under the same conditions as in Theorem \ref{bahadur}, we have with probability $1-\smallO(1)$,
\begin{equation}
\max_{(j,k)\in G}|\E{_n} \psi_{jk}\{Z_{j,t},\beta^0_{jk},h^0_{jk}(X_{j(-k),t})\}|\lesssim r_n. 
\end{equation}
\end{lemma}

\begin{proof}[Proof of Lemma \ref{max2}]
	Consider the class of function $\mathcal{F}_G = \{z\mapsto\psi_{jk}\{z,\beta^0_{jk},h^0_{jk}(x_{j(-k)})\}: (j,k)\in G\}$, the cardinality of the set is $|G|$.
	Therefore, the corresponding covering number is given by $\sup_{\mathcal Q}\mathcal{N}(\epsilon \|\bar F_G\|_{\mathcal Q,2},\mathcal{F}_G, \|\cdot\|_{\mathcal Q,2})=|G|/\epsilon$, with $\bar F_G=\sup_{f\in\mathcal F_G}|f|$. Let $\psi_{jk,t}^0\defeq\psi_{jk}\{Z_{j,t},\beta^0_{jk},h^0_{jk}(X_{j(-k),t})\}$ and applying the tail probability bounds in Lemma \ref{tail}, we have with probability $1-\smallO(1)$,
	\begin{equation}
	\max_{(j,k)\in G}|\E{_n} \psi^0_{jk,t}|\lesssim n^{-1} (r_1+r_{\varsigma} r_2) = r_n,
	\end{equation}
	{where $r_1=\{n\log(a_n/\epsilon)\}^{1/2}\underset{(j,k)\in G}{\max}\|\psi^0_{jk,\cdot}\|_{2,\varsigma}$, $r_2=\{\log(a_n/\epsilon)\}^{3/2} \|\underset{(j,k)\in G}{\max}|\psi^0_{jk,\cdot}|\|_{q,\varsigma}$, with a sufficiently small $\epsilon$, $r_{\varsigma} = n^{1/q}$ for $\varsigma > 1/2-1/q$ and $r_{\varsigma} = n^{1/2-\varsigma}$ for $\varsigma < 1/2-1/q$.} 
\end{proof}

\begin{lemma}\label{max3}
Under the same conditions as in Theorem \ref{bahadur}, consider the class of functions $\mathcal{F}'=\{z\mapsto\psi_{jk}\{z,\beta,\tilde h(x_{j(-k)})\}: (j,k)\in G, \beta\in\mathcal{B}_{jk},\tilde h\in\mathcal{H}_{jk}\cup \{h^0_{jk}\}\}$, we have with probability $1-\smallO(1)$,
\begin{equation}
n^{-1/2}\sup_{f\in \mathcal{F}'}|G_n(f)|\lesssim \rho_n.
\end{equation}
\end{lemma}
\begin{proof}[Proof of Lemma \ref{max3}]
The covering number of the function class $\mathcal{F}'$ is given by \\
$\sup_{\mathcal Q}\mathcal{N}(\epsilon \|\bar F'\|_{\mathcal Q,2},\mathcal{F}', \|\cdot\|_{\mathcal Q,2})=(a_n/\epsilon)^{cs}$, with $\bar F'=\sup_{f\in\mathcal F'}|f|$. Also, for any $f\in\mathcal F'$, there exists a set $F'_n$ such that $\min_{f'\in F'_n}\|f-f'\|_{\mathcal Q,2} \leq \epsilon\|\bar F'\|_{\mathcal Q,2}$ and the cardinality of the set $|F'_n|=(a_n/\epsilon)^{cs}$.

One can apply the technique we used in the proof of Lemma \ref{ratedeltan} to achieve the concentration inequality. Similarly, consider the function set $F'_n$, for each $f\in F'_n$, let $\varphi_{l,t}\defeq f(z_t)$ and $\tilde \varphi_{l,t}\defeq \varphi_{l,t} - \E(\varphi_{l,t}|\mF_{t-1},X_{j(-k),t})$, $l=1,\ldots,|F'_n|$. We have $$n|\underset{f\in F'_n}{\max}|\E{_{n}} f - \E{_{n}}\E(f|\mF_{t-1},X_{j(-k),t}) |\lesssim_{\P} 4\sqrt{ns \log(a_n/\epsilon)}\max_{f\in\mathcal F'}\|f(z_t)\|_{2} + r_{\varsigma} \{s\log (a_n/\epsilon)\}^{3/2} \|4\bar F'(z_t)\|_q,$$
{given $\epsilon$ is sufficiently small.}

For each $f\in F'_n$, there exists a function $g$ such that $g(z_t,\beta,\tilde h)=\E\{f(z_t)|\mF_{t-1},X_{j(-k),t}\}-\E\{f(z_t)\}$, where $\beta\in\mathcal{B}_{jk},\tilde h\in\mathcal{H}_{jk}\cup \{h^0_{jk}\}, (j,k)\in G$.
As by the mean value theorem and the continuity of the function $g$, we have
\begin{align*}
g(Z_{j,t},\beta,\tilde h)=&\,\partial_\beta g(Z_{j,t},\bar\beta,\tilde h)(\beta-\beta^0_{jk})\\
&+\sum_{m=1}^2\partial_{h_m} g(Z_{j,t},\beta,\bar h)\{\tilde h_m(X_{j(-k),t})-h^0_{jk,m}(X_{j(-k),t})\},
\end{align*}
where $(\bar\beta,\bar h(\cdot))$ is the corresponding point which joins the line segment between $(\beta,\tilde h(\cdot))$ and $(\beta^0_{jk},h^0_{jk}(\cdot))$. {Let $F^{'\beta}_n$ and $F^{'\tilde h}_n$ collect all the points of $\beta$ and $\tilde h$ according to $F'_n$, and define the following quantities ($m=1,2$) }

{
\begin{align}\label{Phi'}
&\Phi^{'h}_{m,2,\varsigma}\defeq \max_{ \bar{h}\in F_n^{'\tilde h}}\big\||\breve X_{\cdot}^{jk}\partial_{h_m} g(Z_{j,\cdot},\beta,\bar h)|_{\infty}\big\|_{2,\varsigma},\,
\Omega^{'h}_{m,q,\varsigma}\defeq\big\|\max_{\bar{h}\in F_n^{'\tilde h}}|\breve X_{\cdot}^{jk}\partial_{h_m} g(Z_{j,\cdot},\beta,\bar h)||_\infty\big\|_{q,\varsigma},\notag\\
&\Phi^{'\beta}_{2,\varsigma}\defeq \max_{\bar\beta\in F_n^{'\beta}}\big\|\partial_{\beta} g(Z_{j,\cdot},\bar\beta,\tilde h)\big\|_{2,\varsigma},\,\Omega^{'\beta}_{q,\varsigma}\defeq\big\|\max_{\bar\beta\in F_n^{'\beta}}|\partial_{\beta} g(Z_{j,\cdot},\bar\beta,\tilde h)|\big\|_{q,\varsigma}.
\end{align}

Then we have with probability approaching $1$,
\begin{align*}
&\max_{\bar h\in F_n^{'\tilde h}}\big|\sum_{t=1}^n\partial_{h_m} g(Z_{j,t},\beta,\bar h)\{\tilde h_m(X_{j(-k),t})-h^0_{jk,m}(X_{j(-k),t})\}\big|\lesssim r'_{N1,m}+ r_{\varsigma} r'_{N2,m},\, m=1,2,\\
&\max_{\bar\beta\in F_n^{'\beta}}\big|\sum_{t=1}^n\partial_{\beta} g(Z_{j,t},\bar\beta,\tilde h)(\beta-\beta^0_{jk})\big|\lesssim r'_{N1,0}+ r_{\varsigma} r'_{N2,0},
\end{align*}
where $r'_{N1,m}= \sqrt{n} s \rho_n \{\log(a_n/\epsilon)\}^{1/2} \Phi^{'h}_{m,2,\varsigma}$, $r'_{N2,m} = s^{2}\rho_n\{\log(a_n/\epsilon)\}^{3/2}\Omega^{'h}_{m,q,\varsigma}$, and $r'_{N1,0}= \rho_n\{n s\log(a_n/\epsilon)\}^{1/2}\Phi^{'\beta}_{2,\varsigma}$, $r'_{N2,0} = \rho_n\{s\log(a_n/\epsilon)\}^{3/2}\Omega^{'\beta}_{q,\varsigma}$, {with a sufficiently small $\epsilon$}. Also \hyperref[C9]{(C9)} constrains the rates of $\Phi^{'h}_{m,2,\varsigma}$, $\Omega^{'h}_{m,q,\varsigma}$, $\Phi^{'\beta}_{2,\varsigma}$, and $\Omega^{'\beta}_{q,\varsigma}$.
}

The rest of the proof is similar as for Lemma \ref{ratedeltan} and thus is omitted.
\end{proof}

{
\begin{lemma}\label{ratedeltane}

Under the same conditions as in Lemma \ref{ratedeltan} with \hyperref[C9e]{(C9')} instead of \hyperref[C6]{(C6)}, \hyperref[C8]{(C8)} and \hyperref[C9]{(C9)},
\beq
n^{-1}\max_{(j,k)\in G}\Delta_n \lesssim \smallO(n^{-1/2}g_n^{-1}),
\eeq
holds with probability $1-\smallO(1)$.
\end{lemma}

\begin{proof}[Proof of Lemma \ref{ratedeltane}]
We now study the tail probability under stronger tail assumptions. In particular, we need to carry out an analogue proof of Lemma \ref{ratedeltan} under \hyperref[C9e]{(C9')}.

Specifically, by Lemma \ref{lemma.exp}, {for a sufficiently small $\epsilon$, we have $K_n \lesssim_{\P} n^{1/2}\{s\log(a_n/\epsilon)\}^{1/\gamma}\rho_{n,\upsilon}^{e}$ (in particular, for the mean regression case $\rho^e_{n,\upsilon}= \rho_n^{e}s $ and $\rho^e_{n,\upsilon}=\sqrt{\rho_n^{e}}$), and
\begin{align}\label{Phi.exp}
&N_n \lesssim_{\P} n^{1/2}\{s\log (a_n/\epsilon)\}^{1/\gamma}\rho_n^e \{(s^{1/2}\underset{m\in\{1,2\}}{\max}\Phi^h_{m,\psi_{\nu},0}) \vee \Phi^\beta_{\psi_{\nu},0}\},\notag\\
&\Phi^h_{m,\psi_{\nu},0}\defeq\underset{\bar h\in F_n^{\tilde h}}{\max}\big\||\breve X_{\cdot}^{jk}\partial_{h_m} g(Z_{j,\cdot},\beta,\bar h)|_\infty\big\|_{\psi_{\nu},0},\,\Phi^\beta_{\psi_{\nu},0}\defeq\underset{\bar\beta\in F_n^{\beta}}{\max}\big\|\partial_{\beta} g(Z_{j,\cdot},\bar\beta,\tilde h)\big\|_{\psi_{\nu},0}.
\end{align}}
The rest of the proof is similar as for Lemma \ref{ratedeltan} and thus is omitted.

\end{proof}

%

\begin{lemma}\label{max2exp}
Under the same conditions as in Lemma \ref{max2} with \hyperref[C9e]{(C9')} instead of \hyperref[C6]{(C6)}, \hyperref[C8]{(C8)} and \hyperref[C9]{(C9)}, and assume that $\underset{(j,k)\in G}{\max} \|\psi_{jk,\cdot}^0\|_{\psi_{\nu},0} < \infty$, we have with probability $1-\smallO(1)$,
\begin{equation}
\max_{(j,k)\in G}|\E{_n} \psi_{jk}\{Z_{j,t},\beta^0_{jk},h^0_{jk}(X_{j(-k),t})\}|\lesssim n^{-1/2} \{\log (a_n/\epsilon)\}^{1/\gamma}\max_{(j,k)\in G} \|\psi_{jk,\cdot}^0\|_{\psi_{\nu},0}\lesssim r_n. 
\end{equation}
\end{lemma}
\begin{proof}[Proof of Lemma \ref{max2exp}]
The proof is similar to the proof of Lemma \ref{max2} by replacing the tail probability bounds therein by Lemma \ref{lemma.exp}.
\end{proof}

\begin{lemma}\label{max3exp}
Under the same conditions as in Lemma \ref{max3} with \hyperref[C9e]{(C9')} instead of \hyperref[C6]{(C6)}, \hyperref[C8]{(C8)} and \hyperref[C9]{(C9)}, 
and assume that $\underset{f\in\mathcal F'}{\max}\|f(z_\cdot)\|_{\psi_{\nu},0}< \infty$, 
we have with probability $1-\smallO(1)$,
\begin{equation}
n^{-1/2}\sup_{f\in \mathcal{F}'}|G_n(f)|
\lesssim \rho^e_n.
\end{equation}
\end{lemma}

\begin{proof}[Proof of Lemma \ref{max3exp}]
The proof is similar to the proof of Lemma \ref{max3} by replacing the tail probability bounds therein by Lemma \ref{lemma.exp}. In particular, it can be shown that {for a sufficiently small $\epsilon$, }
\begin{align}\label{Phi.exp'}
&n^{-1/2}\sup_{f\in \mathcal{F}'}|G_n(f)|\lesssim_{\P} n^{-1/2}(s \log (a_n/\epsilon))^{1/\gamma}  \big[\max_{f\in\mathcal F'}\|f(z_\cdot)\|_{\psi_{\nu},0} \vee \rho_n^e \{(s^{1/2}\underset{m\in\{1,2\}}{\max}\Phi^{'h}_{m,\psi_{\nu},0}) \vee \Phi^{'\beta}_{\psi_{\nu},0}\}\big],\notag\\
&\Phi^{'h}_{m,\psi_{\nu},0}\defeq\underset{\bar h\in F_n^{'\tilde h}}{\max}\big\||\breve X_{\cdot}^{jk}\partial_{h_m} g(Z_{j,\cdot},\beta,\bar h)|_\infty\big\|_{\psi_{\nu},0},\,\Phi^{'\beta}_{\psi_{\nu},0}\defeq\underset{\bar\beta\in F_n^{'\beta}}{\max}\big\|\partial_{\beta} g(Z_{j,\cdot},\bar\beta,\tilde h)\big\|_{\psi_{\nu},0}.
\end{align}
The final conclusion can be achieved by \hyperref[C9e]{(C9')}.
\end{proof}


}

\subsubsection{Proofs of Section \ref{SI}}

\begin{proof}[Proof of Theorem \ref{bahadur}]
	The sketch of the proof follows the proof of Theorem 2 in \cite{BCK15Bio}. 
	
	\underline{Step 1:} Let $\tilde{\beta}_{jk}$ be any estimator such that $\max_{(j,k)\in G}|\tilde{\beta}_{jk}- \beta^0_{jk}|\leq C \rho_n$ with probability $1-\smallO(1)$.
	By rewriting (using the fact that $\E[\psi_{jk}\{Z_{j,t},\beta^0_{jk},h^0_{jk}(X_{j(-k),t})\}]=0$), we have
	\begin{eqnarray}\label{exp}
	&&\E{_n}[\psi_{jk}\{Z_{j,t},\tilde\beta_{jk},\hat h_{jk}(X_{j(-k),t})\}]\notag\\
	&=& \E{_n}[\psi_{jk}\{Z_{j,t},\beta^0_{jk},h^0_{jk}(X_{j(-k),t})\}]+\E[\psi_{jk}\{Z_{j,t},\beta,\tilde h(X_{j(-k),t})\}]\big|_{\beta=\tilde\beta_{jk},\tilde h=\hat{h}_{jk}} + n^{-1}\Delta_n\notag\\
	\end{eqnarray}
	where $\Delta_n\defeq n^{1/2}G_n[\psi_{jk}\{Z_{j,t},\tilde\beta_{jk},\hat h_{jk}(X_{j(-k),t})\}-\psi_{jk}\{Z_{j,t},\beta^0_{jk},h^0_{jk}(X_{j(-k),t})\}]$.
	
	We first observe that with probability $1-\smallO(1)$, $\max_{(j,k)\in G}\Delta_n\leq\sqrt{n}\sup_{f\in\widetilde{\mathcal{F}}}|G_n(f)|$, where $\widetilde{\mathcal{F}}$ is the class of functions defined by $\widetilde{\mathcal{F}} = \{z\mapsto\psi_{jk}\{z,\beta,\tilde h(x_{j(-k)})\} - \psi_{jk}\{z,\beta^0_{jk},h^0_{jk}(x_{j(-k)})\}: (j,k)\in G, \beta\in\mathcal{B}_{jk},|\beta-\beta^0_{jk}|\leq C\rho_n,\tilde h\in\mathcal{H}_{jk}\}$. The key to our proof is to achieve a concentration inequality for
	$\Delta_n$, such that $n^{-1}\max_{(j,k)\in G}\Delta_n \lesssim \smallO(n^{-1/2}g_n^{-1})$ holds with probability $1-\smallO(1)$. This is done in Lemma \ref{ratedeltan}.

	Then we expand the second term in \eqref{exp} by Taylor expansion. Pick any $\beta\in\mathcal{B}_{jk}$ such that $|\beta-\beta^0_{jk}|\leq C\rho_n$ and $\tilde h\in\mathcal{H}_{jk}$. For any $(j,k)\in G$, let $(\bar\beta,\bar h(X_{j(-k),t})^\top)^\top$ lie on the line segment between $(\beta,\tilde h(X_{j(-k),t})^\top)^\top$ and $(\beta^0_{jk},h^0_{jk}(X_{j(-k),t})^\top)^\top$. Therefore, we can write $\E[\psi_{jk}\{Z_{j,t},\beta,\tilde h(X_{j(-k),t})\}]$ as follows
	\begin{align}\label{eq:exp}
	\E[&\psi_{jk}\{Z_{j,t},\beta^0_{jk},h^0_{jk}(X_{j(-k),t})\}]+\E\big(\partial_{\beta}\E[\psi_{jk}\{Z_{j,t},\beta^0_{jk},h^0_{jk}(X_{j(-k),t})\}|X_{j(-k),t}]\big)(\beta-\beta^0_{jk}) \notag\\
	&+\sum^M_{m=1}\E\big(\partial_{h_{m}}\E[\psi_{jk}\{Z_{j,t},\beta^0_{jk},h^0_{jk}(X_{j(-k),t})\}|X_{j(-k),t}]\{\tilde{h}_{m}(X_{j(-k),t}) -h^0_{jk,m}(X_{j(-k),t})\}\big)\notag\\
	&+\frac{1}{2}\E\big(\partial^2_{\beta}\E[\psi_{jk}\{Z_{j,t},\bar\beta,\bar h(X_{j(-k),t})\}|X_{j(-k),t}]\big)(\beta-\beta^0_{jk})^2\notag\\
	&+ \frac{1}{2}\sum_{m,m'=1}^M\E\big(\partial_{h_{m}} \partial_{h_{m'}}\E[\psi_{jk}\{Z_{j,t},\bar\beta,\bar h(X_{j(-k),t})\}|X_{j(-k),t}]\{\tilde{h}_{m}(X_{j(-k),t}) - h^0_{jk,m}(X_{j(-k),t})\}\notag\\
	&\hspace{2.4cm}\times\{\tilde{h}_{m'}(X_{j(-k),t}) - h^0_{jk,m'}(X_{j(-k),t})\}\big)\notag\\
	&+ \frac{1}{2}\sum_{m=1}^M \E\big(\partial_{h_m} \partial_{\beta}\E[\psi_{jk}\{Z_{j,t},\bar\beta,\bar h(X_{j(-k),t})\}|X_{j(-k),t}] \{\tilde{h}_{m}(X_{j(-k),t}) - h^0_{jk,m}(X_{j(-k),t})\}\big)(\beta-\beta^0_{jk}).
    \end{align}
	
	It can be seen from the orthogonality condition \eqref{ortho} that the third term in \eqref{eq:exp} is zero. By \hyperref[C3]{(C3)} we have $\E(\partial_{\beta}\E[\psi_{jk}\{Z_{j,t},\beta^0_{jk},h^0_{jk}(X_{j(-k),t})\}|X_{j(-k),t}])=\partial_{\beta}\E[\psi_{jk}\{Z_{j,t},\beta^0_{jk},h^0_{jk}(X_{j(-k),t})\}]$ $=\phi_{jk}$. Moreover, each of the last three terms in \eqref{eq:exp} is of the order $\bigO(L_{1n}\rho_n^2)=\smallO(n^{-1/2}g_n^{-1})$, given \hyperref[C3]{(C3)} {(suppose the moments are bounded by constant)}, \hyperref[C5]{(C5)}, and \hyperref[C8]{(C8)}. Therefore, we have shown that the second term in \eqref{exp} equals $\phi_{jk}(\tilde\beta_{jk}-\beta^0_{jk})+\smallO(n^{-1/2}g_n^{-1})$, uniformly over $(j,k)\in G$. Then, combining the results in Lemma \ref{ratedeltan} gives
	\begin{eqnarray}\label{cont}
	&&\E{_n}[\psi_{jk}\{Z_{j,t},\tilde\beta_{jk},\hat h_{jk}(X_{j(-k),t})\}] \notag\\
	&=& \E{_n}[\psi_{jk}\{Z_{j,t},\beta^0_{jk},h^0_{jk}(X_{j(-k),t})\}] +\phi_{jk}(\tilde\beta_{jk}-\beta^0_{jk})+\smallO(n^{-1/2}g_n^{-1}).
	\end{eqnarray}
	
	\underline{Step 2:} Next, we need to prove that $\inf_{\beta \in \hat{\mathcal{B}}_{jk}} |\E{_n}[\psi_{jk}\{Z_{j,t},\beta,\hat h_{jk}(X_{j(-k),t})\}]| =\smallO(n^{-1/2}g_n^{-1})$ holds with probability $1-\smallO(1)$.
	For any $(j,k)\in G$, we focus on any point $\beta^\ast_{jk} = \beta^0_{jk}- \phi_{jk}^{-1}\E{_n} [\psi_{jk}\{Z_{j,t},\beta^0_{jk},h^0_{jk}(X_{j(-k),t})\}]$, thus
	\begin{equation}
	\max_{(j,k)\in G}|\beta^\ast_{jk}- \beta^0_{jk}| \leq C\max_{(j,k)\in G}|\E{_n}[\psi_{jk}\{Z_{j,t},\beta^0_{jk},h^0_{jk}(X_{j(-k),t})\}]|.\notag
	\end{equation}
	

    By Lemma \ref{max2}, we have $|\beta^\ast_{jk} - \beta^0_{jk}|\lesssim_{\P} r_n$ uniformly over $(j,k)\in G$. By \hyperref[C2]{(C2)}, $[\beta^0_{jk} \pm c_1r_n]\subset\hat{\mathcal{B}}_{jk}$ with probability $1-\smallO(1)$, thus $\beta^\ast_{jk}$ is contained in $\hat{\mathcal{B}}_{jk}$ with probability $1- \smallO(1)$.
	Using the continuity argument as in \eqref{cont} with $\tilde\beta_{jk}=\beta^\ast_{jk}$ and combining the fact that $\phi_{jk}(\beta^\ast_{jk}-\beta^0_{jk}) = -\E{_{n}}[\psi_{jk}\{Z_{j,t},\beta^0_{jk},h^0_{jk}(X_{j(-k),t})\}]$,
	we have,
	\begin{eqnarray*}
    &&\E{_n}[\psi_{jk}\{Z_{j,t},\beta^\ast_{jk},\hat h_{jk}(X_{j(-k),t})\}]\\
    & =& \E{_{n}}[\psi_{jk}\{Z_{j,t},\beta^0_{jk},h^0_{jk}(X_{j(-k),t})\}]+ \phi_{jk}(\beta^\ast_{jk}-\beta^0_{jk})+ \smallO(n^{-1/2}g_n^{-1})=\smallO(n^{-1/2}g_n^{-1}).
	\end{eqnarray*}
	Therefore,
	\begin{align}\label{step2}
	\inf_{\beta \in \hat{\mathcal{B}}_{jk}}|\E{_n}[\psi_{jk}\{Z_{j,t},\beta,\hat h_{jk}(X_{j(-k),t})\}]| \leq |\E{_n}[\psi_{jk}\{Z_{j,t},\beta^\ast_{jk},\hat h_{jk}(X_{j(-k),t})\}]| = \smallO(n^{-1/2}g^{-1}_n),
	\end{align}
	holds with probability $1- \smallO(1)$ uniformly over $(j,k)\in G$.
	
	\underline{Step 3:} Lastly, it is left to prove that with probability $1- \smallO(1)$, $\max_{(j,k)\in G} |\hat{\beta}_{jk} - \beta^0_{jk}|\leq C \rho_n$, which will lead to the desired Bahadur representation. From \eqref{step2} and by the definition of $\hat\beta_{jk}$, with probability $1- \smallO(1)$ we have $\max_{(j,k)\in G}\big|\E_n[\psi_{jk}\{Z_{j,t},\hat\beta_{jk},\hat h_{jk}(X_{j(-k),t})\}]\big|=\smallO(n^{-1/2}g_n^{-1})$. Consider the class of functions $\mathcal{F}'=\{z\mapsto\psi_{jk}\{z,\beta,\tilde h(x_{j(-k)})\}: (j,k)\in G, \beta\in\mathcal{B}_{jk},\tilde h\in\mathcal{H}_{jk}\cup \{h^0_{jk}\}\}$. Thus, with probability $1- \smallO(1)$,	
	\begin{equation}
	\big|\E{_n}[\psi_{jk}\{Z_{j,t},\hat\beta_{jk},\hat h_{jk}(X_{j(-k),t})\}]\big|\geq \big|\E[\psi_{jk}\{Z_{j,t},\beta,\tilde h(X_{j(-k),t})\}]|_{\beta=\hat\beta_{jk},\tilde h=\hat h_{jk}}\big|- n^{-1/2}\sup_{f\in \mathcal{F}'}|G_n(f)|,\notag
	\end{equation}
	holds uniformly over $(j,k)\in G$.
	Recall that Lemma \ref{max3} ensures $n^{-1/2}\underset{f\in \mathcal{F}'}{\sup}|G_n(f)|=\bigO_{\P}(\rho_n)$. It follows that $ \big|\E[\psi_{jk}\{Z_{j,t},\beta,\tilde h(X_{j(-k),t})\}]|_{\beta=\hat\beta_{jk},\tilde h=\hat h_{jk}}\big|\leq\bigO(\rho_n)+\smallO(n^{-1/2}g_n^{-1})$.
	
	In addition, applying the expansion in \eqref{eq:exp} with $\beta^0_{jk}=\beta$ together with the Cauchy-Schwarz inequality implies that $\big|\E[\psi_{jk}\{Z_{j,t},\beta,\tilde h(X_{j(-k),t})\}]-\E[\psi_{jk}\{Z_{j,t},\beta,h^0_{jk}(X_{j(-k),t})\}]\big|\leq C(L_n^{1/2}\rho_n+L_{1n}\rho_n^2)=\bigO(\rho_n)$ {(suppose the moments are bounded by constant)}, so with probability $1- \smallO(1)$,
	\begin{eqnarray}\label{step3}
	\hspace{-0.5cm}\big|\E[\psi_{jk}\{Z_{j,t},\beta,\tilde h(X_{j(-k),t})\}]|_{\beta=\hat\beta_{jk},\tilde h=\hat h_{jk}}\big|\geq \big|\E[\psi_{jk}\{Z_{j,t},\beta,h^0_{jk}(X_{j(-k),t})\}]|_{\beta=\hat\beta_{jk}}\big|- \bigO(\rho_n),
	\end{eqnarray}
	uniformly over $(j,k)\in G$, where \hyperref[C3]{(C3)} and the fact that $\E[\{\tilde h_{m}(X_{j(-k),t})-h^0_{jk,m}(X_{j(-k),t})\}^2]\leq C\rho_n^2$ for all $m=1,\ldots,M$ and any $\tilde h=(\tilde h_m)_{m=1}^M\in\mathcal{H}_{jk}$ are used.
	
	Moreover, given the identification condition \hyperref[C4]{(C4)}, the first term on the right-hand side of \eqref{step3} is bounded from below by $\frac{1}{2}\{|\phi_{jk}(\hat{\beta}_{jk} - \beta^0_{jk})|\wedge c_1\}$ and this results in that with probability $1- \smallO(1)$, $|\hat{\beta}_{jk} - \beta^0_{jk}|\leq\smallO(n^{-1/2}g_n^{-1}) + \bigO(\rho_n) = \bigO(\rho_n)$ uniformly over $(j,k)\in G$.
	
	In summary, we have shown that, with probability $1- \smallO(1)$,
	\begin{eqnarray}\label{step4}
	&&\E{_n}[\psi_{jk}\{Z_{j,t},\hat\beta_{jk},\hat h_{jk}(X_{j(-k),t})\}]\notag\\
	 &= &\E{_n}[\psi_{jk}\{Z_{j,t},\beta^0_{jk},h^0_{jk}(X_{j(-k),t})\}] +\phi_{jk}(\hat\beta_{jk}-\beta^0_{jk})+\smallO(n^{-2}g_n^{-1}),
	\end{eqnarray}
	uniformly over $(j,k)\in G$. And with probability $1- \smallO(1)$, the left-hand side is $\smallO(n^{-1/2}g_n^{-1})$ uniformly over $(j,k)\in G$. Lastly, the uniform Bahadur representation can be obtained by solving \eqref{step4} with respect to $(\hat{\beta}_{jk}-\beta^0_{jk})$.
	
\end{proof}

\begin{proof}[Proof of Corollary \ref{asy.norm}]
The proof is an application of Theorem \ref{bahadur} with verification of conditions \hyperref[C1]{(C1)}-\hyperref[C9]{(C9)}.

Here we focus on the estimator by Algorithm \hyperref[algo2]{2} as the proof of Algorithm \hyperref[algo1]{1} is basically the same. In particular, with the LAD regression case, we have $|G|=1$, $a_n=\max(JK,n)$, $g_n=1$, $M=2$, $h^0_{jk}(X_{j(-k),t})=(X_{j(-k),t}^\top\beta^0_{j(-k)},X_{j(-k),t}^\top\gamma^0_{j(-k)})^\top$,  $\psi_{jk}\{Z_{j,t},\beta_{jk},h^0_{jk}(X_{j(-k),t})\}=\{1/2-\IF(Y_{j,t}\leq X_{jk,t}\beta_{jk}+X_{j(-k),t}^\top\beta^0_{j(-k)})\}(X_{jk,t}-X_{j(-k),t}^{\top}\gamma^0_{j(-k)})$.

Verification of \hyperref[C1]{(C1)}: Our model setting assumes $F_{\vps_j}(0)=1/2$ and $\E( v_{jk,t}X_{j(-k),t}) = 0$; hence we have
\begin{align*}
&\E(\partial_{h_1}\E[\psi_{jk}\{Z_{j,t},\beta^0_{jk},h^0_{jk}(X_{j(-k),t})\}|X_{j(-k),t}]h_1(X_{j(-k),t})) =-\beta_{j(-k)}^\top\E\{f_{\vps_j}(0)v_{jk,t}X_{j(-k),t}\}=0\\
&\E(\partial_{h_2}\E[\psi_{jk}\{Z_{j,t},\beta^0_{jk},h^0_{jk}(X_{j(-k),t})\}|X_{j(-k),t}]h_2(X_{j(-k),t})) =-\gamma_{j(-k)}^\top\E[\{1/2-F_{\vps_j}(0)\}X_{j(-k),t}]=0
\end{align*}

Verification of \hyperref[C2]{(C2)}: The true parameter $\beta_{jk}^0$ satisfies \eqref{mc} given $F_{\vps_j}(0)=1/2$. {Moreover, 
by Remark 2 in \citet{BCK15_sup}, 
we have $|\hat \beta_{jk}^{[2]}-\beta^0_{jk}|\lesssim\rho_n$ with probability $1-\smallO(1)$. Provided finite dependence adjusted norm in polynomial rates, it is not hard to verify $r_n\lesssim \rho_n$ with proper restrictions on the rate of $\log a_n$, so that for some sufficiently small $c_1>0$ the condition holds.}

Verification of \hyperref[C3]{(C3)}: The map
\begin{align*}
(\beta,h)\mapsto&\E\{\psi_{jk}(Z_{j,t},\beta,h)|X_{j(-k),t}\}\\
&=\E([1/2-F_{\vps_j}\{X_{jk,t}(\beta-\beta^0_{jk})-X_{j(-k),t}^\top\beta_{j(-k)}^0+h_1\}](X_{jk,t}-h_2)|X_{j(-k),t})
\end{align*}
is twice continuously differentiable as $f'_{\vps_j}$ is continuous. For every $\vartheta\in\{\beta,h_1,h_2\}$, \\
$\partial_{\vartheta}\E\{\psi_{jk}(Z_{j,t},\beta,h)|X_{j(-k),t}\}$ is $-\E[f_{\vps_j}\{X_{jk,t}(\beta-\beta^0_{jk})-X_{j(-k),t}^\top\beta_{j(-k)}^0+h_1\}X_{jk,t}(X_{jk,t}-h_2)|X_{j(-k),t}]$ (w.r.t. $\beta$) or $-\E[f_{\vps_j}\{X_{jk,t}(\beta-\beta^0_{jk})-X_{j(-k),t}^\top\beta_{j(-k)}^0+h_1\}(X_{jk,t}-h_2)|X_{j(-k),t}]$ (w.r.t. $h_1$) or $-\E[1/2-F_{\vps_j}\{X_{jk,t}(\beta-\beta^0_{jk})-X_{j(-k),t}^\top\beta_{j(-k)}^0+h_1\}|X_{j(-k),t}]$ (w.r.t. $h_2$). Hence, for every $\beta\in\mathcal{B}_{jk}$, $$|\partial_{\vartheta}\E\{\psi_{jk}(Z_{j,t},\beta,h^0_{jk}(X_{j(-k),t})|X_{j(-k),t}\}|\leq C_1\E(|X_{jk,t}v_{jk,t}|\,|X_{j(-k),t})\vee C_1\E(|v_{jk,t}|\,|X_{j(-k),t})\vee1.$$
Observe that the expectation of the square of the right-hand side is {bounded by constant}. Moreover, let $\mathcal{T}_{jk}(X_{j(-k),t})=\{\tau\in{\R}^2:|\tau_2-X_{j(-k),t}^\top\gamma_{j(-k)}^0|\leq c_3\}$, where $c_3>0$ is a constant. Then for every $\vartheta,\vartheta'\in\{\beta,h_1,h_2\}$, $\beta\in\mathcal{B}_{jk}$, $h\in\mathcal{T}_{jk}(X_{j(-k),t})$, we have
\begin{eqnarray*}
&&|\partial_{\vartheta}\partial_{\vartheta'}\E\{\psi_{jk}(Z_{j,t},\beta,h)|X_{j(-k),t}\}|\\
&\leq& C_1[1\vee\E\{|X_{jk,t}^2(X_{jk,t}-h_2)|\,|X_{j(-k),t}\}\vee\E\{|X_{jk,t}(X_{jk,t}-h_2)|\,|X_{j(-k),t}\}\vee\E(|X_{jk,t}|\,|X_{j(-k),t})\\
&&\qquad\vee\E(|X_{jk,t}-h_2|\,|X_{j(-k),t})].
\end{eqnarray*}
{In particular, 
\begin{eqnarray*}
&&\E\{|X_{jk,t}^2(X_{jk,t}-h_2)|\,|X_{j(-k),t}\}\\
&\leq&\E\{|(X_{j(-k),t}^\top\gamma_{j(-k)}^0+v_{jk,t})^2(c_3+|v_{jk,t}|)|\,|X_{j(-k),t}\}\\
&\leq&2\E\{|\{(X_{j(-k),t}^\top\gamma_{j(-k)}^0)^2+v_{jk,t}^2\}(c_3+|v_{jk,t}|)|\,|X_{j(-k),t}\}\\
&\leq& C\{|X_{j(-k),t}^\top\gamma_{j(-k)}^0|^2\E(|v_{jk,t}|\,|X_{j(-k),t}) + \E(|v_{jk,t}|^3|X_{j(-k),t}) + |X_{j(-k),t}^\top\gamma_{j(-k)}^0|\E(v_{jk,t}^2|X_{j(-k),t})\}.
\end{eqnarray*}
And by similar computation we can show that $|\partial_{\vartheta}\partial_{\vartheta'}\E\{\psi_{jk}(Z_{j,t},\beta,h)|X_{j(-k),t}\}|\leq \ell_1(X_{j(-k),t}):=C'\{|X_{j(-k),t}^\top\gamma_{j(-k)}^0|^2\E(|v_{jk,t}|\,|X_{j(-k),t}) + \E(|v_{jk,t}|^3|X_{j(-k),t}) + |X_{j(-k),t}^\top\gamma_{j(-k)}^0|\E(v_{jk,t}^2|X_{j(-k),t})\}$, 
where the constants $C,C'$ depend on $c_3$ and $C_1$. Lastly, for every $\beta,\beta'\in\mathcal{B}_{jk}$, $h,h'\in\mathcal{T}_{jk}(X_{j(-k),t})$ we have
\begin{eqnarray*}
&&\E[\{\psi_{jk}(Z_{j,t},\beta,h)-\psi_{jk}(Z_{j,t},\beta',h')\}^2|X_{j(-k),t}]\\
&\leq& C_1\E\{|X_{jk,t}(X_{jk,t}-h_2)^2|\,|X_{j(-k),t}\}|\beta-\beta'|+C_1\E\{(X_{jk,t}-h_2)^2\,|X_{j(-k),t}\}|t_1-t'_1|+(t_2-t'_2)^2\\
&\leq& C''\{|X_{j(-k),t}^\top\gamma_{j(-k)}^0|\E(|v_{jk,t}|\,|X_{j(-k),t}) + \E(v_{jk,t}^2|X_{j(-k),t})\}(|\beta-\beta'|+|t_1-t'_1|)+(t_2-t'_2)^2\\
&\leq&\sqrt{2}[C''\{|X_{j(-k),t}^\top\gamma_{j(-k)}^0|\E(|v_{jk,t}|\,|X_{j(-k),t}) +\E(v_{jk,t}^2|X_{j(-k),t})\}+2c_3](|\beta-\beta'|+|t-t'|_2),
\end{eqnarray*}
where constant $C''$ depends on $c_3$ and $C_1$. Consequently, we have verified the last condition in \hyperref[C3]{(C3)} by taking $\ell_2(X_{j(-k),t}):=\sqrt{2}[C''\{|X_{j(-k),t}^\top\gamma_{j(-k)}^0|\E(|v_{jk,t}|\,|X_{j(-k),t}) +\E(v_{jk,t}^2|X_{j(-k),t})\}+2c_3]$ 
and $\upsilon=1$. And given the {bounded} moments conditions on $X_t$, we have $\E\{|\ell_1(X_{j(-k),t})|^4\}\leq L_{1n}$, $\E\{|\ell_2(X_{j(-k),t})|^4\}\leq L_{2n}$.}

Verification of \hyperref[C4]{(C4)}: For any $\beta\in\mathcal{B}_{jk}$, there exists $\beta'$ between $\beta^0_{jk}$ and $\beta$ such that
\begin{eqnarray*}
&&\E[\psi_{jk}\{Z_{j,t},\beta,h^0_{jk}(X_{j(-k),t})\}]\\
&=&\partial_\beta\E[\psi_{jk}\{Z_{j,t},\beta^0_{jk},h^0_{jk}(X_{j(-k),t})\}](\beta-\beta^0_{jk}) +\frac{1}{2}\partial_{\beta}^2\E[\psi_{jk}\{Z_{j,t},\beta',h^0_{jk}(X_{j(-k),t})\}](\beta-\beta^0_{jk})^2.
\end{eqnarray*}
Let $\phi_{jk}=\partial_{\beta}\E[\psi_{jk}\{Z_{j,t},\beta^0_{jk},h^0_{jk}(X_{j(-k),t})\}]\geq c_1^2$. Since $\partial_{\beta}^2\E[\psi_{jk}\{Z_{j,t},\beta',h^0_{jk}(X_{j(-k),t})\}]\leq C_1\E|X_{jk,t}^2v_{jk,t}|\leq C_2$, we have
\begin{align*}
2\big|\E[\psi_{jk}\{Z_{j,t},\beta,h^0_{jk}(X_{j(-k),t})\}]\big|\geq2\phi_{jk}|\beta-\beta^0_{jk}|-C_2(\beta-\beta^0_{jk})^2\geq\phi_{jk}|\beta-\beta^0_{jk}|,
\end{align*}
whenever $|\beta-\beta^0_{jk}|\leq c_1^2/C_2$.

{Verification of \hyperref[C5]{(C5)}: According to Corollary \ref{betabound4}, with probability $1-\smallO(1)$ we have
\begin{align*}
\| \hat \beta_{j(-k)}^{[1]} - \beta_{j(-k)}^0\|_{j, pr} \lesssim \sqrt{s(\log a_n)/n},\quad\| \hat \gamma_{j(-k)} - \gamma_{j(-k)}^0\|_{j, pr} \lesssim \sqrt{s(\log a_n)/n},
\end{align*}}
which means the algorithms can provide an estimator of the nuisance function with good sparsity and rate properties given IC $\lambda$. Thus, by Lemma 7 in \citet{BCK15_sup}, we have \hyperref[C5]{(C5)} holds.

{Verification of \hyperref[C6]{(C6)}: We refer to the proof of Theorem 1 in \citet{BCK15_sup}.
	
Verification of \hyperref[C7]{(C7)}: Recall that $\psi_{jk,t}^0=\{1/2-\IF(\varepsilon_{j,t}\leq0)\}v_{jk,t}$. Hence, $\E(\frac{1}{\sqrt{n}}\sum_{t=1}^n\psi_{jk,t}^0)^2=\sum_{\ell=-(n-1)}^{n-1}(1-|\ell|/n)\E(\psi_{jk,t}^0\psi_{jk,t-\ell}^0)\geq\frac{1}{4}\sum_{\ell=-(n-1)}^{n-1}(1-|\ell|/n)\E(v_{jk,t}v_{jk,t-\ell})\geq c_1/4$.

Verification of \hyperref[C8]{(C8)} and \hyperref[C9]{(C9)}: 
See Comment \ref{comment.rate} where we discuss the admissible dimension rates either under the special case of VAR(1) with geometric decay rate (which gives bounded dependence adjusted norm) or more generally with finite dependence adjusted norm in polynomial rates.

Verification of \hyperref[C9e]{(C9')}: See Comment \ref{comment.rate2} and the discussion can be generalized to the case of finite dependence adjusted norm in polynomial rates easily.
}

\end{proof}

\begin{lemma}\label{normality}
	Let $\psi_{jk,t}^0\defeq\psi_{jk}\{Z_{j,t},\beta^0_{jk},h^0_{jk}(X_{j(-k),t})\}$, $T^{jk}_n\defeq\sigma^{-1}_{jk}\phi_{jk}^{-1}\sum_{t=1}^n \psi^0_{jk,t}$, and assume that $\|\psi^0_{jk,\cdot}\|_{2,\varsigma}<\infty$. Then
	$${\|T_n^{jk}\|_2 = \bigO(\sqrt{n}\|\psi^0_{jk,\cdot}\|_{2,\varsigma})},\,\,\text{and } n^{-1/2}T_n^{jk} \stackrel{\mathcal{L}}{\rightarrow} \N(0,1)$$
\end{lemma}

\begin{proof}[Proof of Lemma \ref{normality}]
	%
	Define the projector operator $\mP_l(X_t) \defeq \E(X_t|\mathcal{F}_l) -\E(X_t|\mathcal{F}_{l-1})$. Note that the projection operator is directly linked to the dependence adjusted norm for $X_{jk,t} = g_{jk}(\mF_t) = g_{jk}(\ldots, \xi_{t-1}, \xi_{t})$,
	and $\|\mP_0(X_{jk,t})\|_2\leq \|g_{jk}(\mF_{t}) - g_{jk}(\mF^\ast_{t})\|_2\leq 2\|\mP_0 (X_{jk,t})\|_2$ (by Theorem 1(i) in \citealp{wu2005nonlinear}).
	
	Let $J^{jk}_{l,n} \defeq \sigma^{-1}_{jk}\phi_{jk}^{-1}\sum_{t=1}^n\mP_{t-l}(\psi^0_{jk,t})$, and it is not hard to see that $T^{jk}_n = \sum^{\infty}_{l=0}J_{l,n}^{jk}$. As $\sigma^{-1}_{jk}\phi_{jk}^{-1}\mP_{t-l}(\psi^0_{jk,t})$'s form the martingale differences over $t$, according to Lemma \ref{buck} we can 
	get $\|J^{jk}_{l,n}\|_2^2 \leq (\sigma_{jk}\phi_{jk})^{-2}\sum_{t=1}^n\|\mP_{t-l}(\psi^0_{jk,t})\|_2^2 \lesssim n(\delta^\psi_{j,k,l})^2$, where $\delta^\psi_{j,k,l}\defeq\|\psi^0_{jk,l}-(\psi^{0}_{jk,l})^\ast\|_2$. Thus, $\|T_n^{jk}\|_2 \lesssim \sqrt{n} \sum^{\infty}_{l=0}\delta^\psi_{j,k,l} \leq \sqrt{n} \|\psi^0_{jk,\cdot}\|_{2, \varsigma} = \bigO(\sqrt{n}\|\psi^0_{jk,\cdot}\|_{2, \varsigma})$. Then the conclusion that $n^{-1/2}T_n^{jk} \stackrel{\mathcal{L}}{\rightarrow} \N(0, 1)$ follows from Lemma \ref{lem:cltwbw} in light of the fact that $\E \psi^0_{jk,t} = 0 $ and $ \|\psi^0_{jk,\cdot}\|_{2,\varsigma} < \infty$.
\end{proof}

\begin{proof}[Proof of Corollary \ref{uninorm}]
The proof follows directly from Lemma \ref{normality}.
\end{proof}

\begin{proof}[Proof of Corollary \ref{cons.boot}]
We apply the high-dimensional central limit theorem (Theorem 3.2 in \cite{ZW15gaussian}) to the vector $\widetilde\Im\defeq\frac{1}{\sqrt{n}}\sum_{t=1}^n\widetilde\zeta_t$ and $\widetilde{\mathcal{Z}} \defeq \operatorname{vec}[\{(\mathcal{Z}_{jk})_{k=1}^K\}_{j=1}^J]$ is the corresponding standard Gaussian random vector, with the same correlation structure. 
Then we have $\rho(D^{-1}\widetilde{\mathcal{\Im}}, D^{-1}\widetilde{\mathcal{Z}})\rightarrow 0$, as $n \to \infty$, where $D$ is a diagonal matrix with the square root of the diagonal elements of the long-run variance-covariance matrix of $\widetilde\zeta_t$, namely $\{\sum^{\ell = \infty}_{\ell=-\infty} \E (\zeta_{jk,t}\zeta_{jk,(t-\ell)})\}^{1/2}$, for $k=1,\ldots,K$, $j=1,\ldots,J$. The rest of the proof is similar to Corollary \ref{gausappro.c} and thus is omitted.
\end{proof}

\begin{proof}[Proof of Corollary \ref{proofboot}]
The proof is similar to that of Theorem \ref{validboot} and Theorem \ref{validboot.joint}; therefore, we omit the detailed proof here.
In particular, the following conditions on $b_n$ are required:
{
\begin{align}\label{ratebn.G}
&b_n=\smallO\{n(\log |G|)^{-4}(\Phi^{\zeta}_{q,\varsigma})^{-4}\wedge n(\log |G|)^{-5}(\Phi_{4,\varsigma}^\zeta)^{-4}\},\,F_\varsigma=\smallO\{n^{q/2}(\log |G|)^{-q}|G|^{-1}(\Gamma_{q,\varsigma}^\zeta)^{-q}\}.\notag\\
&\Phi^\zeta_{2,0}\Phi^\zeta_{2,\varsigma}\{b_n^{-1}+\log(n/b_n)/n+(n-b_n)\log b_n/(nb_n)\}(\log |G|)^2=\smallO(1),\,\text{if } \varsigma=1;\notag\\
&\Phi^\zeta_{2,0}\Phi^\zeta_{2,\varsigma}\{b_n^{-1}+n^{-\varsigma}+(n-b_n)b_n^{-\varsigma+1}/(nb_n)\}(\log |G|)^2=\smallO(1),\,\text{if } \varsigma<1;\notag\\
&\Phi^\zeta_{2,0}\Phi^\zeta_{2,\varsigma}\{b_n^{-1}+n^{-1}b_n^{-\varsigma+1}+(n-b_n)/(nb_n)\}(\log |G|)^2=\smallO(1),\,\text{if } \varsigma>1.
\end{align}
}
where 
$F_{\varsigma} = n$, for $\varsigma >1-2/q$; $F_{\varsigma} = l_nb_n^{q/2-\varsigma q/2}$, for $1/2-2/q<\varsigma<1-2/q$; $F_{\varsigma} = l_n^{q/4-\varsigma q/2}b_n^{q/2-\varsigma q/2} $, for $\varsigma<1/2-2/q$.
\end{proof}

\begin{remark}[Consistency of the pre-estimators $\hat\zeta_{jk,t}$]\label{zeta.remark}
The pre-estimators of the influence functions $\hat\zeta_{jk,t}$ are used in constructing the bootstrap statistics. It is important to discuss the consistency requirement on them for the inference implementation. In particular, the deviation $\underset{(j,k),(j',k')}{\max}|\sum_{i=1}^{l_n} \hat\eta_{j'k',i}\hat\eta_{jk,i} - \sum_{i=1}^{l_n} \eta_{j'k',i}\eta_{jk,i}|$ should be controlled under a certain rate, where $\eta_{jk,i} \defeq \frac{1}{\sqrt{n}} \sum_{l=(i-1)b_n+1}^{ib_n} \zeta_{jk,l} $, and $\hat{\eta}_{jk,i} \defeq\frac{1}{\sqrt{n}}  \sum_{l=(i-1)b_n+1}^{ib_n} \hat \zeta_{jk,l}$. In this comment, we consider the case with stronger tail assumptions as an example to illustrate.

We first observe that
 \begin{eqnarray*}
&& \max_{(j,k),(j',k')}\Big|\sum^{l_n}_{i=1}( \hat{\eta}_{j'k',i}\hat{\eta}_{jk,i} - \eta_{j'k',i}\eta_{jk,i})\Big|\\
&\leq&  \max_{(j,k),(j',k')}\Big|\sum^{l_n}_{i=1} (\hat{\eta}_{j'k',i}- \eta_{j'k',i})(\hat{\eta}_{jk,i}-\eta_{jk,i})\Big|+ 2\max_{(j,k),(j',k')} \Big|\sum^{l_n}_{i=1} \eta_{j'k',i}(\hat{\eta}_{jk,i}- \eta_{jk,i})\Big|\\
&\leq & \max_{(j,k)}\sum_{i=1}^{l_n} (\hat{\eta}_{jk,i}- \eta_{jk,i})^2 + 2\Big(\max_{(j,k)}\sum^{l_n}_{i=1} \eta^2_{jk,i}\Big)^{1/2}\Big\{\max_{(j,k)} \sum^{l_n}_{i=1} (\hat{\eta}_{jk,i}- \eta_{jk,i})^2\Big\}^{1/2}\\
&=:&U_n + 2V_n^{1/2}U_n^{1/2}.
\end{eqnarray*}
Firstly, note that $V_n$ can be analyzed by a concentration inequality similar to the proof of Theorem \ref{residuals.theorem}. In particular, assuming that $\underset{(j,k)}{\max}\|(\psi^0_{jk,\cdot})^2\|_{\psi_\nu,0}<\infty$, the order of $V_n$ is given by $cn^{-1/2}b_n\{\log(KJ)\}^{1/\gamma}\underset{(j,k)}{\max}\|(\psi^0_{jk,\cdot})^2\|_{\psi_\nu,0} + b_n\underset{(j,k)}{\max}\|(\psi^0_{jk,t})\|_{2}^2$, with $\gamma = 1/(2\nu+1)$ and a sufficiently large $c$. Secondly, recall the definitions $\zeta_{jk,t}=-\phi^{-1}_{jk}\sigma_{jk}^{-1}\psi^0_{jk,t}=-\omega_{jk}^{-1/2}\psi^0_{jk,t}$, $\hat\zeta_{jk,t}=-\hat\phi^{-1}_{jk}\hat\sigma_{jk}^{-1}\hat\psi_{jk,t}=-\hat\omega_{jk}^{-1/2}\hat\psi_{jk,t}$, where $\omega_{jk}$'s are essentially the long-run variance of $\psi^0_{jk,t}$. We have
\begin{align*}
U_n&\leq 2\big(\max\limits_{(j,k)}|\hat\omega_{jk}^{-1/2} - \omega_{jk}^{-1/2}|\big)^2n^{-1}\sum^{l_n}_{i=1} \Big(\sum_{l=(i-1)b_n+1}^{ib_n} \psi^0_{jk,l}\Big)^2\notag\\
&\quad+\,2\max_{(j,k)}\hat\omega_{jk}^{-1}n^{-1}\sum^{l_n}_{i=1} \Big\{\sum_{l=(i-1)b_n+1}^{ib_n} (\hat\psi_{jk,l}-\psi^0_{jk,l})\Big\}^2\\
&=:2W_n^2U_{n,1} + 2U_{n,2}.
\end{align*}
Given $\min\limits_{(j,k)}\omega_{j,k} \geq c_\omega$, we have $\max\limits_{(j,k)}|\hat{\omega}_{jk}- \omega_{jk}| \leq x$ (for $x\leq c_\omega/2$) implies $W_n=\max\limits_{(j,k)} |\hat{\omega}_{jk}^{-1/2}- \omega^{-1/2}_{jk}|\leq 2xc_\omega^{-3/2} $. Again, the first term above can be handled by a maximal inequality for $U_{n,1}$ with the order of $cn^{-1/2}b_n\{\log(KJ)\}^{1/\gamma}\underset{(j,k)}{\max}\|(\psi^0_{jk,\cdot})^2\|_{\psi_\nu,0} + b_n\underset{(j,k)}{\max}\|(\psi^0_{jk,t})\|_{2}^2$ (same as $V_n$), together with a consistent rate of $w_n^2$, which can be analyzed by dealing with $\max\limits_{(j,k)}|\hat{\omega}_{jk}- \omega_{jk}|$. \\
As for $U_{n,2}$, consider the event $\mA_0=\{\max\limits_{(j,k)}|\hat{\omega}_{jk}- \omega_{jk}|\leq x\}$. Note that on the event $\mA_0$ for $x\leq c_\omega/2$, we have $\min\limits_{(j,k)}\hat{\omega}_{jk}\geq\min\limits_{(j,k)}{\omega}_{jk} -\max\limits_{(j,k)}|\hat{\omega}_{jk}- \omega_{jk}|\geq c_\omega/2$, which implies $\max\limits_{(j,k)}\hat{\omega}_{jk}^{-1}\leq 2/c_\omega$. Let $\varPsi_n\defeq n^{-1}\max\limits_{(j,k)}\sum^{l_n}_{i=1} \{\sum^{ib_n}_{l= (i-1)b_n+1}(\hat\psi_{jk,l}-\psi^0_{jk,l})\}^2$. It follows that 
\begin{align*}
\P(U_{n,2}\geq z) &\leq \P(\mA_0^C) + \P(\{U_{n,2}\geq z\}\cap\mA_0)\\
&\leq \P(\mA_0^C) + \P(\varPsi_n\geq c_\omega z/2).
\end{align*}
In particular, to get the estimator of $\omega_{jk}$, one can consider $\hat\omega_{jk}=n^{-1}\sum^{l_n}_{i=1}(\sum^{ib_n}_{l= (i-1)b_n+1}\hat{\psi}_{jk,l})^2$, which gives
\begin{eqnarray*}
&&\max_{(j,k)}|\hat{\omega}_{jk}- \omega_{jk}|\\
&\leq& n^{-1}\max_{(j,k)}\Big|\sum^{l_n}_{i=1} \Big\{\Big(\sum^{ib_n}_{l= (i-1)b_n+1}\hat\psi_{jk,l}\Big)^2-\Big(\sum^{ib_n}_{l= (i-1)b_n+1}\psi^0_{jk,l}\Big)^2\Big\}\Big|\\
&&+\,\max_{(j,k)}\Big|\omega_{jk}-n^{-1}\sum^{l_n}_{i=1}\Big(\sum^{ib_n}_{l= (i-1)b_n+1}\psi^0_{jk,l}\Big)^2\Big|\\
&=:&W_{n,1} + W_{n,2},
\end{eqnarray*}
where $W_{n,2}$ has been similarly analyzed in the proof of Theorem \ref{validboot} along with Comment \ref{validboot.exp}. The order of $W_{n,2}$ is given by $cn^{-1/2}b_n\{\log (KJ)\}^{1/\gamma}\Phi_{j,\psi_v, 0}^2\vee \Phi_{4,0}\Phi_{4,\varsigma} v(b_n)$. Moreover, $W_{n,1}$ is bounded by $2\varPsi_n$. We shall tackle the rate of it in the following lemma.

To summarize, suppose the moments and dependence adjusted norms are all bounded by constants, the dominant term in $\underset{(j,k),(j',k')}{\max}|\sum_{i=1}^{l_n} \hat\eta_{j'k',i}\hat\eta_{jk,i} - \sum_{i=1}^{l_n} \eta_{j'k',i}\eta_{jk,i}|$ is given by $b_n^{1/2}\rho_n^{1/2}s^{1/4}$, with the LASSO rate $\rho_n=\bigO(n^{-1/2}s^{1/2}\{\log(KJ)\}^{1/\gamma}))$. Therefore, if we assume $b_n^{1/2}\rho_n^{1/2}s^{1/4}=\smallO(\{\log(JK)\}^{-2})$, it follows that $\underset{(j,k),(j',k')}{\max}|\sum_{i=1}^{l_n} \hat\eta_{j'k',i}\hat\eta_{jk,i} - \sum_{i=1}^{l_n} \eta_{j'k',i}\eta_{jk,i}|=\smallO_{\P}(\{\log(JK)\}^{-2})$.
\end{remark}

\begin{lemma}\label{block.lemma}
Given \hyperref[C6]{(C6)})-\hyperref[C7]{(C7)}) , assume $b_n^{2c+1}n^{(1-c)/2}\{s\log(a_nb_n)\}^{-1/2}\|F(z_{\cdot})\|_{2(1+c),0}^{2(1+c)}=\smallO(1)$, $c(ns)^{1/2}b_n^3\rho_n\{\log(KJ)\}^{1/\gamma}\max\limits_{(j,k)}\max\limits_{1\leq i\leq l_n}\Phi^{\ell_i}_{jk,\psi_\nu} + ns^{1/2}b_n^3\rho_n\max\limits_{(j,k)}\max\limits_{1\leq i\leq l_n}\mu^{\ell_i}_{jk,\psi_\nu}=\bigO(n)$, with $\gamma = 1/(2\nu+1)$. For $w\geq cn^{-1/2}s^{1/2}\rho_nb_n\{\log(KJ)\}^{1/\gamma}\max\limits_{(j,k)}\max\limits_{1\leq i\leq l_n}\Phi^{\tilde\ell_i}_{jk,\psi_\nu} + s^{1/2}\rho_n\max\limits_{(j,k)}\max\limits_{1\leq i\leq l_n}\mu^{\tilde\ell_i}_{jk,\psi_\nu}$,  we have
$$\P\Big(n^{-1}\max_{(j,k)}\sum^{l_n}_{i=1} \Big\{\sum^{ib_n}_{l= (i-1)b_n+1}(\psi_{jk}(Z_{j,l},\tilde\beta_{jk},\hat h_{jk})-\psi^0_{jk,l})\Big\}^2\geq w\Big)\to0,\,\text{ as }n\to\infty,$$
where $\tilde{\beta}_{jk}$ is any estimator such that $|\tilde{\beta}_{jk}- \beta^0_{jk}|\leq C \rho_n$ with probability $1-\smallO(1)$, and $\hat h_{jk}$ satisfies \hyperref[C5]{(C5)}. ($\Phi^{\ell_i}_{jk,\psi_\nu}$,$\mu^{\ell_i}_{jk,\psi_\nu}$ and $\Phi^{\tilde\ell_i}_{jk,\psi_\nu}$,$\mu^{\tilde\ell_i}_{jk,\psi_\nu}$ are defined in \eqref{norms} and \eqref{norms2})
\end{lemma}

\begin{proof}[Proof of Lemma \ref{block.lemma}]
Same as in the proof of Lemma \ref{ratedeltan}, we look at the functional class $\widetilde{\mathcal{F}} = \{z\mapsto\psi_{jk}\{z,\beta,\tilde h(x_{j(-k)})\} - \psi_{jk}\{z,\beta^0_{jk},h^0_{jk}(x_{j(-k)})\}: (j,k)\in G, \beta\in\mathcal{B}_{jk},|\beta-\beta^0_{jk}|\leq C\rho_n,\tilde h\in\mathcal{H}_{jk}\}$, which has $2F$ as the envelope (the definition of $F$ is given in \hyperref[C6]{(C6)}). For any $f,\tilde f\in\widetilde{\mathcal F}$, we have
\begin{eqnarray*}
\delta_{b_n}(f,\tilde f)&:=&\frac{1}{n}\sum^{l_n}_{i=1}\Big|\Big\{\sum^{ib_n}_{l= (i-1)b_n+1}f(z_l)\Big\}^2-\Big\{\sum^{ib_n}_{l= (i-1)b_n+1}\tilde f(z_l)\Big\}^2\Big|\\
&=&\frac{1}{n}\sum^{l_n}_{i=1}\Big|\sum^{ib_n}_{l= (i-1)b_n+1}\{f(z_l)-\tilde f(z_l)\}\sum^{ib_n}_{l= (i-1)b_n+1}\{f(z_l)+\tilde f(z_l)\}\Big|\\
&\leq&\Big[\frac{1}{n}\sum^{l_n}_{i=1}\Big\{\sum^{ib_n}_{l= (i-1)b_n+1}f(z_l)-\tilde f(z_l)\Big\}^2\Big]^{1/2}\Big[\frac{1}{n}\sum^{l_n}_{i=1}\Big\{\sum^{ib_n}_{l= (i-1)b_n+1}2F(z_l)\Big\}^2\Big]^{1/2}\\
&\leq& 2\Big[\frac{1}{n}\sum_{t=1}^n\{f(z_t)-\tilde f(z_t)\}^2\Big]^{1/2}\Big[\frac{1}{l_n}\sum^{l_n}_{i=1}\Big\{\sum^{ib_n}_{l= (i-1)b_n+1}F(z_l)\Big\}^2\Big]^{1/2}.\\
\end{eqnarray*}
In addition, it is not hard to see that $\|\sum^{ib_n}_{l = (i-1)b_n+1}F(z_l)\|_2\leq b_n\E\{F(z_t)\}\vee b_n^{1/2}\|F(z_{\cdot})\|_{2,0}$.
As a result, we can assume that
\begin{align*}
\mathcal N(\epsilon\|2F\|_{\mathcal Q, 2}, \widetilde{\mathcal F}, \delta_{b_n}(\cdot))&\lesssim_{\P} \mathcal N(\epsilon\|2F\|_{\mathcal Q, 2}/(2b_n\|F(z_{\cdot})\|_{2,0}), \widetilde{\mathcal F}, \|\cdot\|_{\mathcal Q, 2})\\
&\lesssim(2a_nb_n\|F(z_{\cdot})\|_{2,0}/\epsilon)^{cs}.
\end{align*}
To deal with the rate of $\sup_{f\in\widetilde\mF}n^{-1}\sum^{l_n}_{i=1}\{\sum^{ib_n}_{l= (i-1)b_n+1}f(z_l)\}^2$, let $g_{b_n,i}(f)\defeq\{\sum^{ib_n}_{l= (i-1)b_n+1}f(z_l)\}^2$. We first analyze $n^{-1}\sum^{l_n}_{i=1}[g_{b_n,i}(f) - \E\{g_{b_n,i}(f)\}]$. Note that for any $f\in\widetilde\mF$, there exists a set $\mA_n$ such that $\min_{f'\in\mA_n}\delta_{b_n}(f,f')\lesssim_{\P}\tilde\epsilon$, where $\tilde\epsilon\defeq\epsilon\|2F\|_{\mQ,2}$, and the cardinality of the set $|\mA_n|=(a_nb_n/\epsilon)^{cs}$. Hence, with probability $1-\smallO(1)$,
\begin{eqnarray*}
&&n^{-1}\sup_{f\in\widetilde\mF}\Big|\sum^{l_n}_{i=1}[g_{b_n,i}(f) - \E\{g_{b_n,i}(f)\}]\Big|\\
&\leq&n^{-1}\sup_{f\in\widetilde\mF}\Big|\sum^{l_n}_{i=1}[g_{b_n,i}(f) - g_{b_n,i}(\varpi(f)) -\E\{g_{b_n,i}(f)|\} + \E\{g_{b_n,i}(\varpi(f))\}]\Big|\\
&& + \,\, n^{-1}\max_{f\in\mA_n}\Big|\sum^{l_n}_{i=1}[g_{b_n,i}(f) - \E\{g_{b_n,i}(f)\}]\Big|\\
&\leq&2\tilde\epsilon + n^{-1}\max_{f\in \mA_n}\Big|\sum^{l_n}_{i=1}[g_{b_n,i}(f) - \E\{g_{b_n,i}(f)|\mF_{(i-1)b_n},\bm{X}_{j(-k),i}\}]\Big| \\
&& +\,\,  n^{-1}\max_{f\in \mA_n}\Big|\sum^{l_n}_{i=1}[\E\{g_{b_n,i}(f)|\mF_{(i-1)b_n},\bm{X}_{j(-k),i}\} - \E\{g_{b_n,i}(f)\}]\Big|\\
&=:& 2\tilde\epsilon+ I_n + II_n,
\end{eqnarray*}
where $\bm{X}_{j(-k),i}\defeq\{X_{j(-k),l}\}^{ib_n}_{l = (i-1)b_n+1}$, $\varpi(f)\defeq\arg\,\underset{f'\in\mA_n}{\min}\delta_{b_n}(f,f')$.

Next, we shall look for the bounds for $I_n$ and $II_n$, respectively. We first truncate the function $g_{b_n,i}(\cdot)$ as $g_{b_n,i}^c(\cdot)\defeq g_{b_n,i}(\cdot)\IF\{g_{b_n,i}(\cdot)\leq\sqrt{n}\delta\}$, where $\delta$ is a positive constant. Note that
\begin{eqnarray}\label{In}
I_n&\leq& n^{-1}\max_{f\in \mA_n}\Big|\sum^{l_n}_{i=1}[g_{b_n,i}^c(f) - \E\{g^c_{b_n,i}(f)|\mF_{(i-1)b_n},\bm{X}_{j(-k),i}\}]\Big| \notag\\
&&+\,\, n^{-1}\Big|\sum^{l_n}_{i=1}[G_{b_n,i} \IF(G_{b_n,i}>\sqrt{n}\delta)+ \E\{G_{b_n,i}\IF(G_{b_n,i}>\sqrt{n}\delta)|\mF_{(i-1)b_n},\bm{X}_{j(-k),i}\}]\Big|,\notag\\
\end{eqnarray}
where $G_{b_n,i}\geq\sup_{f\in\widetilde\mF}g_{b_n,i}(f)$. For the second term, applying the Markov inequality gives
\begin{eqnarray*}
&&\P\Big(\Big|\sum^{l_n}_{i=1}[G_{b_n,i} \IF(G_{b_n,i}>\sqrt{n}\delta)+ \E\{G_{b_n,i}\IF(G_{b_n,i}>\sqrt{n}\delta)|\mF_{(i-1)b_n},\bm{X}_{j(-k),i}\}]\Big|\geq x\Big)\\
&\leq& \frac{2\sum^{l_n}_{i=1}\E|G_{b_n,i}\IF(G_{b_n,i}>\sqrt{n}\delta)|}{x}.
\end{eqnarray*}
It is not hard to see that for any $c>1$, $\E|G_{b_n,i}\IF(G_{b_n,i}>\sqrt{n}\delta)|\leq\E|G_{b_n,i}|^{1+c}/(\delta\sqrt{n})^{c}\leq [b_n^{2(1+c)}\{\E(F(z_t))\}^{2(1+c)}\vee b_n^{1+c}\|F(z_{\cdot})\|_{2(1+c),0}^{2(1+c)}]/(\delta\sqrt{n})^{c}$.
Thus, we have $\frac{\sum^{l_n}_{i=1}\E|G_{b_n,i}\IF(G_{b_n,i}>\sqrt{n}\delta)|}{x}=\smallO(1)$ with $x=\sqrt{ns\log(a_nb_n)}$, provided $b_n^{2c+1}n^{(1-c)/2}\{s\log(a_nb_n)\}^{-1/2}\|F(z_{\cdot})\|_{2(1+c),0}^{2(1+c)}=\smallO(1)$. It follows that the order of the second term in \eqref{In} is given by $\sqrt{s\log(a_nb_n)}/\sqrt{n}$.

For the truncated term in \eqref{In}, note that
\begin{eqnarray*}
&&\max_{f\in \mA_n}\sum_{i=1}^{l_n}\E\big([g_{b_n,i}^c(f) - \E\{g^c_{b_n,i}(f)|\mF_{(i-1)b_n},\bm{X}_{j(-k),i}\}]^2|\mF_{(i-1)b_n},\bm{X}_{j(-k),i}\big)\\
&\leq&4b_n^3\max_{f\in \mA_n}\sum^{n}_{l =1}\E\{f(z_l)^4|\mF_{(i-1)b_n},\bm{X}_{j(-k),i}\},
\end{eqnarray*}
where we have used the Jensen inequality. Moreover, for each $f\in\mA_n$, $1\leq i\leq l_n$, there is a corresponding function $\ell_i$ such that $\ell_i(z_l,\beta,\theta)=\E\{f(z_l)^4|\mF_{(i-1)b_n},\bm{X}_{j(-k),i}\}$ for $(i-1)b_n+1\leq l\leq ib_n$, where $\theta=(\theta_1^\top,\theta_2^\top)^\top$, $\theta_1$ and $\theta_2$ correspond to the nuisance parameters $\beta_{j(-k)}$ and $\gamma_{j(-k)}$ of $\tilde h\in\mH_{jk}$, $|\beta-\beta_{jk}^0|\lesssim\rho_n$. By the mean value theorem and the continuity of the function $\ell_i$, we have
\begin{align*}
\ell_i(Z_{j,l},\beta,\theta)=&\,\partial_\beta \ell_i(Z_{j,l},\bar\beta,\theta)(\beta-\beta^0_{jk})+\sum_{m=1}^2\partial_{\theta_m} \ell_i(Z_{j,l},\beta,\bar\theta)(\theta_m-\theta_m^0),
\end{align*}
where $(\bar\beta,\bar\theta)$ is the corresponding point which joins the line segment between $(\beta,\theta)$ and $(\beta_{jk}^0,\theta^0)$, with $\theta^0_1=\beta_{j(-k)}^0$, $\theta^0_2=\gamma_{j(-k)}^0$. It follows that
\begin{eqnarray*}
&&\max_{f\in \mA_n}\sum_{l=1}^n\E\{f(z_l)^4|\mF_{(i-1)b_n},\bm{X}_{j(-k),i}\}\\
&\lesssim&\rho_n\max_{(j,k)}\max_{\bar\beta\in\mA_n^\beta}\sum_{i=1}^{l_n}\sum^{ib_n}_{l = (i-1)b_n+1}\partial_\beta \ell_i(Z_{j,l},\bar\beta,\theta) \\
&& + \,\,\rho_n\sqrt{s}\max_{(j,k)}\max_{\bar\theta\in\mA_n^\theta}\sum_{m=1}^2\sum_{i=1}^{l_n}\sum^{ib_n}_{l = (i-1)b_n+1}\partial_{\theta_m} \ell_i(Z_{j,l},\beta,\bar\theta)\\
&=:& \rho_n\max_{(j,k)}T_{1,n,jk} + \rho_n\sqrt{s}\max_{(j,k)}T_{2,n,jk},
\end{eqnarray*}
where $\mA_n^\beta$ and $\mA_n^\theta$ collect all the points of $\beta$ and $\theta$ (with fixed $(j,k)$) according to $\mA_n$, respectively. We shall establish the maximal inequalities for $T_{1,n,jk}$ and $T_{2,n,jk}$. To this end, we first define the following quantities:
\begin{align}\label{norms}
&\Phi^{\ell_i}_{jk,\psi_\nu}\defeq\big\|\max_{\bar\beta\in\mA_n^\beta}\partial_\beta\ell_i(z_\cdot,\bar\beta,\theta)\big\|_{\psi_\nu,0}\vee\max_{m\in\{1,2\}}\max_{k'\neq k}\big\|\max_{\bar\theta\in\mA_n^\theta}\partial_{\theta_{m,k'}}\ell_i(z_{\cdot},\beta,\bar\theta)\big\|_{\psi_\nu,0},\notag\\
&\mu^{\ell_i}_{jk,\psi_\nu}\defeq\E\big\{\max_{\bar\beta\in\mA_n^\beta}\partial_\beta\ell_i(z_t,\bar\beta,\theta)\big\}\vee\max_{m\in\{1,2\}}\max_{k'\neq k}\E\big\{\max_{\bar\theta\in\mA_n^\theta}\partial_{\theta_{m,k'}}\ell_i(z_{t},\beta,\bar\theta)\big\}.
\end{align}
Given $\max\limits_{(j,k)}\max\limits_{1\leq i\leq l_n}\Phi^{\ell_i}_{jk,\psi_\nu}<\infty$, $\max\limits_{(j,k)}\max\limits_{1\leq i\leq l_n}\mu^{\ell_i}_{jk,\psi_\nu}<\infty$, applying a similar approach as in the proof of Theorem \ref{residuals.theorem} gives
\begin{eqnarray*}
&&\max_{f\in \mA_n}\sum_{i=1}^{l_n}\E\big([g_{b_n,i}^c(f) - \E\{g^c_{b_n,i}(f)|\mF_{(i-1)b_n},\bm{X}_{j(-k),i}\}]^2|\mF_{(i-1)b_n},\bm{X}_{j(-k),i}\big)\\
&\lesssim& c\sqrt{ns}b_n^3\rho_n\{\log(KJ)\}^{1/\gamma}\max_{(j,k)}\max_{1\leq i\leq l_n}\Phi^{\ell_i}_{jk,\psi_\nu} + n\sqrt{s}b_n^3\rho_n\max_{(j,k)}\max_{1\leq i\leq l_n}\mu^{\ell_i}_{jk,\psi_\nu}=\bigO(n).
\end{eqnarray*}
Therefore, we can apply Lemma \ref{freedman} with $u=\sqrt{n}\delta$ and $v=\bigO(n)$ to get
$$\P\Big(\max_{f\in \mA_n}\Big|\sum^{l_n}_{i=1}[g_{b_n,i}^c(f) - \E\{g^c_{b_n,i}(f)|\mF_{(i-1)b_n},\bm{X}_{j(-k),i}\}]\Big|\geq z\Big)\leq |\mA_n|e^{-z^2/(2z\sqrt{n}\delta+2n)}.$$
By letting $z\gtrsim c\sqrt{ns\log(a_nb_n)}$, for sufficiently large $c$, we can obtain $I_n\lesssim cn^{-1/2}\sqrt{s\log(a_nb_n)}$.

For $II_n$, we just need to replicate a similar procedure to derive the bound. In particular, for each $f\in\mA_n$, $1\leq i\leq l_n$, there is a corresponding function $\tilde\ell_i$ such that $\tilde\ell_i(z_l,\beta,\theta)=\E\{f(z_l)^2|\mF_{(i-1)b_n},\bm{X}_{j(-k),i}\}$ for $(i-1)b_n+1\leq l\leq ib_n$. We also make use of the continuity of the function $\tilde\ell_i$ and assume
\begin{align}\label{norms2}
&\Phi^{\tilde\ell_i}_{jk,\psi_\nu}\defeq\big\|\max_{\bar\beta\in\mA_n^\beta}\partial_\beta\tilde\ell_i(z_\cdot,\bar\beta,\theta)\big\|_{\psi_\nu,0}\vee\max_{m\in\{1,2\}}\max_{k'\neq k}\big\|\max_{\bar\theta\in\mA_n^\theta}\partial_{\theta_{m,k'}}\tilde\ell_i(z_{\cdot},\beta,\bar\theta)\big\|_{\psi_\nu,0}<\infty,\notag\\
&\mu^{\tilde\ell_i}_{jk,\psi_\nu}\defeq\E\big\{\max_{\bar\beta\in\mA_n^\beta}\partial_\beta\tilde\ell_i(z_t,\bar\beta,\theta)\big\}\vee\max_{m\in\{1,2\}}\max_{k'\neq k}\E\big\{\max_{\bar\theta\in\mA_n^\theta}\partial_{\theta_{m,k'}}\tilde\ell_i(z_{t},\beta,\bar\theta)\big\}<\infty.
\end{align}
Then, we obtain that
$$II_n\lesssim_{\P}cn^{-1/2}s^{1/2}\rho_nb_n\{\log(KJ)\}^{1/\gamma}\max_{(j,k)}\max_{1\leq i\leq l_n}\Phi^{\tilde\ell_i}_{jk,\psi_\nu}.$$

At last, combining the two parts and adding the mean back, we see that
$$n^{-1}\sup_{f\in\widetilde\mF}\Big|\sum^{l_n}_{i=1}g_{b_n,i}(f)\Big|\lesssim_{\P}cn^{-1/2}s^{1/2}\rho_nb_n\{\log(KJ)\}^{1/\gamma}\max_{(j,k)}\max_{1\leq i\leq l_n}\Phi^{\tilde\ell_i}_{jk,\psi_\nu} + s^{1/2}\rho_n\max_{(j,k)}\max_{1\leq i\leq l_n}\mu^{\tilde\ell_i}_{jk,\psi_\nu},$$
given $\tilde\epsilon$ is sufficiently small ($\tilde\epsilon\leq \frac{1}{2} n^{-1/2}\sqrt{s\log(a_nb_n)}$). This concludes the proof.
\end{proof}

\setcounter{subsection}{0}
\renewcommand{\thesubsection}{C.\arabic{subsection}}
\setcounter{equation}{0}
\renewcommand{\theequation}{C.\arabic{equation}}
\setcounter{theorem}{0}
\renewcommand{\thetheorem}{C.\arabic{theorem}}
\setcounter{lemma}{0}
\renewcommand{\thelemma}{C.\arabic{lemma}}
\setcounter{figure}{0}
\renewcommand{\thefigure}{C.\arabic{figure}}
\setcounter{table}{0}
\renewcommand{\thetable}{C.\arabic{table}}
\setcounter{remark}{0}
\renewcommand{\theremark}{C.\arabic{remark}}
\setcounter{example}{0}
\renewcommand{\theexample}{\arabic{example}}

\section{Supplementary Examples}
\subsection{Practical Examples of SRE} \label{pae}

\begin{example}[Identification Test for Large Structural Vector Autoregression Models]
	Denote $U_t = (U_{1,t}, U_{2,t}, \ldots, U_{M,t})^\top$.
	A large structural VAR can be represented in the following form (without loss of generality, consider only lag one):
	\begin{equation*}
	\mathbf{A}U_t = \mathbf{B}U_{t-1} + \vps_t,
	\end{equation*}
	where $\mathbf{A}$(invertible) and $\mathbf{B}$ are $M\times M$ matrices. The structural shocks $\vps_t$ satisfy $\E(\vps_t)=0$ and $\Var(\vps_t)= \mathbf{I}_M$.
	The corresponding reduced form is given by
	\begin{equation}
	U_t =\mathbf{D}U_{t-1} + \nu_t,
	\end{equation}
	with $\mathbf{D}=\mathbf{A}^{-1}\mathbf{B}$ and $\nu_t = \mathbf{A}^{-1}\vps_t$, where $\nu_t$ is denoted as the reduced form VAR shocks. Suppose $\nu_t$ spans the space of $\vps_t$.
	The crucial question is the identification of $\mathbf{A}$.
	Typically, the covariance matrix of the reduced form shock $\nu_t$ is estimated with $M(M+1)/2$ restrictions, which are smaller than the $M^2$ restrictions needed to pin down $\vps_t$.
	Adopting the identification approach proposed by \cite{stock2012disentangling}, we may use external instruments that are correlated with the shock of interest and are uncorrelated with other shocks.
	Without loss of generality, suppose the structural shock of interest is $\vps_{j,t}$. Then we can define $z_{j,t}$ as an external instrument for the $j$th structural shock satisfying
	\begin{eqnarray}
	\E(\vps_{j,t}z_{j,t})&\neq& 0,\notag\\
	\E(\vps_{j',t}z_{j,t}) &=& 0, \quad \text{for }j'\neq j.\notag
	\end{eqnarray}
	Thus, we propose to regress $z_{j,t}$ on $\nu_{t}$:
	$$
	z_{j,t} = \nu_{t}^\top\delta_{j} + e_{j,t}.
	$$
	In practice, $\nu_{t}$ are replaced by the residuals obtained from a large VAR reduced form regression as in Example \ref{examp4} {(in the main article)}.  The estimator of $\delta_{j}$ is denoted as $\hat{\delta}_j$. It can be obtained by LASSO estimation, which give us a sparse estimator of the $j$th row of the matrix $\mathbf{A}^{-1}$ up to a scaling factor.	Repeating this step for any $j$, one may formulate estimators for each row and perform simultaneous inference/hypothesis testing on the structural matrix $\mathbf{A}^{-1}$.
	
	In summary, this is also a special case of SRE with
	$$
	(Y_{j,t}, X_{j,t}, \vps_{j,t}, \beta_{j}^0) = (U_{j, t}, U_{-j,t-1}, \nu_{t}, \mathbf{D}_{j\cdot}^\top),\quad j = 1, \ldots, M,
	$$
	$$
	(Y_{j,t}, X_{j,t}, \vps_{j,t}, \beta_{j}^0) = (z_{(j-M),t}, \nu_{t}, e_{(j-M),t}, \delta_{(j-M)}), \quad j =M+1, \ldots, 2M.
	$$
\end{example}

\begin{example}[Cross-sectional Asset Pricing]
	Denote $Y_{j,t}$ as the excess return for asset $j$ and period $t$.
	Asset pricing models explain the cross sectional variation in expected returns across assets; see e.g. \cite{cochrane2009asset}. In particular, the variation
	of expected cross sectional returns is explained by the exposure to $K$ factors $X_{jk,t},$ $k = 1,\ldots, K$. One commonly used way to estimate an asset pricing model is to
	run a system of regression equations:
	\begin{equation}
	Y_{j,t} = \beta_{j0}+ \sum_{k=1}^{K} \beta_{jk} X_{jk,t} +\vps_{j,t}, 
	\end{equation}
	where $X_{jk,t}$'s are the factor returns (assumed to be excess returns of zero-cost portfolios).
	The selection of factors is a critical issue and the SRE framework addresses this issue, in particular when the number of factors $K$ is large. {See \cite{feng2017taming} for a detailed model-selection
		exercise on picking asset pricing factors}.
	The factor premiums are $\E(X_{jk,t})$ and the pricing errors are $\beta_{j0}$.
	Usually, asset pricing imposes the restriction that all $\beta_{j0}$'s are zero.
	Our simultaneous inference framework naturally serves the purpose of simultaneously testing the zero pricing errors in a cross sectional regression setup.
	Namely, we are interested in testing $H_0: \beta_{j0} = 0, \forall j= 1, \ldots,J$ versus $H_A:\exists\,j$ such that $\beta_{j0} \neq 0 $.
	Our test procedure in Section \ref{scr} can be directly applied to achieve this goal.
\end{example}

\begin{example}[Network Formation and Spillover Effects]
	There is an emerging literature in economics concerning quantifying spillover effects and network formation.
	One leading example is as in \cite{manresa2013estimating}, which attempts to quantify social returns to research and development (R\&D).
	Here, $U_{j,t}$ is taken to be the log output for firm $j$ and time $t$. This output is loading on
	$D_{j,t}$ (capital stock for firm $j$ and period $t$), and the aggregated spill-overs from the capital stock of other firms $\sum_{i\neq j}w_{ij} D_{i,t}$.
	The regression equation also controls for other covariates $X_{j,t}$ (e.g., log labor, log capital etc.):
	\begin{equation}
	U_{j,t} = \beta_{j}D_{j,t} + \sum_{i\neq j}\omega_{ij}D_{i,t}+ \gamma_j^\top X_{j,t}+\vps_{j,t},
	\end{equation}
	where $\omega_{ij}$ is referred to as the spillover effects of the R\&D development of firm $i$ on firm $j$. This again is contained in the SRE with
	$$
	(Y_{j,t}, X_{j,t}, \vps_{j,t}, \beta_{j}^0) = (U_{j, t}, (D_{j,t},D_{-j,t}^\top,X_{j,t}^\top)^\top, \vps_{j,t}, (\beta_j,\omega_{(-j)j}^\top,\gamma_j^\top)^\top),\quad j = 1, \ldots, J.
	$$
	Our simultaneous inference procedure (Section \ref{scr}) can be applied to check the significance of the spillover effects for any set
	of parameters of interest.
	As an analogy, the presented framework displays a general class of network models, where $U_{j,t}$ is taken to be the nodal response, and $D_{i,t}$ are the nodal covariates. Global or local inference on the network parameters $\omega_{ij}$ is the subject of research. Section \ref{app} is devoted to inference on the spillover effects of a textual sentiment index.
\end{example}

\begin{remark}
	Suppose there is unobserved heterogeneity in $U_{j,t}$, e.g. $U_{j,t} = \alpha_j + \sum_{i\neq j} w_{ij}D_{i,t}$ $ +\vps_{j,t}$, where $w_{ij}$ characterizes the spillover of individual $i$ on $j$, and $\alpha_j$ is the individual fixed effect. For this situation consider the demeaned version to eliminate the individual specific effects and work with the new model: $\widetilde{U}_{j,t} = \sum_{i\neq j} w_{ij} \widetilde{D}_{i,t} + \widetilde{\vps}_{j,t},$ where $\widetilde{U}_{j,t} = U_{j,t} - \En U_{j,t}$, $\widetilde{D}_{i,t} = D_{i,t} - \En D_{i,t}$, $\widetilde{\vps}_{j,t} = \vps_{j,t} - \En\vps_{j,t}$, under the condition that $U_{j,t}$ has no feedback effects on $D_{i,t}$ (for example, $D_{i,t}$ should not be the lagged variable of $U_{j,t}$).
\end{remark}

\subsection{Examples of the Dependence Measure}\label{example}
\begin{itemize}
	\item[1.] \textbf{AR(1)}: $Y_t$ follows $Y_t = a Y_{t-1}+ \vps_t$, with $|a| < 1$, $\vps_t\sim \mbox{i.i.d.}(0,\sigma^2)$.
	Therefore, the MA representation is given by $Y_t = \sum^\infty_{l=0} a^{l} \vps_{t-l}$ and $Y_t^\ast = \sum^\infty_{l=0}a^{l} \vps_{t-l} + a^t\vps^\ast_0 - a^t\vps_0$. $\|Y_t - Y_t^\ast\|_q=|a|^t\|\vps_0-\vps^\ast_0\|_q$, $\Delta_{m,q} \lesssim|a|^m$, $\|Y_\cdot\|_{q,\varsigma} \lesssim \sup_{m\geq0} (m+1)^{\varsigma}|a|^m <\infty$.
	\item[2.] \textbf{ARCH(1)}: An ARCH (Autoregressive conditionally heteroscedastic) model is given by $Z_t = \sigma_t \vps_t$, $\sigma^2_t = w + \alpha^2 Z^2_{t-1}$, with $w>0$, $\vps_t$ are i.i.d. shocks and $\Var(Z_t)=\sigma^2<\infty$. Thus, it is not hard to see that $Z_t^2= w\sum^{\infty}_{l=0} \alpha^{2l}\prod^{l}_{k=0} \vps^2_{t-k}$. Rewrite the model as $Z_t = R(Z_{t-1}, \vps_t) = \sqrt{(w+ \alpha^2 Z^2_{t-1})}\vps_t$. According to \cite{wu2004limit}, we have the Lipschitz constant involved in the Lyapunov type condition ensuring the forward iteration contraction $\sup_{x\neq x'}\frac{|R(x,\vps_0)-R(x',\vps_0)|}{|x-x'|}\leq |\alpha\vps_0|$. Let $\mu\defeq\E|\alpha\vps_0|<1$ and assume $|\alpha\vps_0|+ |R(t_0,\vps_0)|$ has finite $q$th moment. Then the process $Z_t$ has stationary solutions. Moreover, $\|Z_t- Z_t^\ast\|_q\leq |\mu|^t\|\vps_0-\vps_0^\ast\|_q$, and thus $\Delta_{m,q} \lesssim |\mu|^m$. Given $|\mu|<1$, then we have
	$\|Z_\cdot\|_{q,\varsigma} \lesssim \sup_{m\geq0} (m+1)^{\varsigma}|\mu|^m <\infty$.
	\item[3.]\textbf{TAR} (Threshold autoregressive model):
	$Y_t = \theta_1 Y_{t-1}\IF\{Y_{t-1} < \tau\}+ \theta_2Y_{t-1}\IF\{Y_{t-1} \geq \tau\} +\vps_t,$ where $\theta_1$ and $\theta_2$ are two parameters and $\vps_t$ are i.i.d. shocks. If $\theta \defeq \mbox{max}\{|\theta_1|, |\theta_2|\}< 1$ and $\vps_t$ has a finite $\alpha$-th order moment, then the TAR model admits a stationary solution with $\|Y_\cdot\|_{q,\varsigma} \lesssim \sup_{m\geq0} (m+1)^{\varsigma}\theta^m <\infty$.
	\item[4.] \textbf{VAR} (Vector autoregressive model):
	Without loss of generality we focus on VAR(1) given by $Y_t = A Y_{t-1} + \vps_t$, where $Y_t,\vps_t \in {\R}^J$, and $\vps_t\sim\mbox{i.i.d.}\N(0, \Sigma)$. If the spectral radius of $A^\top A$, $\rho(A^\top A) < 1$, then $\underset{m\to \infty}{\lim} \|A\|^{m} \to0$, where $\|\cdot\|$ denotes the spectral norm of a matrix. Rewrite the model as $Y_t = \sum^{\infty}_{l=0} A^{l}\vps_{t-l}$. The existence of a stationary solution can be checked by Kolmogorov's three series theorem. For each equation $j$, $Y_{j,t}- Y^\ast_{j,t} = [A^{t}]_j (\vps_0- \vps_0^\ast)$, where $[A^{t}]_j $ is the $j$th row of the matrix $A^t$. $(\E(|Y_{j,t}-Y^\ast_{j,t}|^q))^{1/q} \leq |[A^{t}]_j|_1 \||\vps_0-\vps_0^\ast|_\infty\|_{q}$. 
	{It follows that $(\E(|Y_{j,t}-Y^\ast_{j,t}|^q))^{1/q} \leq 2|[A^{t}]_j|_1\mu_q$, where $\mu_q\defeq\underset{1\leq j\leq J}{\max}\|\vps_{j,0}\|_q$. Suppose $\underset{1\leq j\leq J}{\max}|[A^{t}]_j|_1 \leq |\alpha|^{t}$ ($|\alpha|<1$). Then we have $\underset{1\leq j\leq J}{\max}\|Y_{j,\cdot}\|_{q,\varsigma} \lesssim \mu_q$, $(\sum_{j=1}^J\|Y_{j,\cdot}\|^q_{q,\varsigma})^{1/q} \lesssim J^{1/q}\mu_q$, and $\||Y_{j,\cdot}|_\infty \|_{q,\varsigma} \lesssim (J)^{1/q}$ by union bounds.}
	\item[5.] \textbf{High-dimensional ARCH}:
	Consider $Y_t\in {\R}^{J}$, a high-dimensional ARCH(1) model follows for example the general specification from \cite{bollerslev1988capital} and \cite{hansen1998stationarity}: $Z_t = H_t^{1/2} \vps_t$, and $\E(Z_tZ_t^{\top}|\mathcal{F}_{t-1}) = H_t$, with $\vps_t \sim \mbox{i.i.d.}\N(0, \mathbf{I}_J)$. The specification of the conditional covariance matrix
	$H_t= \Omega+ A Z_{t-1}Z^{\top}_{t-1}A^{\top},$ where $\Omega$ is positive definite and $A$ is a $J\times J$ matrix. Studying the stationarity condition of the process is not trivial.
	Define $h_t\defeq\vech(H_t)$, the selection matrix $D_J$ ($J^2\times J(J+1)/2$) gives $\vec (H_t) = D_Jh_t$ and its generalized inverse matrix $D_J^{+}$ such that $D_J^{+}D_J = \mathbf{I}_{J(J+1)/2}$. The $\vech$ notation of the iterations follows $h_t = \vech(\Omega)+ D_J^{+}(A\otimes A)D_J \vech(Y_{t-1}Y_{t-1}^{\top})$. Define $\tilde{A}\defeq D_J^{+}(A\otimes A)D_J$, 
	$w\defeq\vech(\Omega)$. For simplicity, we look at the process $h_t$, with the state space representation $h_t = w + G(h_{t-1},\vps_{t-1}) = F(h_{t-1},\vps_{t-1}) = w + \tilde{A}\vech(\{\vech^{-1}(h_{t-1})\}^{1/2}\vps_{t-1}\vps_{t-1}^{\top}\{\vech^{-1}(h_{t-1})\}^{-1/2})$. The partial derivative matrix is $\Delta_t=\Delta(h_t,\vps_t)=\partial h_{t+1}/\partial h^{\top}_t =\tilde{A} D_J^{+}(H_t^{1/2}\vps_t\vps_t^{\top}H_t^{-1/2} \otimes I_J)D_J$, and $\E \Delta_t = \tilde{A}$. Therefore, the spectral radius of $AA^\top$, $\rho(AA^{\top})< 1$ ensures a stationary solution to the process $h_t$.
	Moreover, by solving the state space iteration recursively, we have $\E|h_{t} - h^\ast_{t}|_1\leq 2 \E |\mathcal{P}_0(h_t)|_1\leq |\tilde{A}^t\{\vech(\Sigma) +w\}+\tilde{A}^{t+1}\vech(\Sigma) |_1\lesssim \{\operatorname{tr}(AA^{\top})\}^t$, where the projector operator $\mP_l(h_t) \defeq \E(h_t|\mathcal{F}_l) -\E(h_t|\mathcal{F}_{l-1})$ and $\Sigma = \E H_t =\sum^{\infty}_{i=0} A^i\Omega(A^{i})^\top$. {Assume that $\{\operatorname{tr}(AA^{\top})\}^t <|c|^t$, with $|c|<1$, we have $\sum_{j=1}^{J(J+1)/2} \|h_{j,\cdot}\|_{1,\varsigma}\lesssim J(J+1)/2$.}
	
	According to \cite{hafner2009asymptotic}, the iteration formulae are given by $h_t = \varpi(\bar h^\star_{t-1}, \vps_{t-1})+ \sum^{m-1}_{l = 1} \Pi^l_{k=1} \Delta(\bar h^\star_{t-k}, \vps_{t-k}) \varpi(\bar h^\star_{t-l-1}, \vps_{t-l-1})+ \Pi^m_{k=1}\Delta(\bar h^\star_{t-k}, \vps_{t-k})h_{t-m}$,
	where $\varpi(h, \vps) = w+ G(h^\star, \vps)- \Delta(h,\vps)h^\star$, $h^\star$ is the contraction state, and $\bar h^\star_{t-k}$'s lie on the line segment between $h^\star$ and $h_{t-k}$. For ease of derivation, we assume a strong assumption such that
	$\E \sup_{h_m} \|\Delta(h_m, \vps_{m})\|^q<s<1$ for all $m\geq1$ and $q\geq2$, where $\|\cdot\|$ denotes the spectral norm of a matrix. Let $h^m=\{(h_1^\top,\ldots,h_m^\top)^\top:|h_t|_2=1,t=1,\ldots,m\}$, it follows $\E \sup_{h^m}\|\Pi^m_{k=1}\Delta(h_{m-k+1}, \vps_{m-k+1})\|^q \leq \Pi^m_{k=1} \E\sup_{h_{m-k+1}}\|\Delta(h_{m-k+1}, \vps_{m-k+1})\|^q\leq s^m$. Hence, $\underset{1\leq j\leq J(J+1)/2}{\max}\|h_{j,\cdot}\|_{q,\varsigma} \leq C$, {$\||h_{\cdot}|_\infty\|_{q,\varsigma} \lesssim \||h_t|_\infty\|_q\lesssim \{J(J+1)/2\}^{1/q}$, 
		and $(\sum_{j=1}^{J(J+1)/2}\|h_{j,\cdot}\|^q_{q,\varsigma})^{1/q} \lesssim \{J(J+1)/2\}^{1/q}$.}
\end{itemize}

\begin{sidewaystable}
	\setcounter{subsection}{0}
	\renewcommand{\thesubsection}{D.\arabic{subsection}}
	\setcounter{equation}{0}
	\renewcommand{\theequation}{D.\arabic{equation}}
	\setcounter{theorem}{0}
	\renewcommand{\thetheorem}{D.\arabic{theorem}}
	\setcounter{lemma}{0}
	\renewcommand{\thelemma}{D.\arabic{lemma}}
	\setcounter{figure}{0}
	\renewcommand{\thefigure}{D.\arabic{figure}}
	\setcounter{table}{0}
	\renewcommand{\thetable}{D.\arabic{table}}
	\vskip 2em
	
	\section{Additional Details for Empirical Analysis}\label{table}
	
	
	\small
	\begin{center}
		\begin{tabular}[H]{ cc | cc | cc }
			\hline\hline
			\multicolumn{2}{c|}{\textbf{Consumer Discretionary (11)}} & \multicolumn{2}{c|}{\textbf{Financials (8)}} & GD & General Dynamics Corporation\\
			\cline{1-4}
			AMZN&Amazon.com, Inc.&AIG&American International Group, Inc.&GE& General Electric Company\\
			BBY&Best Buy Co. Inc.&AMT&American Tower Corporation (REIT)& HON &Honeywell International Inc.\\
			CBS&CBS Corporation&AXP&American Express Company& LMT&Lockheed Martin Corporation\\
			CMCSA&Comcast Corporation&BAC&Bank of America Corporation& LUV &Southwest Airlines Company\\
			\cline{5-6}
			CMG&Chipotle Mexican Grill, Inc.&C&Citigroup Inc.& \multicolumn{2}{c}{\textbf{Information Technology (11)}} \\
			\cline{5-6}
			DIS&Walt Disney Company (The)&ETFC&E*TRADE Financial Corporation& AAPL &Apple Inc.\\
			F&Ford Motor Company&GS&Genpact Limited& ACN &Accenture plc\\
			GM&General Motors Company& JPM &J P Morgan Chase \& Co& ADP & Automatic Data Processing, Inc.\\
			\cline{3-4}
			GPS&Gap, Inc. (The)&\multicolumn{2}{c|}{\textbf{Health Care (8)}}& CSCO & Cisco Systems, Inc.\\
			\cline{3-4}
			HD&Home Depot, Inc. (The)&AET&Aetna Inc.& EA &Electronic Arts Inc.\\
			LEN&Lennar Corporation&AMGN&Amgen Inc.& EBAY & eBay Inc.\\
			\cline{1-2}
			\multicolumn{2}{c|}{ \textbf{Consumer Staples (4)} }& BIIB &Biogen Inc.& EMC & EMC Corporation \\
			\cline{1-2}
			COST&Costco Wholesale Corporation&BMY&Bristol-Myers Squibb Company&FSLR &First Solar, Inc.\\
			CVS&CVS Health Corporation&CELG&Celgene Corporation & HPQ &HP Inc. \\
			KO&Coca-Cola Company (The)&GILD&Gilead Sciences, Inc. &IBM &International Business Machines Corporation \\
			KR&Kroger Company (The)&JNJ& Johnson\& Johnson & INTC & Intel Corporation \\
			\cline{1-2}\cline{5-6}
			\multicolumn{2}{c|}{ \textbf{Energy (6)} }& LLY & Eli Lilly and Company &\multicolumn{2}{c}{\textbf{Materials (3)}}  \\
			\cline{1-6}
			APC&Anadarko Petroleum Corporation&\multicolumn{2}{c|}{\textbf{Industrials (10)}} & AA & Alcoa Corporation\\
			\cline{3-4}
			BHI&Black Hills Corp.&BA&Boeing Company (The) & DD &EI du Pont de Nemours \& Co\\
			CHK&Chesapeake Energy Corporation&CAT&Caterpillar, Inc. &DOW &Dow Chemical \\
			\cline{5-6}
			COP&ConocoPhillips&DAL&Delta Air Lines, Inc. &\multicolumn{2}{c}{\textbf{Utilities (2)}} \\
			\cline{5-6}
			CVX&Chevron Corporation&DHR&Danaher Corporation & DUK & Duke Energy Corp.\\
			HAL&Halliburton Company&FDX&FedEx Corporation & EXC & Exelon Corporation \\
			\hline\hline
		\end{tabular}
		\caption{The list of the stock symbols and the corresponding company names grouped by industries.}  \label{table:List_of_Firms}
	\end{center}
\end{sidewaystable}

	
\end{appendices}

\newpage
\bibliography{ref}

\end{document}